\documentclass[superscriptaddress, notitlepage]{revtex4-2}
\usepackage[margin=1in]{geometry}
\usepackage{caption}
\usepackage{subcaption}
\usepackage{amsmath, amsfonts, amssymb, amsthm, mathtools, bbm, mathrsfs}
\usepackage{float}
\usepackage{etoolbox}
\usepackage[dvipsnames]{xcolor}
\usepackage{tikz}
\usetikzlibrary{decorations.pathreplacing, fit, cd}
\usepackage{diagbox}
\usepackage{makecell}
\usepackage{enumitem}
\usepackage[linesnumbered, ruled, vlined]{algorithm2e}
\SetKwInput{KwInput}{Input}                
\SetKwInput{KwOutput}{Output}              
\usepackage{thmtools}
\usepackage{thm-restate}
\usepackage{hyperref}
\usepackage{cleveref}
\usepackage[export]{adjustbox}
\allowdisplaybreaks

\definecolor{TangoAluminium2}{HTML}{d3d7cf}
\definecolor{TangoSkyBlue2}{HTML}{3465a4}

\pdfpageattr {/Group << /S /Transparency /I true /CS /DeviceRGB>>}

\makeatletter
\patchcmd{\@thm}
  {\trivlist}
  {\list{}{\leftmargin=\thm@leftmargin\rightmargin=\thm@rightmargin}}
  {}{}
\patchcmd{\@endtheorem}
  {\endtrivlist}
  {\endlist}
  {}{}
\newlength{\thm@leftmargin}
\newlength{\thm@rightmargin}
\newcommand{\xnewtheorem}[3]{%
  \newenvironment{#3}
    {\thm@leftmargin=#1\relax\thm@rightmargin=#2\relax\begin{#3INNER}}
    {\end{#3INNER}}%
  \newtheorem{#3INNER}%
}
\newcommand{\snewtheorem}[3]{%
  \newenvironment{#3}
    {\thm@leftmargin=#1\relax\thm@rightmargin=#2\relax\begin{#3INNER}}
    {\end{#3INNER}}%
  \newtheorem*{#3INNER}%
}
\makeatother

\newtheoremstyle{thmstyle}		
  {3pt}						
  {3pt}						
  {\itshape} 					
  {}							
  {\bfseries}					
  {}							
  {\newline}					
  {}							

\newtheoremstyle{defstyle}		
  {3pt}						
  {3pt}						
  {\itshape} 					
  {}							
  {\bfseries}					
  {}							
  {\newline}					
  {}							

\newtheoremstyle{proofstyle}		
  {3pt}						
  {3pt}						
  {\itshape} 					
  {}							
  {\itshape}					
  {}							
  {0.25em}						
  {}							

\theoremstyle{thmstyle}
\xnewtheorem{2em}{0em}{theorem}{Theorem}[section]
\xnewtheorem{2em}{0em}{lemma}[theoremINNER]{Lemma}
\xnewtheorem{2em}{0em}{corollary}[theoremINNER]{Corollary}
\xnewtheorem{2em}{0em}{proposition}[theoremINNER]{Proposition}
\theoremstyle{defstyle}
\xnewtheorem{2em}{0em}{definition}[theoremINNER]{Definition}
\xnewtheorem{2em}{0em}{fact}[theoremINNER]{Fact}
\theoremstyle{proofstyle}
\snewtheorem{2em}{0em}{xproof}{Proof}

\newcounter{example}
\newenvironment{example}[1][]{\refstepcounter{example}\par\medskip \noindent \textbf{Example~\theexample. #1} \rmfamily}{\medskip}

\def\F{\mathbb{F}}
\def\<{\langle}
\def\>{\rangle}
\def\di{\displaystyle}
\def\wt{\mathrm{wt}}
\def\dim{\mathrm{dim} \,}
\def\im{\mathrm{im}\,}
\def\Pr{\mathrm{Pr}\,}
\DeclareMathOperator*{\argmax}{argmax}
\DeclareMathOperator*{\argmin}{argmin}
\let\emptyset\varnothing

\newlength{\stretchlen}
\def\splitterm{\_}
\newcommand{\stretchit}[1]{\leavevmode\realstretch#1\_}
\def\realstretch#1{%
    \def\temp{#1}%
    \ifx\temp\splitterm
    \else
    \hbox to \stretchlen{\hss#1\hss}\expandafter\realstretch
\fi}

\setlist[enumerate]{
  label={\itshape(\roman*)},
  widest=iii,
  labelsep=8pt,
  labelindent=0.5\parindent,
  itemindent=0pt,
  leftmargin=*,
  before=\setlength{\listparindent}{-\leftmargin},
}

\def\Nzz{\tikz[baseline={(current bounding box.center)}]{
	\node[circle, fill, inner sep=1.5pt, align=center] (a) {};
	\node[circle, fill, inner sep=1.5pt, align=center, right=20pt] (b) at (a) {};
	\draw[-, thick] (a) -- (b);
}}

\def\Noz{\tikz[baseline={(current bounding box.center)}]{
	\node[circle, fill, inner sep=1.5pt, align=center] (a) {};
	\node[circle, fill, inner sep=1.5pt, align=center, above=10pt, left=20pt] (b) at (a) {};
	\node[circle, fill, inner sep=1.5pt, align=center, below=10pt, left=20pt] (c) at (a) {};
	\draw[-, thick] (b) -- (a);
	\draw[-, thick] (c) -- (a);
}}

\def\Ntz{\tikz[baseline={(current bounding box.center)}]{
	\node[circle, fill, inner sep=1.5pt, align=center] (a) {};
	\node[circle, fill, inner sep=1.5pt, align=center, above=15pt, left=20pt] (b) at (a) {};
	\node[circle, fill, inner sep=1.5pt, align=center, above=5pt, left=20pt] (c) at (a) {};
	\node[circle, fill, inner sep=1.5pt, align=center, below=5pt, left=20pt] (d) at (a) {};
	\node[circle, fill, inner sep=1.5pt, align=center, below=15pt, left=20pt] (e) at (a) {};
	\draw[-, thick] (b) -- (a);
	\draw[-, thick] (c) -- (a);
	\draw[-, thick] (d) -- (a);
	\draw[-, thick] (e) -- (a);
}}

\def\Nzo{\tikz[baseline={(current bounding box.center)}]{
	\node[circle, fill, inner sep=1.5pt, align=center] (a) {};
	\node[circle, fill, inner sep=1.5pt, align=center, above=10pt, right=20pt] (b) at (a) {};
	\node[circle, fill, inner sep=1.5pt, align=center, below=10pt, right=20pt] (c) at (a) {};
	\draw[-, thick] (a) -- (b);
	\draw[-, thick] (a) -- (c);
}}

\def\Nooa{\tikz[baseline={(current bounding box.center)}]{
	\node[circle, fill, inner sep=1.5pt, align=center] (a) {};
        	\node[circle, fill, inner sep=1.5pt, align=center, right=20pt] (b) at (a) {};
	\draw[-, thick] (a) edge[bend left] (b);
        	\draw[-, thick] (a) edge[bend right] (b);
}}

\def\Noob{\tikz[baseline={(current bounding box.center)}]{
	\node[circle, fill, inner sep=1.5pt, align=center] (a) {};
        	\node[circle, fill, inner sep=1.5pt, align=center, above=10pt] (b) at (a) {};
        	\node[circle, fill, inner sep=1.5pt, align=center, right=20pt] (c) at (a) {};
        	\node[circle, fill, inner sep=1.5pt, align=center, right=20pt] (d) at (b) {};
        	\draw[-, thick] (a) -- (c);
        	\draw[-, thick] (a) -- (d);
        	\draw[-, thick] (b) -- (c);
        	\draw[-, thick] (b) -- (d);
}}

\def\Ntoa{\tikz[baseline={(current bounding box.center)}]{
	\node[circle, fill, inner sep=1.5pt, align=center] (a) {};
	\node[circle, fill, inner sep=1.5pt, align=center, above=10pt, right=20pt] (b) at (a) {};
	\node[circle, fill, inner sep=1.5pt, align=center, below=10pt, right=20pt] (c) at (a) {};
	\draw[-, thick] (a) edge[bend left] (b);
	\draw[-, thick] (a) edge[bend right] (b);
	\draw[-, thick] (a) edge[bend left] (c);
	\draw[-, thick] (a) edge[bend right] (c);
}}

\def\Ntob{\tikz[baseline={(current bounding box.center)}]{
	\node[circle, fill, inner sep=1.5pt, align=center] (a) {};
	\node[circle, fill, inner sep=1.5pt, align=center, below=10pt] (b) at (a) {};
	\node[circle, fill, inner sep=1.5pt, align=center, below=10pt] (c) at (b) {};
	\node[circle, fill, inner sep=1.5pt, align=center, below=10pt] (d) at (c) {};
				
	\node[circle, fill, inner sep=1.5pt, align=center, left=20pt] (e) at (b) {};
	\node[circle, fill, inner sep=1.5pt, align=center, left=20pt] (f) at (c) {};
	\draw[-, thick] (e) -- (a);
	\draw[-, thick] (e) -- (b);
	\draw[-, thick] (e) -- (c);
	\draw[-, thick] (e) -- (d);
	\draw[-, thick] (f) -- (a);
	\draw[-, thick] (f) -- (b);
	\draw[-, thick] (f) -- (c);
	\draw[-, thick] (f) -- (d);
}}

\def\Nzt{\tikz[baseline={(current bounding box.center)}]{
	\node[circle, fill, inner sep=1.5pt, align=center] (a) {};
	\node[circle, fill, inner sep=1.5pt, align=center, above=15pt, right=20pt] (b) at (a) {};
	\node[circle, fill, inner sep=1.5pt, align=center, above=5pt, right=20pt] (c) at (a) {};
	\node[circle, fill, inner sep=1.5pt, align=center, below=5pt, right=20pt] (d) at (a) {};
	\node[circle, fill, inner sep=1.5pt, align=center, below=15pt, right=20pt] (e) at (a) {};
	\draw[-, thick] (a) -- (b);
	\draw[-, thick] (a) -- (c);
	\draw[-, thick] (a) -- (d);
	\draw[-, thick] (a) -- (e);
}}

\def\Nota{\tikz[baseline={(current bounding box.center)}]{
	\node[circle, fill, inner sep=1.5pt, align=center] (a) {};
	\node[circle, fill, inner sep=1.5pt, align=center, above=10pt, left=20pt] (b) at (a) {};
	\node[circle, fill, inner sep=1.5pt, align=center, below=10pt, left=20pt] (c) at (a) {};
	\draw[-, thick] (b) edge[bend left] (a);
	\draw[-, thick] (b) edge[bend right] (a);
	\draw[-, thick] (c) edge[bend left] (a);
	\draw[-, thick] (c) edge[bend right] (a);
}}

\def\Notb{\tikz[baseline={(current bounding box.center)}]{
	\node[circle, fill, inner sep=1.5pt, align=center] (a) {};
	\node[circle, fill, inner sep=1.5pt, align=center, below=10pt] (b) at (a) {};
	\node[circle, fill, inner sep=1.5pt, align=center, below=10pt] (c) at (b) {};
	\node[circle, fill, inner sep=1.5pt, align=center, below=10pt] (d) at (c) {};
				
	\node[circle, fill, inner sep=1.5pt, align=center, right=20pt] (e) at (b) {};
	\node[circle, fill, inner sep=1.5pt, align=center, right=20pt] (f) at (c) {};
	\draw[-, thick] (a) -- (e);
	\draw[-, thick] (a) -- (f);
	\draw[-, thick] (b) -- (e);
	\draw[-, thick] (b) -- (f);
	\draw[-, thick] (c) -- (e);
	\draw[-, thick] (c) -- (f);
	\draw[-, thick] (d) -- (e);
	\draw[-, thick] (d) -- (f);
}}

\def\Ntta{\tikz[baseline={(current bounding box.center)}]{
	\node[circle, fill, inner sep=1.5pt, align=center] (a) {};
        	\node[circle, fill, inner sep=1.5pt, align=center, right=20pt] (b) at (a) {};
	\draw[-, thick] (a) edge[bend left=70] (b);
	\draw[-, thick] (a) edge[bend left] (b);
        	\draw[-, thick] (a) edge[bend right] (b);
	\draw[-, thick] (a) edge[bend right=70] (b);
}}

\def\Nttb{\tikz[baseline={(current bounding box.center)}]{
	\node[circle, fill, inner sep=1.5pt, align=center] (a) {};
	\node[circle, fill, inner sep=1.5pt, align=center, above=10pt] (b) at (a) {};
	\node[circle, fill, inner sep=1.5pt, align=center, right=20pt] (c) at (a) {};
	\node[circle, fill, inner sep=1.5pt, align=center, right=20pt] (d) at (b) {};
	\draw[-, thick] (a) edge[bend right] (c);
	\draw[-, thick] (a) -- (c);
	\draw[-, thick] (a) -- (d);
	\draw[-, thick] (a) edge[bend right] (d);
	\draw[-, thick] (b) -- (c);
	\draw[-, thick] (b) edge[bend left] (c);
	\draw[-, thick] (b) -- (d);
	\draw[-, thick] (b) edge[bend left] (d);
}}

\def\Nttc{\tikz[baseline={(current bounding box.center)}]{
	\node[circle, fill, inner sep=1.5pt, align=center] (a) {};
	\node[circle, fill, inner sep=1.5pt, align=center, below=10pt] (b) at (a) {};
	\node[circle, fill, inner sep=1.5pt, align=center, below=10pt] (c) at (b) {};
	\node[circle, fill, inner sep=1.5pt, align=center, below=10pt] (d) at (c) {};
				
	\node[circle, fill, inner sep=1.5pt, align=center, right=20pt] (e) {};
	\node[circle, fill, inner sep=1.5pt, align=center, right=20pt] (f) at (b) {};
	\node[circle, fill, inner sep=1.5pt, align=center, right=20pt] (g) at (c) {};
	\node[circle, fill, inner sep=1.5pt, align=center, right=20pt] (h) at (d) {};
	
	\draw[-, thick] (a) -- (e);
	\draw[-, thick] (a) -- (f);
	\draw[-, thick] (a) -- (g);
	\draw[-, thick] (a) -- (h);
	\draw[-, thick] (b) -- (e);
	\draw[-, thick] (b) -- (f);
	\draw[-, thick] (b) -- (g);
	\draw[-, thick] (b) -- (h);
	\draw[-, thick] (c) -- (e);
	\draw[-, thick] (c) -- (f);
	\draw[-, thick] (c) -- (g);
	\draw[-, thick] (c) -- (h);
	\draw[-, thick] (d) -- (e);
	\draw[-, thick] (d) -- (f);
	\draw[-, thick] (d) -- (g);
	\draw[-, thick] (d) -- (h);
}}

\begin{document}

\title{Trellis Decoding For Qudit Stabilizer Codes And Its Application To Qubit Topological Codes}

\author{Eric Sabo}
\email{esabo3@gatech.edu}
\affiliation{School of Mathematics, Georgia Institute of Technology, Atlanta, GA 30332, USA}

\author{Arun B. Aloshious}
\email{arunbarnabas.aloshious@duke.edu}
\affiliation{Department of Electrical and Computer Engineering, Duke University, Durham, NC 27708, USA}

\author{Kenneth R. Brown}
\email{ken.brown@duke.edu}
\affiliation{Department of Electrical and Computer Engineering, Duke University, Durham, NC 27708, USA}
\affiliation{School of Chemistry and Biochemistry, Georgia Institute of Technology, Atlanta, GA 30332, USA}
\affiliation{Department of Physics, Duke University, Durham, NC 27708, USA}
\affiliation{Department of Chemistry, Duke University, Durham, NC 27708, USA}

\begin{abstract}

	Trellis decoders are a general decoding technique first applied to qubit-based quantum error correction codes by Ollivier and Tillich in 2006. Here we improve the scalability and practicality of their theory, show that it has strong structure, extend the results using classical coding theory as a guide, and demonstrate a canonical form from which the structural properties of the decoding graph may be computed. The resulting formalism is valid for any prime-dimensional quantum system. The modified decoder works for any stabilizer code $S$ and separates into two parts: a one-time, offline computation which builds a compact, graphical representation of the normalizer of the code, $S^\perp$, and a quick, parallel, online query of the resulting vertices using the Viterbi algorithm. We show the utility of trellis decoding by applying it to four high-density, length 20 stabilizer codes for depolarizing noise and the well-studied Steane, rotated surface, and 4.8.8/6.6.6 color codes for $Z$-only noise. Numerical simulations demonstrate a 20\% improvement in the code-capacity threshold for color codes with boundaries by avoiding the mapping from color codes to surface codes. We identify trellis edge number as a key metric of difficulty of decoding, allowing us to quantify the advantage of single-axis decoding for Calderbank-Steane-Shor codes and block-decoding for concatenated codes.
\end{abstract}

\maketitle

\section{Introduction}
Quantum error correction allows for the in-principle preservation of quantum information in the presence of noise by encoding the information into a larger quantum system. The utility of any quantum error correction approach relies not only on the code parameters but also the existence of good decoders. There are presently few generic decoders applicable to random quantum error correcting codes. This is not without reason; quantum decoding is known to be computationally more difficult than classical decoding \cite{berlekamp1978inherent, iyer2015hardness}. Most decoders developed so far are specifically tailored to a given code family based on intuitive, visual, or physical arguments. Many decoders and decoding schemes in the literature are unique to the simulation they are presented with while others, especially those for topological codes, have enjoyed widespread use, study, and success.

For topological codes, such as toric, rotated surface, and color codes, the full decoding problem is typically reduced to that of finding a minimum-weight error, usually under an uncorrelated noise model, and is modeled as a perfect matching for the non-zero syndrome. The Blossom V implementation of Edmond's minimum-weight perfect matching (MWPM) algorithm is most often used \cite{fowler2013minimum, stephens2014fault} but other approaches such as greedy matching have also been applied \cite{wootton2015simple}. These are applicable to topological codes whenever the syndrome and errors have a string-like pattern. Despite the deceptive similarities to the surface codes, applying the same techniques to the color codes results instead in a hypergraph matching problem \cite{wang2009graphical}, which has no efficient algorithm. To get around this, some color code decoders first map the problem to two or more copies of a toric code which are then decoded and the information pieced back together to form a correction for the original code \cite{bombin2012universal, delfosse2014decoding, stephens2014efficient, aloshious2018local, kubica2019efficient}. The union find decoder \cite{delfosse2017almost, huang2020fault, delfosse2021union, delfosse2021toward} has been successfully applied to surface codes and homological product codes but a color code or more general stabilizer code implementation has still yet to be developed. Other popular decoders for topological codes include those based on cellular automata \cite{herold2015cellular, herold2017cellular, kubica2019cellular}, integer programing \cite{landahl2011fault}, and renormalization \cite{duclos2010fast, duclos2010renormalization, sarvepalli2012efficient}. Topological codes are often simulated on infinite lattices or finite lattices with periodic boundary conditions, and decoders sometimes require structural features such as locality and translational invariance.

Of particular interest are decoders which apply to multiple families of codes with little to no modification up to input data. Belief propagation \cite{mackay2004sparse, poulin2008iterative}, tensor network \cite{bravyi2014efficient, ferris2014tensor, darmawan2017tensor, darmawan2018linear, tuckett2018ultrahigh, chubb2021general}, and machine learning (ML) \cite{torlai2017neural, baireuther2018machine, chamberland2018deep, baireuther2019neural, chinni2019neural, maskara2019advantages, baireuther2018machine} based decoders generally fall into this category. Belief propagation is useful when codes satisfy specific sparsity properties but is inherently more difficult for quantum than classical codes. Recent work has improved this by introducing a common classical post processing step to prevent the decoder from getting stuck in loops \cite{panteleev2019degenerate, roffe2020decoding}. Tensor network decoders are theoretically exact maximum likelihood decoders but remain practically limited by computation with finite resources. As with much of machine learning, ML decoders trade good performance with a thorough understanding of its decisions and the theoretical guarantees that come with it.

One benefit to general decoding techniques is that they remove things such as geometric or topological constraints that complicate algorithms. Ollivier and Tillich introduced a decoding algorithm in 2006, called trellis decoding, which is applicable to any qubit stabilizer code $S$ \cite{ollivier2006trellises}. This is a port of a highly successful and well-understood classical decoder and works by building a highly compact and efficient graphical representation of the algebraic structure of the normalizer of $S$, $S^\perp$, called a trellis. The trellis contains all valid combinations of logical operators and stabilizer generators that will return the system to its code space. Decoding proceeds by using dynamical programming to globally search for the minimal weight path in the trellis corresponding to the measured syndrome. This is performed efficiently in exactly $n$ major steps for an $[[n, k, d]]$ stabilizer code, although the amount of work required in each step varies with respect to a predictable, code-dependent formula and can be significant depending on the amount of available resources. Beyond the brief introduction in \cite{ollivier2006trellises}, this decoder was discussed in \cite{poulin2009quantum, pelchat2013degenerate} for quantum convolutional codes and a portion of the algorithm was improved in \cite{xiao2013construction}. Many fundamental questions and theoretical properties remained unanswered.

As previously described in the literature, trellis decoding was simply not practical as it required the repeated processing of potentially massive amounts of data (all of the elements of $S^\perp$). A simple observation made in this paper allows us to reduce the number of times this data is processed to just once. However, this is still unpractical as enumerating the elements of $S^\perp$ may not be possible for many interesting codes. Here we expand on the previous literature by fully developing the theoretical foundations of the decoder. The outcome is a way to extract all the information needed to construct the trellis solely from the generators of $S$ and $S^\perp$. This object is independent of error model and may be computed once and saved for future simulations. The main goal and overall contribution of this paper is in making trellis decoding practical. To date, it is the most general and widely applicable decoding technique in quantum error correction. 

Non-binary (qudit) stabilizer codes are less explored and understood than qubit codes, yet potentially offer significant computational advantages. Of particular interest here is that qudit codes offer improved error thresholds over their qubit counterparts \cite{duclos2013kitaev, anwar2014fast, watson2015fast, aloshious2019local}. Non-binary decoders are lacking as many standard qubit decoders are not applicable to even direct generalizations of qubit codes without significant modification, if at all possible \cite{duclos2013kitaev, anwar2014fast, watson2015fast, hutter2015improved, marks2017comparison, aloshious2019local}. Although qudit codes are not simulated in this work, we notably extend the theory for any prime-dimension in Section \ref{sec:theory}, and the decoder works for both qubit and qudit stabilizer codes without any need for higher-dimensional modifications.

We begin by summarizing the main ingredients and fundamental results of our work in Section \ref{sec:construction}. A rigorous theory is developed in Section \ref{sec:theory}, while proofs of the numerous technical lemmas are relegated to Appendix \ref{sec:proofs} to increase readability. We attempt to mimic the classical coding theory literature as closely as possible to make this work available to the widest audience. It is perhaps remarkable that many of our results have classical analogues despite the differences between classical and quantum codes forcing alternative proof techniques. We view this as a strength of the overall theory of trellis decoding and point out similarities and differences, providing references to the classical literature, whenever possible. 

In Section \ref{sec:ViterbiProps} we adopt a classical, quantitative metric for the difficulty of decoding a given stabilizer code based on the structure of the trellis. This allows us to show that the color codes are fundamentally more difficult to decode than the rotated surface codes and the 4.8.8 color code is more difficult to decode than the 6.6.6 color code. We also make quantitative arguments as to exactly how much easier it is to decode a CSS code using independent $X$- and $Z$-decoders versus a single decoder for the full code. Although not unique to trellis decoding, we show that large gains are made when splitting the trellis into $X$- and $Z$-parts.

Numerical results for code-capacity (memory model) simulations are presented in Section \ref{sec:numerics}. First, we demonstrate the power of the theory by decoding four high-density codes from the quantum error-correcting code table at \cite{codetables}. As we are unaware of any other decoding algorithm for these codes, it is difficult to quantify the performance of the trellis with this data. We therefore, second, proceed to apply our work to the rotated surface and color code families for which there already exist numerous (and some highly optimized) decoders. Notably, the color codes (with boundary) are natively decoded without reference to other codes or notions of color, homology, matching, boundaries, lifting, projections, restrictions, charges, excitations, strings, and/or mappings. Trellis decoding for stabilizer codes is optimal minimum-weight (in contrast to classical trellis decoding which is maximum-likelihood) and numerical results for $Z$-only noise should match that of MWPM for the surface codes and exceed the current best thresholds in the literature for the color codes, although we only simulate codes up to distances 17-21. Finally, we simulate a 2-stage, suboptimal decoder for the level-2 concatenated Steane code.

\section{The Syndrome Trellis}\label{sec:construction}
Let $p$ be a prime and $\F_p$ the finite field of cardinality $p$. Denote the generalized Pauli group on one qudit of dimension $p$ by $\mathcal{P}(1, p) = \< X, Z \> = \{\omega^c X^a Z^b \, : \, a, b, c \in \F_p\}$, where $\omega$ is a primitive $2p$th or $p$th root of unity if $p$ is even or odd, respectively, and the generalized Pauli group on $n$ qudits by
\begin{equation*}
	\mathcal{P}(n, p) = \mathcal{P}(1, p)^{\otimes n} = \{ \omega^c X^a Z^b \, : \, a, b \in \F_p^n, c \in \F_p\}.
\end{equation*}
We will refer to elements in the form $IXXIXZY$ as Pauli strings to distinguish them from any of their numerical representations. By \textit{stabilizer code} we mean an Abelian subgroup $S \leq \mathcal{P}(n, p)$ defined with respect to some nice error basis, here the standard $X$ and $Z$ qudit Pauli operators. This is typically followed by a map to $\F_p^{2n}$ but we will not utilize this at the moment and will come back to this differentiation later. See \cite{ketkar2006nonbinary} for details. Defined as such, $S$ is a classical code which corresponds to a \textit{quantum stabilizer code}, $\mathcal{Q} = \mathcal{Q}(S) = \{ v \in \mathbb{C}^{p^n} \mid P v = v, \forall P \in S \}$, for which $S$ is the \textit{stabilizer group}. In what follows, suppose $S$ is the stabilizer group of a quantum stabilizer code with $\dim \mathcal{Q} > 1$ and that the number of non-identity Pauli operators of any element in $S$ is at least two. The center of $\mathcal{P}(n, p)$, $\mathcal{Z}(\mathcal{P}(n, p)) = \{ \omega^c \, : \, c \in \F_p \}$, is the set of elements which commute with all other elements of $\mathcal{P}(n, p)$. The centralizer of $S$ in $\mathcal{P}(n, p)$, $C_{\mathcal{P}(n, p)}(S)$, is all the elements of $\mathcal{P}(n, p)$ that commute with all elements of $S$. The dual of $S$ is defined to be the set of all elements of $\mathcal{P}(n, p)$ which have zero (symplectic) inner product with all the elements of $S$, i.e., $S^\perp \cong C_{\mathcal{P}(n, p)} (S)/\mathcal{Z}(\mathcal{P}(n, p))$. The dual hence contains all stabilizers and all logical Pauli operators for $\mathcal{Q}$ and has dimension $n + k$.

The classical coding theory literature contains several definitions for the trellis of a linear block code. Some of them, such as the the Bahl-Cocke-Jelinek-Raviv (BCJR) \cite{bahl1974optimal}, the Wolf \cite{wolf1978efficient}, and the Forney-Muder \cite{forney1988coset, muder1988minimal} trellises, are now known to be isomorphic. With the benefit of this hindsight, Ollivier and Tillich demonstrated how to port what is often known as the \textit{syndrome trellis} to the quantum setting \cite{ollivier2006trellises}. They referred to their construction as the Wolf trellis, however, the discussion here more closely follows that of BCJR. As such, we simply use \textit{syndrome trellis} in this work, although, since unlike classical coding theory, there is only one definition of a trellis for quantum error correction, we may just as well shorten this to simply \textit{trellis}. One may attempt to port the other classical trellises to the quantum case but the definition used in \cite{ollivier2006trellises} and this work is known to be minimal in a rigorous sense that will be defined in Section \ref{sec:theory}, rendering alternate definitions undesirable. The interested reader is referred to the tutorial piece \cite{mceliece1996bcjr} for a general discussion of classical trellises. 

Recall that a directed edge, $e$, in a graph goes from \textit{source}, $s(e)$, to \textit{terminus}, $t(e)$.
\begin{definition}[(Quantum) Syndrome Trellis]
	A \emph{trellis} for an $[[n, k, d]]_p$ stabilizer code $\mathcal{Q}$ is a directed multigraph with vertex set, $V$, and edge set, $E$, such that
	\begin{enumerate}
		\item there are $n + 1$ disjoint sets of vertices $V_i$ with $V = V_0 \sqcup \hdots \sqcup V_n$ and $|V_0| = |V_n| = 1$;
		\item there are $n$ disjoint sets, $E_i$, of directed edges from $V_{i - 1}$ to $V_i$ with $E = E_1 \sqcup \hdots \sqcup E_n$;
		\item each vertex $v \in V_i$ has a unique label given by an $(n - k)$-tuple of syndromes, although the same label may be present in multiple $V_i$;
		\item each edge $e \in E_i$ is a unique triple of the form $(s(e), P, t(e))$, where $P$ is a label of the form $X^a Z^b$ for $a, b \in \F_p$;
		\item each edge is assigned a weight, $\wt(e) \in \mathbb{R}^- \cup \{-\infty\}$.
	\end{enumerate}
\end{definition}

Following the classical literature, vertices are referred to as \textit{states}, the $V_i$ as \textit{state spaces}, and $V_i$ is said to be at \textit{depth} $i$. The edge sets $E_i$ are referred to as the $i$th \textit{section} and the edges as \textit{branches}. The condition that edge labels must be unique between a fixed pair of source and terminus vertices means our trellises are \textit{proper}. Trellises which do not satisfy this condition are called \textit{improper} and are not considered in this work. It is perhaps more convenient to define the trellis without weights \textit{(v)}, but we stick with decades of precedent in including them here.

\begin{figure}[t]
	\centering
     	\begin{subfigure}{\textwidth}
        		\centering
		\includegraphics[scale=1]{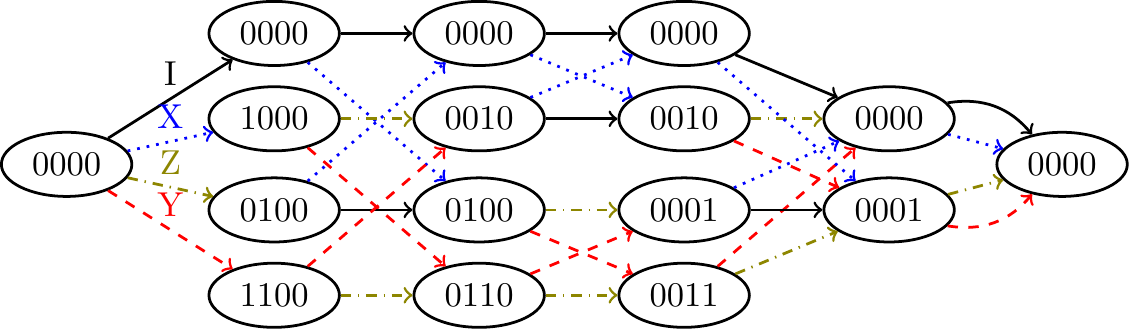}
		\caption{}
	\end{subfigure}
	\vspace{1cm}
	\begin{subfigure}{\textwidth}
        		\centering
		\includegraphics[scale=1]{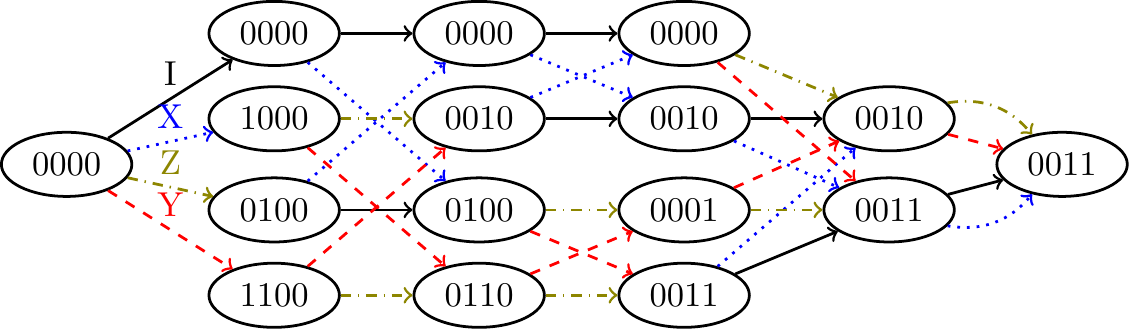}
		\caption{}
	\end{subfigure}
	\vspace{-1cm}
	\caption{(a) The trellis for the $[[5, 1, 1]]$ code with stabilizers $\{ZXIII, XZXII, IXZXI, IIXZX\}$ demonstrated in \cite{ollivier2006trellises}. Measuring the syndrome $s = (0, 0, 1, 1)$, we determine $P_s = IIIZZ$. Adding $\pi_0(P_s) = \pi_1(P_s) = \pi_2(P_s) = \pi_3(P_s) = (0, 0, 0, 0)$, $\pi_4(P_s) = (0, 0, 1, 0)$, and $\pi_5(P_s) = (0, 0, 1, 1)$ to the appropriate $V_i$ produces (b), the trellis found in Figure 2 of \cite{ollivier2006trellises}. The presence of the parallel edges in $E_5$ demonstrate that this code cannot tell the difference between $I$/$X_5$ and $Y_5$/$Z_5$.}
	\label{fig:trellisex}
\end{figure}

To construct the syndrome trellis, choose a set of stabilizer generators $\{ S_1, \hdots, S_{n - k}\}$ for the stabilizer group $S$, and fix this set throughout.
\begin{definition}
	Let $\mathcal{J} \subset \{0, \hdots, n\}$ be an index set. Then define $\pi_\mathcal{J}: \mathcal{P}(n, p) \to \mathcal{P}(n, p)$ to be the map which projects all Pauli operators at indices $\mathcal{J}^c$ to the identity, where $\pi_0(P)$ is the all identity Pauli string.
\end{definition}
\noindent To avoid ambiguities with which qudit receives the relative phase, assume $\omega^c = 1$. Let $\sigma_i : \mathcal{P}(n, p) \to \F_p^{n - k}$ be the map from a Pauli string to its syndrome,
\begin{equation}
	\sigma_i(P) = (\<S_1, \pi_{\{0, \hdots, i\}}(P) \>, \hdots, \<S_{n - k}, \pi_{\{0, \hdots, i\}}(P) \>),
\end{equation}
where $\< \cdot, \cdot \>$ is an appropriate inner product, and let $P_s$ be a Pauli string with syndrome $s$ with respect to the chosen generators, $s = \sigma_n(P_s)$. Compute $S^\perp$ and then $P_s S^\perp$. Since everything in $S^\perp$ has zero syndrome, everything in $P_s S^\perp$ has syndrome $s$.

The vertices in each set are labeled by the values $\sigma_i(P)$ for all $P \in P_s S^\perp$,
\begin{equation}\label{vertexset}
	V_i = \{ \sigma_i(P) \, : \, 1 \leq j \leq n - k, P \in P_s S^\perp \}.
\end{equation}
A directed edge is created from the vertex $\sigma_{i - 1}(P) \in V_{i - 1}$ to the vertex $\sigma_i(P) \in V_i$ and labeled with the $i$th component of $P$,
\begin{equation}\label{edgeset}
	E_i = \{ (\sigma_{i - 1}(P), P_i, \sigma_i(P)) \, : \, P \in P_s S^\perp \}.
\end{equation}
The weight of an edge with label $P_i$ is defined to be the log-likelihood $-\log \Pr(P_i)$, where this probability comes from the assumed error channel. 
Duplicate edges, which have the same source, label, and terminus, are not allowed, although they will appear often during this method of construction. Parallel edges, which have the same source and terminus but different labels, are allowed and imply the existence of a weight one error for the code. These will generally not appear, but we will consider them in this work for completeness. The example trellis of Figure \ref{fig:trellisex} has parallel edges in $E_5$.

This represents the construction process of \cite{ollivier2006trellises} and \cite{xiao2013construction}. Assuming $S^\perp$ is able to be computed and stored, the trellis as written so far is dependent on $P_s$ and must be recomputed for every measured syndrome, which is computationally expensive for a practical implementation for many codes. This can be avoided by noting that since for $P \in P_s S^\perp$, $\sigma_i(P) = \sigma_i(P_s) + \sigma_i(P^\prime)$ for $P^\prime \in S^\perp$, the trellis with respect to syndrome $s$ is simply a shift of the trellis with respect to the zero syndrome. Since $\sigma_i(P_s)$ is a single value, the set $\{ \sigma_i(P_s) + \sigma_i(P^\prime) \}$ is unique when $\{ \sigma_i(P^\prime) \}$ is unique so $|V_i|$ remains invariant. Likewise, the map $(\sigma_{i - 1}(P^\prime), P^\prime_i, \sigma_i(P^\prime)) \mapsto (\sigma_{i - 1}(P_s) + \sigma_{i - 1}(P^\prime), P_{s \, i} P^\prime_i, \sigma_i(P_s) + \sigma_i(P^\prime))$ is an isomorphism permuting edge labels via the action of $P_{s \, i}$. It follows that one may pre-compute the trellis for the zero syndrome then update each $V_i$ with the syndrome of $\pi_i(P_s)$, updating edges accordingly. See Figure \ref{fig:trellisex} for an example. We may thus trade the affine space $P_s S^\perp$ with the vector space $S^\perp$.\footnote{Up until now we have considered $S^\perp$ as a set of Pauli strings and not as a subset of $\F_p^n$. One may define an equivalent vector space over the strings by defining scalar multiplication as the raising of a Pauli operator to a power.} This is convenient as many mathematical objects are not well-defined over affine spaces and working with the associated vector space allows for easier proofs of properties that more closely mimic their classical counterparts. 

We show in Section \ref{sec:theory} that the vertices and edges are highly constrained. At depth and section $i$ there are (Theorem \ref{thm:quantstate}) $|V_i| = p^{\dim S^\perp - \dim S^\perp_{\mathfrak{p}_i} - \dim S^\perp_{\mathfrak{f}_i}}$ vertices and $|E_i| = p^{\dim S^\perp - \dim S^\perp_{\mathfrak{p}_{i - 1}} - \dim S^\perp_{\mathfrak{f}_i}}$ edges. Every vertex in $V_i$ has the same incoming, $\deg_\text{in}(v) = p^{\dim S^\perp_{\mathfrak{p}_i} - \dim S^\perp_{\mathfrak{p}_{i - 1}}}$, and outgoing, $\deg_\text{out}(v) = p^{\dim S^\perp_{\mathfrak{f}_i} - \dim S^\perp_{\mathfrak{f}_{i + 1}}}$, degrees (Corollary \ref{cor:inoutdegs}), and only certain patterns of edge connections are allowed (Theorem \ref{thm:bipartthm}). For a fixed qudit ordering the structure of the trellis is invariant of the choice of stabilizer or logical generators (Lemma \ref{lem:lemmafourteen}), but a specific choice of generators, the trellis oriented form, allows us to immediately read-off the dimensions of the sets in the previous formulas (Proposition \ref{prop:dimargs}).

We take the trellis for $S^\perp$ (with respect to the zero syndrome) to be the fundamental object in this work. The general decoding scheme proceeds as follows: 1) construct the trellis, 2) measure a syndrome and shift the base trellis with respect to it, 3) decode the shifted trellis, 4) repeat steps 2) and 3). We avoid assigning edge weights to the trellis until step 2) so the trellis resulting from step 1) is independent of any measured syndromes or error model. Hence step 1) is an offline procedure while 2) and 3) are online. Trellises for larger distance codes may therefore be computed once and saved for future use. For the codes considered in this work, the efficiency of the construction algorithm made this unnecessary for most distances. Since $|V| \leq |E|$, shifting the trellis is $\Theta(|E|)$ in time but the procedure is easily parallelized. As we will see later, decoding is also $\Theta(|E|)$ and may also be parallelized.

With the edge weights prescribed above, the desired correction is given by the minimum-weight path from $V_0$ to $V_n$. We refer to this in the following as the ``optimal path". The Viterbi algorithm is an example of forward dynamical programming and is the most common trellis-based decoder \cite{forney1973viterbi}. The idea is as follows. Choose an arbitrary vertex $v \in V_i$ for some $i \neq 0$ or $n$ and suppose that the optimal path travels through $v$. Then the optimal path can be split into two parts: the optimal path going from $V_0$ to $v$ and the optimal path going from $v$ to $V_n$. Compute the optimal path from $V_0$ to $v$ and repeat for all $v$. Once the optimal paths are computed for every vertex in $V_{i - 1}$, the optimal path for a vertex in $V_i$ is chosen by finding the minimum value of the sums of the incoming edge weights plus the weight of the optimal path for each edge's source vertex in $V_{i - 1}$. The weight at $V_0$ is arbitrary but is convenient to initialize to zero. Having completed this for all vertices, the Pauli correction may be read off the trellis by taking the edge labels for the optimal path connecting $V_0$ to $V_n$. See Appendix \ref{sec:viterbiapp} for a visual example of the Viterbi algorithm run on Figure \ref{fig:trellisex}; an example implementation is also provided. The algorithm is often described in a manner in which vertices with no outgoing edges are removed from the system. This creates nicer diagrams but actual deletion steps can be algorithmically expensive and should be ignored as in Figure \ref{fig:trellisex}. Also note that the partial syndromes, or vertex labels, make no appearance in the algorithm and may be safely removed after the construction phase, also ignoring vertices in the shifting phase. The partial syndromes of the distance 15 rotated surface code are 224 bits long, for example, and not storing them, even in a more efficient fashion, leads to large savings in storage.

The Viterbi algorithm is sometimes confused with other standard minimum-path graph algorithms but it is distinct and more efficient given the strict edge and vertex dependencies required in the definition of a trellis. Underlying the success of this algorithm is the implicit assumption that the trellis behaves as a Markov chain with each $V_i$ only depending on the result at $V_{i - 1}$. In particular, at each $V_i$, all the information regarding the correction of all previous qudits has already been processed.

The behavior of the decoder depends strongly on the assumed error model. If all Pauli operators are equally likely, as in the standard depolarizing noise model, the Viterbi algorithm acts as a minimum Hamming-weight decoder and may run into a significant number of ties which may be broken arbitrarily. For a biased noise model, the decoder is able to differentiate between, say, $X$ and $Z$, and is able to make a more informed choice.

\section{Properties}\label{sec:theory}
Throughout this section we denote the vertex with zero syndrome by $\overline{0}$. Dimension will always refer to the number of generators of an object, and $| \cdot |$ will be reserved for cardinality. A Pauli string will be denoted by $P$, whereas $p$ will be a prime number. We first show several properties of the syndrome trellis then we turn to trellis-oriented forms before finally discussing the Viterbi algorithm. Many, but not all, of the results here have classical analogies, and we attempt to provide original citations to the best of our knowledge for those which are not typically discussed in textbooks. Classically, however, codes are vector spaces whereas here we consider stabilizer codes as groups. This changes the techniques, but we try to maintain the same overall direction of the proofs if possible. We consider it a strength of the theory that a single framework can handle both classical and quantum trellises. The closest classical work to this is that of Forney and Trott on convolutional group codes \cite{forney1993dynamics}.

\subsection{Syndrome Trellis}
The first result was understood and implicitly used in \cite{ollivier2006trellises} and \cite{xiao2013construction} but was never explicitly stated. It is perhaps obvious, but we include a proof here for completeness.
\begin{restatable}[]{proposition}{propone}
	\label{prop:iso}
	There is a one-to-one correspondence between Pauli strings in $S^\perp$ and length-$n$ paths in the trellis.
\end{restatable}
\begin{restatable}[]{corollary}{corotwo}
	\label{cor:bruteforce}
	Let $v_i \in V_i$ and $v_{i + 1} \in V_{i + 1}$ be arbitrary. If there is a Pauli element, $P_i$, which takes $v_i$ to $v_{i + 1}$, then $(v_i, P_i, v_{i + 1}) \in E_i$.
\end{restatable}

In the previous section we attempted to stick to the notation of the preceding literature, following \cite{ollivier2006trellises}. At this point, we make some adjustments which we find necessary to simplify the proofs and discussion. The main reason for this is that $S$ and $S^\perp$ are groups whereas $V_i$ and $E_i$ are vector spaces. Promoting the Pauli strings to vectors or the vector spaces to groups provides for a more coherent argument. Choosing Pauli strings over their symplectic representation, we henceforth introduce the following notation. We begin with a rather general definition for completeness; however, our goal and use case throughout this paper is the rather natural splitting of the qudit indices into a ``past" and ``future", which is defined later in Definition \ref{def:pastfut}.
\begin{definition}
	\leavevmode
	\vspace*{-\bigskipamount}\vspace*{-\medskipamount}
	\begin{enumerate}
		\item[$\bullet$] Let $\mathcal{J} \subset \{0, \hdots, n\}$ be an index set and define $\mathcal{A}_\mathcal{J}$ to be the set of Pauli strings in a subgroup $\mathcal{A} \leq \mathcal{P}(n, p)$ whose Pauli operators are equal to the identity on the complement of $\mathcal{J}$, $\mathcal{J}^c \coloneqq \{0, \hdots, n\} \setminus \mathcal{J}$.
		\item[$\bullet$] Denote by $\mathcal{A}_{|_\mathcal{J}}$ the set of length-$|\mathcal{J}|$ Pauli strings constructed from Pauli strings in $\mathcal{A}$ whose elements at indices $\mathcal{J}^c$ have been deleted. Alternatively, $\mathcal{A}_{|_\mathcal{J}}$ is image of $\mathcal{A}$ whose elements at indices $\mathcal{J}^c$ have been set to the identity.
	\end{enumerate}
\end{definition}
\noindent The alternative definition of $\mathcal{A}_{|_\mathcal{J}}$ is useful to keep all Pauli strings the original length, whereas in the first definition the strings are shortened to length $|\mathcal{J}|$. Both definitions are conceptually equivalent.\\

\noindent {\bf Remark:} If $\mathcal{C}$ is a classical code, then $\mathcal{C}_\mathcal{J}$ and $\mathcal{C}_{|_\mathcal{J}}$ are the shortened and punctured codes of $\mathcal{C}$ with respect to $\mathcal{J}$, respectively. The concepts of shortening and puncturing are a bit more complicated for quantum codes \cite{rains1999nonbinary}, so we refrain from using this terminology here.\\

\begin{example}
	Let $\mathcal{J} = \{0, \hdots, i\}$.  Then $\mathcal{A}_\mathcal{J}$ is the set Pauli strings which naturally have identity elements in indices $\{i + 1, \hdots, n\}$ and $\mathcal{A}_{|_\mathcal{J}} = \im \pi_{\mathcal{J}}(\mathcal{A})$ is the set of Pauli strings whose elements in indices $\{i + 1, \hdots, n\}$ have been projected to the identity regardless of their initial value.
\end{example}

Critical to our proofs is the fact that since $\pi_\mathcal{J}(PP^\prime) = \pi_\mathcal{J}(P) \pi_\mathcal{J}(P^\prime)$ for $P, P^\prime \in \mathcal{P}(n, p)$, $\pi_{\mathcal{J}}$ is a group homomorphism. Then $\ker \pi_\mathcal{J}(\mathcal{A}) = \mathcal{A}_{|_{\mathcal{J}^c}}$ and $\ker \pi_{\mathcal{J}^c}(\mathcal{A}) = \mathcal{A}_{|_\mathcal{J}}$, and $\mathcal{A}_{|_{\mathcal{J}^c}} \trianglelefteq \mathcal{A}$ and $\mathcal{A}_{|_\mathcal{J}} \trianglelefteq \mathcal{A}$ as all kernels are normal. (Recall the notation $A \trianglelefteq B$ means $A$ is a normal subgroup of $B$.) The product group $\mathcal{A}_{|_\mathcal{J}} \mathcal{A}_{|_{\mathcal{J}^c}}$ is also normal, and as $\mathcal{J} \cap \mathcal{J}^c = \emptyset$, $\mathcal{A}_{|_\mathcal{J}} \mathcal{A}_{|_{\mathcal{J}^c}} = \mathcal{A}_{|_\mathcal{J}} \times \mathcal{A}_{|_{\mathcal{J}^c}} \trianglelefteq \mathcal{A}$.

In keeping with the standard trellis literature, we introduce the following further terminology.
\begin{definition}\label{def:pastfut}
	Fix an index $i \in \{0, \hdots, n\}$. Then the \emph{past} and \emph{future}, with respect to $i$, are defined to be the index sets $\mathfrak{p}_i \coloneqq \{0, \hdots, i\}$ and $\mathfrak{f}_i \coloneqq \{i + 1, \hdots, n\}$, respectively.
\end{definition}
\noindent In particular, we will utilize the sets of Pauli strings $S_{\mathfrak{p}_i}$, $S_{\mathfrak{f}_i}$, $S_{|_{\mathfrak{p}_i}}$, $S_{|_{\mathfrak{f}_i}}$, and likewise for $S^\perp$. See Figure \ref{fig:commdiag} for a summary of the relationships between the groups. Note that the literature varies on which set includes the index $i$. In this work we index the qudits on $\{1, \hdots, n \}$ but it is necessary to include zero in order to project to the identity.
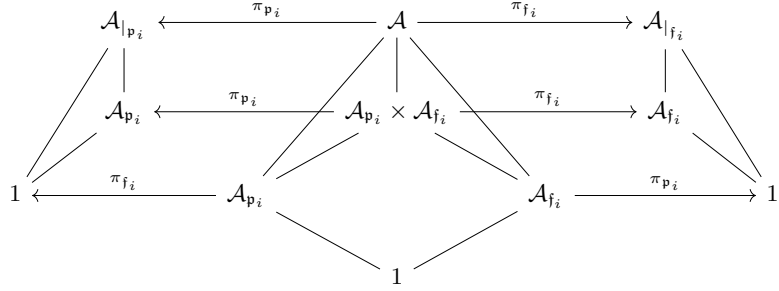
\begin{figure}
	\centering
            \begin{tikzcd}
            	& \mathcal{A}_{|_{\mathfrak{p}_i}} \arrow[ddl, dash] \arrow[d, dash] & & \mathcal{A} \arrow[ll, "\pi_{\mathfrak{p}_i}" above] \arrow[rr, "\pi_{\mathfrak{f}_i}" above] \arrow[ddl, dash] \arrow[d, dash] \arrow[ddr, dash] & & \mathcal{A}_{|_{\mathfrak{f}_i}} \arrow[ddr, dash] \arrow[d, dash] &\\
            	& \mathcal{A}_{\mathfrak{p}_i} \arrow[dl, dash] & & \mathcal{A}_{\mathfrak{p}_i} \times \mathcal{A}_{\mathfrak{f}_i} \arrow[ll, "\pi_{\mathfrak{p}_i}" above] \arrow[rr, "\pi_{\mathfrak{f}_i}" above]  \arrow[dl, dash] \arrow[dr, dash] & & \mathcal{A}_{\mathfrak{f}_i} \arrow[dr, dash] &\\
            	1 & & \mathcal{A}_{\mathfrak{p}_i} \arrow[ll, "\pi_{\mathfrak{f}_i}" above] \arrow[dr, dash] & & \mathcal{A}_{\mathfrak{f}_i} \arrow[rr, "\pi_{\mathfrak{p}_i}" above] \arrow[dl, dash] & & 1\\
            	& & & 1 & & &
            \end{tikzcd}
	\caption{A summary of the relationships between the groups and the maps (Figure 5 of \cite{forney1993dynamics}). A vertical line without an arrow between two groups means the bottom group is a subgroup of the top group.}
	\label{fig:commdiag}
\end{figure}

The previously mentioned subgroups (Figure \ref{fig:commdiag}) depend only on the set $\mathcal{A}$ and not the choice of generators, although $S$ and $S^\perp$ may have different subgroups at the same index. Lagrange's theorem restricts their cardinality to powers of $p$. (In fact, this was a motivation for promoting $V$ and $E$ to groups instead of dealing with less-constrained vector subspace dimensions.) Initially, $\mathcal{A}_{\mathfrak{f}_0} = \mathcal{A}$ and monotonically decreases in size as $i$ goes to $n$, $\mathcal{A}_{\mathfrak{f}_n} = \{ I^{\otimes n} \}$. Likewise, $\mathcal{A}_{\mathfrak{p}_0} = \{ I^{\otimes n} \}$ and monotonically increases in size as $i$ goes to $n$, $\mathcal{A}_{\mathfrak{p}_n}= \mathcal{A}$.

For a fixed $i$, sort the elements of $S^\perp$ into sets $S^\perp_{\mathfrak{p}_i}$, $S^\perp_{\mathfrak{f}_i}$, and $S^\perp_{\mathfrak{a}_i}$, where the ``active" set, denoted by $\mathfrak{a}$, contains the remaining elements neither wholly in the past nor future. Include the identity in $S^\perp_{\mathfrak{a}_i}$ to give these trivial intersection. An element of $S^\perp$ may hence be decomposed in the form $P_\mathfrak{p} P_\mathfrak{a} P_\mathfrak{f}$, where $P_\mathfrak{p} \in S^\perp_{\mathfrak{p}_i}$, $P_\mathfrak{a} \in S^\perp_{\mathfrak{a}_i}$, and $P_\mathfrak{f} \in S^\perp_{\mathfrak{f}_i}$. An element $P_{|_{\mathfrak{p}_i}} \in S^\perp_{|_{\mathfrak{p}_i}}$ is of the form $P_{|_{\mathfrak{p}_i}} = \left(P_{\mathfrak{p}_i}\right)_{|_{\mathfrak{p}_i}} \left(P_{\mathfrak{a}_i}\right)_{|_{\mathfrak{p}_i}}$, and similarly for the future. Thus $S^\perp_{|_{\mathfrak{p}_i}} = \text{span}\left\{ S^\perp_{{\mathfrak{p}_i}}, \left(S^\perp_{\mathfrak{a}_i}\right)_{|_{\mathfrak{p}_i}}\right\}$ and $S^\perp_{|_{\mathfrak{f}_i}} = \text{span}\left\{ S^\perp_{{\mathfrak{f}_i}}, \left(S^\perp_{\mathfrak{a}_i}\right)_{|_{\mathfrak{f}_i}}\right\}$. Fix a Pauli string $P \in S^\perp_{\mathfrak{a}_i}$ and let $v$ be the vertex generated by $P$ in $V_i$. Then every length-$i$ path from $V_0$ to $v$ is generated by the coset $P_{|_{\mathfrak{p}_i}} S^\perp_{{\mathfrak{p}_i}}$. Likewise, every path from $v$ to $V_n$ is generated by an element of the coset $P_{|_{\mathfrak{f}_i}} S^\perp_{{\mathfrak{f}_i}}$. Putting these together, every path from $V_0$ to $V_n$ may be viewed as a coset $P_{|_{\mathfrak{p}_i}} S^\perp_{{\mathfrak{p}_i}} P_{|_{\mathfrak{f}_i}} S^\perp_{{\mathfrak{f}_i}}$ with respect to $v$. See Figure \ref{fig:trelliscosets} for a visual summary. 
\begin{figure}
	\begin{center}
		\includegraphics[scale=1]{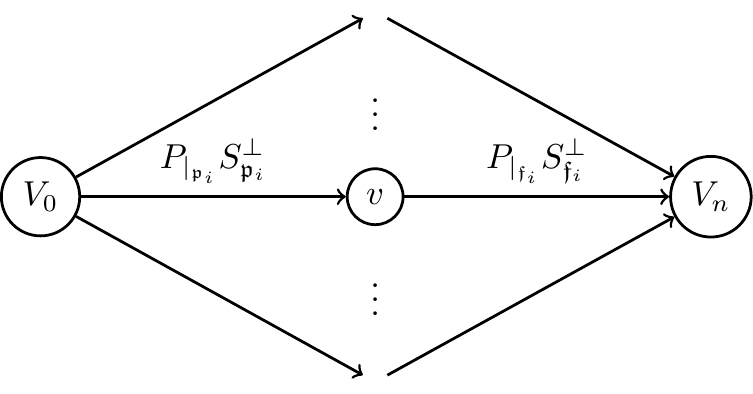}
	\end{center}
	\caption{Grouping all paths from $V_0$ to $v$ into a single line, the trellis may be seen as the depicted set of cosets. As is clear in the diagram, each vertex has a unique past and future.}
	\label{fig:trelliscosets}
\end{figure}

As mentioned in the previous section, there are historically numerous definitions of trellises in the classical literature. It was therefore important to determine whether some are ``better" than others, leading to the concept of a \textit{minimal trellis} (defined below). We will see below that if another trellis would be defined for stabilizer codes, the trellis of this work is minimal. The following proposition takes us in this direction and holds for any definition of a trellis for a stabilizer code.
\begin{restatable}[Lemma 1 \cite{ollivier2006trellises}]{proposition}{propfive}
	\label{prop:minbounds}
	Given a stabilizer code $S$, any trellis for $S$ must satisfy
	\begin{align*}
		|V_i| &\geq p^{\dim S^\perp - \dim S^\perp_{\mathfrak{p}_i} - \dim S^\perp_{\mathfrak{f}_i}},\\
		|E_i| &\geq p^{\dim S^\perp - \dim S^\perp_{\mathfrak{p}_{i - 1}} - \dim S^\perp_{\mathfrak{f}_i}}.
	\end{align*}
\end{restatable}

\begin{definition}
	A trellis that meets the lower bounds of Proposition \ref{prop:minbounds} is \emph{minimal}.\\
\end{definition}

\noindent {\bf Remark:} Non-minimal trellises will not be discussed in this work. As a trivial example, one could define a trellis where each element of $S^\perp$ consists of its own path from $V_0$ to $V_n$ without intersecting any other path. The number of vertices at each depth would then be $|S^\perp|$. We refer the reader to \cite{mceliece1996bcjr} for further (classical) examples. The concept of a minimal trellis comes from \cite{muder1988minimal}.\\

While the edges of the trellis encode $S^\perp$, the vertices do not take the logical operators into account and are determined by $S$, an important point in sharp contrast with the classical case whose vertices and edges are both constructed relative to the same object: the code.

For the vertices, fix $i$, let $P$ be an arbitrary element of $S^\perp$, and let $S_j$ be a fixed generator of $S$. If $S_j \in S_{\mathfrak{p}_i}$ then $\< S_j, \pi_{\{0, \hdots, i\} }(P) \> = 0$ since $S_j$ is the identity at indices $\{ i + 1, \hdots, n\}$ and these will commute with any elements of $P$ at these positions. Thus, the syndrome of $P$ with respect to the generator $S_j$ has already been completely determined (and is zero since $P \in S^\perp$). Likewise, if $S_j \in S_{\mathfrak{f}_i}$ then $\< S_j, \pi_{\{0, \hdots, i\} }(P) \> = 0$ since $S_j$ is the identity at indices $\{1, \hdots, i \}$, which, again, commute with $\pi_{\{0, \hdots, i\} }(P)$. Only generators $S_j \in S_{\mathfrak{a}_i}$ can have nonzero syndromes.

\begin{restatable}[Quantum Space Theorem(s)]{theorem}{thmseven}
	\label{thm:quantstate}
	The syndrome trellis of Section \ref{sec:construction} is minimal, i.e.,
	\begin{enumerate}
		\item (Quantum State Space Theorem, Lemma 2 \cite{ollivier2006trellises})
			\begin{equation}\label{Vi}
				|V_i| = p^{\dim S^\perp - \dim S^\perp_{\mathfrak{p}_i} - \dim S^\perp_{\mathfrak{f}_i}}
			\end{equation}
		\item (Quantum Branch Space Theorem)
			\begin{equation}\label{Ei}
				|E_i| = p^{\dim S^\perp - \dim S^\perp_{\mathfrak{p}_{i - 1}} - \dim S^\perp_{\mathfrak{f}_i}}
			\end{equation}
	\end{enumerate}
\end{restatable}

\noindent {\bf Remark:} The sequences $\{|V_i|\}^n_{i = 0}$ and $\{|E_i|\}^n_{i = 1}$ are often referred to as the \textit{state space} and \textit{branch space} complexity profiles, respectively. The quantity $\max_i \dim V_i$ is the \textit{state complexity}, $\max_i \dim E_i$ the \textit{branch complexity}, and $|E|$ the \textit{edge complexity}. We will not utilize these in this work.\\

\begin{restatable}[]{lemma}{lemmaeight}
	\label{lem:VEvs}
	The set of edges in $E_i$ with terminus vertex $\overline{0}$ is a subgroup of $E_i$.
\end{restatable}

\begin{restatable}[]{corollary}{coronine}
	\label{cor:inoutdegs}
	For $1 \leq i \leq n$, every vertex $v \in V_i$ has incoming degree
	\begin{equation}
		\deg_\text{in}(v) = p^{\dim S^\perp_{\mathfrak{p}_i} - \dim S^\perp_{\mathfrak{p}_{i - 1}}},
	\end{equation}
	and for $0 \leq i \leq n - 1$, every vertex $v \in V_i$ has outgoing degree
	\begin{equation}
		\deg_\text{out}(v) = p^{\dim S^\perp_{\mathfrak{f}_i} - \dim S^\perp_{\mathfrak{f}_{i + 1}}}.
	\end{equation}
\end{restatable}

\begin{restatable}[]{lemma}{lemmaten}
	\label{lem:deltadim}
	Let $\mathcal{A} = S$ or $S^\perp$ and appropriate $i$ be fixed. Then
	\begin{gather*}
		0 \leq \dim \mathcal{A}_{\mathfrak{f}_i} - \dim \mathcal{A}_{\mathfrak{f}_{i + 1}} \leq 2,\\
		0 \leq \dim \mathcal{A}_{\mathfrak{p}_i} - \dim \mathcal{A}_{\mathfrak{p}_{i - 1}} \leq 2.
	\end{gather*}
\end{restatable}

\begin{corollary}
	For appropriate $i$, $\deg_\text{in/out}(v) \in \{1, p, p^2\}$. Equivalently, for $1 \leq i \leq n$, $\dim V_i \leq \dim E_i \leq \dim V_i + 2$.
\end{corollary}

Table \ref{tab:edgeconfigs} records this result graphically for the special case $p = 2$. It is well-known in classical trellis theory that the dimension change in Lemma \ref{lem:deltadim} is bounded by one \cite{forney1994dimension}. This is due to classical codes only having one symbol alphabet instead of the two, $\< X \>$ and $\< Z \>$, which allows one to completely row reduce to a single pivot per column. Any configuration in Table \ref{tab:edgeconfigs} with a two is therefore unique to the quantum setting. The next theorem shows that these are the only possible configurations. No corresponding proof exists in the classical literature.
\begin{table}[t]
	\centering
	\begin{tabular}{c}
		\includegraphics[scale=0.9]{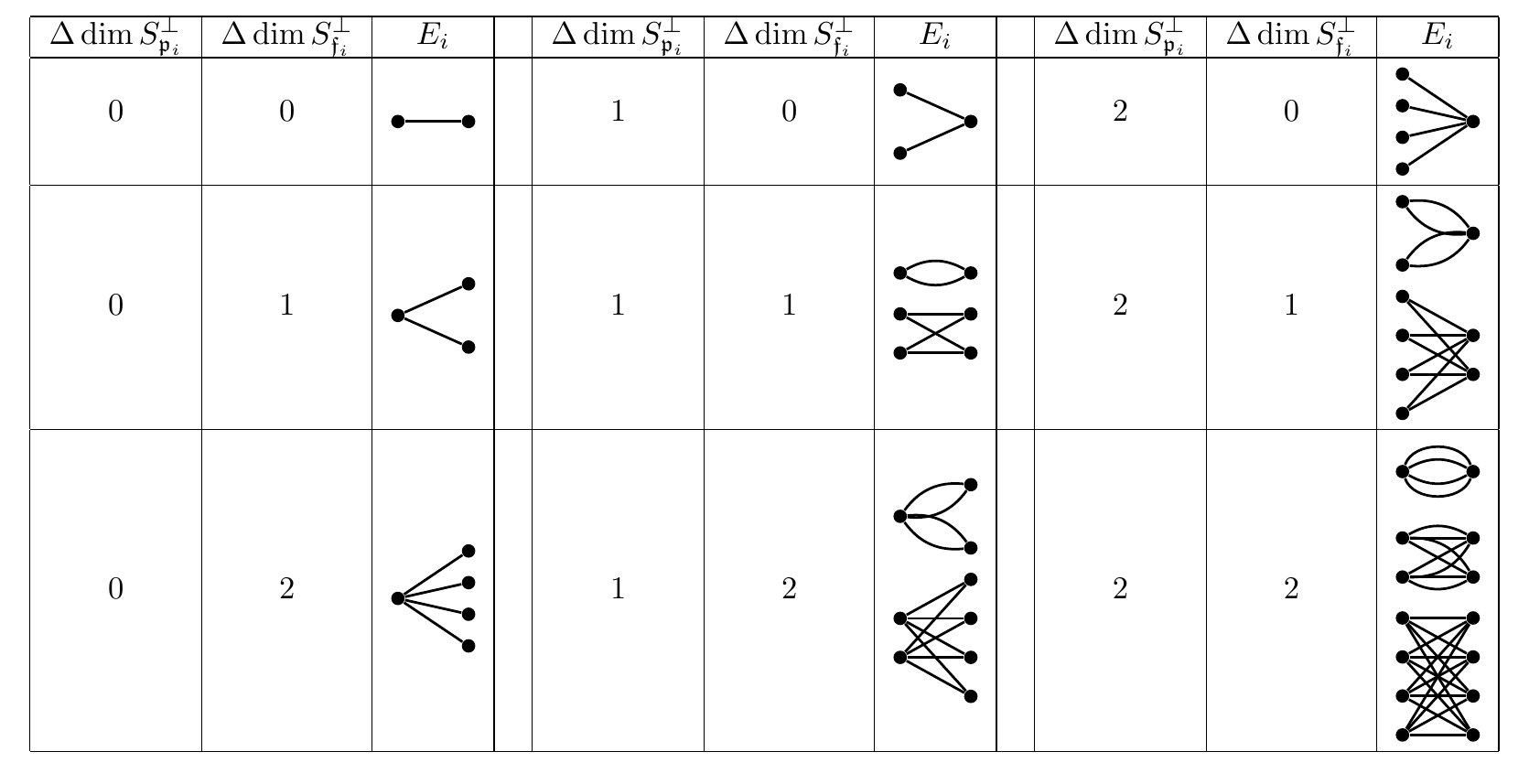}
	\end{tabular}
	\caption{The various possible edge configurations for $p = 2$. Since $\Delta \dim S^\perp_{\mathfrak{p}_i}$ and $\Delta \dim S^\perp_{\mathfrak{f}_i}$ are bounded by one in classical theory, any configurations with a two is unique to the quantum setting.}
	\label{tab:edgeconfigs}
\end{table}

\begin{restatable}[]{theorem}{thmtwelve}
	\label{thm:bipartthm}
	Consider an arbitrary edge $e = (s(e), P_i, t(e)) \in E_i$ and define $\mathcal{I}_i = \{ v \in V_i \mid \exists e^\prime \in E_i, t(e^\prime) = t(e)\}$ and $\mathcal{I}_{i + 1} = \{ v \in V_{i + 1} \mid \exists e^\prime \in E_i, s(e^\prime) = s(e)\}$. Then the vertices of $\mathcal{I}_i$ and $\mathcal{I}_{i + 1}$ form a completely-connected bipartite graph and no other elements of $V_i \backslash \mathcal{I}_i$ or $V_{i + 1} \backslash \mathcal{I}_{i + 1}$ are connected to the vertices in $\mathcal{I}_{i + 1}$ and $\mathcal{I}_i$, respectively. If there exists a parallel edge in $E_i$ then all edges in $E_i$ are parallel with the same number of edges in parallel.
\end{restatable}

At first glance, Equation \eqref{Vi} should be viewed with skepticism. It is common for $\dim S_{\mathfrak{p}_i} = \dim S_{\mathfrak{f}_i} = 0$ (see Example \ref{ex:stabs2}) and if the same holds true for $S^\perp$, then $|V_i| = p^{n + k}$ which is notably greater than the number of total possible syndromes, $p^{n - k}$. Applying the classical formula for $V_i$ or using the logic above, one would expect $\dim V_i = \dim S - \dim S_{\mathfrak{p}_i} - \dim S_{\mathfrak{f}_i}$, and in fact this also works. The inclusion of the logicals therefore always forces $\dim S^\perp_{\mathfrak{p}_i} + \dim S^\perp_{\mathfrak{f}_i} \geq 2k$. The various formulas throughout this subsection may therefore be written with this alternative view of $\dim V_i$, but this leads to a lack of cancellation of terms and an unpleasant factor of $2k$ floating around whose lack of obvious effect on the results requires justification.
\begin{restatable}[]{corollary}{corothirteen}
	For $0 \leq i \leq n$, $\dim S^\perp_{\mathfrak{p}_i} + \dim S^\perp_{\mathfrak{f}_i} \geq 2k$.
\end{restatable}
\noindent The same proof applied to $S$ instead of $S^\perp$ shows $\dim S_{\mathfrak{p}_i} + \dim S_{\mathfrak{f}_i} \geq 0$.

The next result says that for a fixed qudit ordering the choice of stabilizer or logical generators is irrelevant.
\begin{restatable}[]{lemma}{lemmafourteen}
	\label{lem:lemmafourteen}
	Any two minimal trellises for the same stabilizer code are isomorphic.
\end{restatable}

\subsection{Trellis-Oriented Form}\label{sec:TOF}
The term ``trellis-oriented" was introduced in passing in an appendix by Forney \cite{forney1988coset} and wasn't thoroughly defined until seven years later by Kschischang and Sorokine \cite{kschischang1995trellis}. Here we adapt the latter approach to the quantum setting, as did \cite{ollivier2006trellises}, but we will stick closer to the classical notation than that of \cite{ollivier2006trellises}. The term ``minimal-span" is preferred by some authors in the literature.

Let $P \in \mathcal{P}(n, p)$ be an arbitrary Pauli string with components $P_i$ at index $i$.
\begin{definition}\label{def:TOFstuff}
	\leavevmode
	\vspace*{-\bigskipamount}\vspace*{-\medskipamount}
	\begin{enumerate}
		\item The \emph{left index} of $P$, $L(P)$, is the smallest index such that $P_i \neq I$ and the \emph{right index} of $P$, $R(P)$, is the largest index such that $P_i \neq I$.
		\item The \emph{span} of $P$ is the index set $\{L(P), L(P) + 1, \hdots, R(P)\}$ and the \emph{span length} of $P$ is the cardinality of the span of $P$, $R(P) - L(P) + 1$. The span of the identity string is defined to be $\{ \, \}$ and the corresponding span length to be $0$.
		\item Say $P$ is \emph{active} at depth $i$ if $i - 1$ and $i$ are in the span of $P$, i.e., $L(P) \leq i - 1$ and $R(P) \geq i$.
	\end{enumerate}
\end{definition}
\begin{definition}[Left-Right Property]
	A set of Pauli strings is said to have the \emph{left-right property} if no two elements with the same left or right index have more than one power of $X$ or $Z$ at these index locations.
\end{definition}

The next definition is intended to be applied to a set of generators and not the entirety of its span.
\begin{definition}[Trellis-Oriented Form]
	A set of Pauli strings is said to be in \emph{trellis-oriented form} (TOF) if the sum of the span lengths of the elements is as small as possible.
\end{definition}
\begin{restatable}[]{proposition}{propeighteen}
	\label{prop:LRTOF}
	A set of Pauli strings is in TOF if and only if it has the left-right property.
\end{restatable}

\begin{example}\label{ex:stabs}
	The generators for the rotated surface and color codes naturally have the left-right property. The canonical form of the stabilizer generators of the $[[5, 1, 3]]$ code as cyclic shifts of $XZZXI$ may be put into TOF as follows:
	\begin{table}[h!]
		\centering
		\begin{tabular}{cccc}
			\begin{tabular}{@{}c@{}}
				$S_1 =$ \stretchit{XZZXI}\\
				$S_2 =$ \stretchit{IXZZX}\\
				$S_3 =$ \stretchit{XIXZZ}\\
				$S_4 =$ \stretchit{ZXIXZ}
			\end{tabular}
			& $\mapsto$
			\begin{tabular}{@{}cc@{}}
				$S_1$ &$=$ \stretchit{XZZXI}\\
				$S_4$ &$=$ \stretchit{ZXIXZ}\\
				$S_2$ &$=$ \stretchit{IXZZX}\\
				$S_1 S_3$ &$=$ \stretchit{IZYYZ}
			\end{tabular}
			& $\mapsto$
			\begin{tabular}{@{}cc@{}}
				$S_1$ &$=$ \stretchit{XZZXI}\\
				$S_1 S_3 S_4$ &$=$ \stretchit{ZYYZI}\\
				$S_2$ &$=$ \stretchit{IXZZX}\\
				$S_1 S_3$ &$=$ \stretchit{IZYYZ}
			\end{tabular}.
		\end{tabular}
	\end{table}\\
	The left and right indices for each row are, in order, $\{ [1, 4], [1, 4], [2, 5], [2, 5]\}$, where we have used the notation $[L(P), R(P)]$. It is customary, and sometimes included in the definition, to rearrange the strings by increasing left index when viewing them in a matrix-like format, similar to a row-reduced matrix, but this is not strictly necessary.\\
\end{example}

The procedure used in the previous example is quite simple: after reducing the left side from top to bottom, the right side is reduced from bottom to top and the upper stabilizer is always replaced instead of the lower. This prevents the multiplication of the stabilizers from changing the left indices. The key to automating this is considering Pauli strings as length-$n$ vectors over $\F_{p^2}$ and using appropriate operations in this field. This should be compared with the complex algorithm given in \cite{xiao2013construction} for obtaining the TOF using the symplectic representation over $\F_p$.

The quantum TOF is more complicated than in classical theory. To gain some intuition for this consider a set of Pauli strings in TOF. A quantum index, either left or right, is of the form $X^a Z^b$, where $a, b \in \F_p$. Given two indices $X^aZ^b$ and $X^{a^\prime}Z^{b^\prime}$ with all exponents nonzero, one can only eliminate the other to produce an identity if and only if they are scalar multiples of each other in $\F_p$, i.e., $(a^\prime, b^\prime) \in \< (a, b) \> \in \mathbb{Z}_p \times \mathbb{Z}_p$. In the more general case, by repeated application of the string to itself, $a$ and $b$ can individually be made to take any value in $\F_p$, i.e., Pauli strings are linear over $\F_p$. In particular, $-a^\prime \in \< a \> = Z_p$ since $p$ is prime, and likewise for $b$. So, one of $a^\prime$ or $b^\prime$ may always be eliminated. An index $X^a$ cannot eliminate an index $Z^b$ and vice versa. To summarize, the possibilities for left or right indices for qubit codes are $\{ \{X, I\}, \{Z, I\}, \{Y, I\}, \{X, Z\}, \{Y, X\}, \{Y, Z\} \}$. This is in stark contrast to the classical case where a matrix can always be put into reduced-row echelon form.

Let $P$ and $P^\prime$ be two Pauli strings with the left-right property and span $[L, R]$, then $PP^\prime$ also has span $[L, R]$. This is often referred to as the \textit{predictable span property}. If $P$ and $P^\prime$ have span $[L, R_1]$ and $[L, R_2]$, respectively, with $R_1 < R_2$, then $PP^\prime$ has span $[L, R_2]$, and likewise for $[L_1, R]$ and $[L_2, R]$. These two cases are not possible classically when there is only one symbol alphabet, necessitating different proof strategies.\\

\begin{example}
	The first two stabilizers for the TOF of the following $[[8, 2, 3]]$ code \cite{codetables} have the same left index but different right indices:
	\begin{table}[h!]
		\centering
		\begin{tabular}{@{}cc@{}}
			\stretchit{YZZZXIII}\\
			\stretchit{ZIIXIYII}\\
			\stretchit{IXZZYXYI}\\
			\stretchit{IZYIYYXI}\\
			\stretchit{IIXYXZYZ}\\
			\stretchit{IIZIIIYX}
		\end{tabular}.
	\end{table}
\end{example}

If we start with a generating set for $\mathcal{A}$, the resulting TOF still generates $\mathcal{A}$. Each generator has order $p$ and generates $p - 1$ nonidentity strings with the same span. Suppose two generators $P$ and $P^\prime$ have spans $[L_1, R_1]$ and $[L_2, R_2]$, respectively, and both $R_1, R_2 \leq i \in \{1, \hdots, n\}$. Then all products of the $p$ elements generated by $P$ and the $p$ elements generated by $P^\prime$ also have right index less than or equal to $i$. If there are no other generators with right index less than or equal to $i$, then these are all of the elements of $\mathcal{A}_{\mathfrak{p}_i}$. Extending this argument shows the following.
\begin{proposition}\label{prop:dimargs}
	Let $\mathcal{A}$ be a set of Pauli strings with generators $\mathcal{A}_i$ in TOF. Then
	\begin{align}
		\dim \mathcal{A}_{\mathfrak{p}_i} &= | \left\{\mathcal{A}_i \mid R(\mathcal{A}_i) \leq i \right\}|\\
		\dim \mathcal{A}_{\mathfrak{f}_i} &= |\left\{\mathcal{A}_i \mid L(\mathcal{A}_i) \geq i + 1 \right\}|.
	\end{align}
\end{proposition}

Unfortunately, the previous argument also shows that the TOF is not unique since a generator may be replaced by any of its $p - 1$ multiples. (In Example \ref{ex:stabs}, replace $S_1 S_3$ with, for example, $S_1 S_2 S_3$). However, in light in of Theorem \ref{thm:quantstate} and Corollary \ref{cor:inoutdegs}, the entire structure of the trellis may be read-off directly from the generators of $S^\perp$ when put into TOF. Proposition \ref{prop:dimargs} also provides trivial proofs of results such as Lemma \ref{lem:deltadim} by merely counting the number of possible new generators obtained when shifting indices in TOF.\\

\begin{example}\label{ex:stabs2}
	The stabilizers of the $[[5, 1, 3]]$ code appear in Example \ref{ex:stabs}. The generators of $S^\perp$ have a different TOF,
	\begin{table}[h!]
		\centering
		\begin{tabular}{@{}cc@{}}
			\stretchit{XYXII}\\
			\stretchit{ZXZII}\\
			\stretchit{IXYXI}\\
			\stretchit{IZXZI}\\
			\stretchit{IIXYX}\\
			\stretchit{IIZXZ}
		\end{tabular}.
	\end{table}
	
	\noindent Applying Proposition \ref{prop:dimargs} to both sets we get Table \ref{tab:ex4tab}.
	\begin{table}[h!]
		\centering
		\begin{tabular}{|c|c|c|c|c|c|c|c|c|}
			\hline
			& \multicolumn{8}{c|}{$[[5, 1, 3]]$}\\
			\hline
			$i$ & $\dim S_{\mathfrak{p}_i}$ & $\dim S_{\mathfrak{f}_i}$ & $\dim S^\perp_{\mathfrak{p}_i}$ & $\dim S^\perp_{\mathfrak{f}_i}$ & $|V_i|$ & $|E_i|$ & in & out\\
			\hline
			0 & 0 & 4 & 0 & 6 & 1 & $-$ & $-$ & 4\\
			1 & 0 & 2 & 0 & 4 & 4 & 4 & 1 & 4\\
			2 & 0 & 0 & 0 & 2 & 16 & 16 & 1 & 4\\
			3 & 0 & 0 & 2 & 0 & 16 & 64 & 4 & 1\\
			4 & 2 & 0 & 4 & 0 & 4 & 16 & 4 & 1\\
			5 & 4 & 0 & 6 & 0 & 1 & 4 & 4 & $-$\\ 
			\hline
		\end{tabular}
		\caption{}
		\label{tab:ex4tab}
	\end{table}
	
	\noindent These may be compared with the trellis diagram for the code given in \cite{xiao2013construction}.
\end{example}

\noindent {\bf Remark:} As discussed in the previous subsection, the structure of the trellis is invariant with respect to a change of stabilizers. One may compute a TOF of a set of generators, apply Proposition \ref{prop:dimargs} to determine $|V_i|$, $|E_i|$, etc, and then construct the trellis with respect to the original set of stabilizers which, for example, may have been more beneficial experimentally.

\subsection{The Viterbi Algorithm}\label{sec:ViterbiProps}
It is worth clarifying exactly which decoding problem the trellis solves. Let $E \in \mathcal{P}(n, p)$ be a Pauli error acting on a quantum stabilizer code $\mathcal{Q}$ with stabilizer group $S = \< S_1, \hdots, S_{n - k} \>$. The syndrome, $s = (s_1, \hdots, s_{n - k})$, is the ordered tuple of commutation relations of $E$ with the stabilizer elements, $s_j = \< S_j, E \>$, where the inner product is chosen appropriately. Logical and stabilizer operators commute with all elements of the stabilizer and therefore have trivial syndrome. It follows that an error can be decomposed as
\begin{equation}\label{edecomp}
	E = L \tilde{S} T,
\end{equation}
where $L \in C_{\mathcal{P}(n, p)} (S) \backslash S$ is a logical operator, $\tilde{S} \in S$, and $T \in \mathcal{P}(n, p) \backslash C_{\mathcal{P}(n, p)} (S)$ is a pure error. One can always write down a pure error $T_0$ given a syndrome $s$, so
\begin{equation}
	\Pr (E \mid s) = \Pr \left( L, \tilde{S} \mid s, T_0 \right) = \frac{\Pr\left(L, \tilde{S}, s, T_0\right)}{\di \sum_{L, \tilde{S}} \Pr\left(L, \tilde{S}, s, T_0\right)}.
\end{equation}
Since elements of the stabilizer do not affect the code state by definition, only $L$ needs to be determined. The quantum (degenerate) maximum \textit{a posteriori} (MAP) decoding problem is therefore
\begin{equation}\label{degenprob}
	\hat{L}_s = \argmax_{L \in C_{\mathcal{P}(n, p)} (S) / S} \Pr (L \mid s),
\end{equation}
where $\Pr (L \mid s) = \sum_{\tilde{S} \in S} \Pr \left( L, \tilde{S} \mid s\right)$. It is generally not known how to develop decoding efficient algorithms that take the degeneracy of the logical operators into account; instead, \eqref{degenprob} is typically relaxed to
\begin{equation}\label{nodegenprob}
	\hat{T}_s = \argmax_{T \in \mathcal{P}(n, p)} \Pr( T \mid s),
\end{equation}
which is equivalent to the classical decoding problem.

After measuring a syndrome, the trellis encodes all combinations \eqref{edecomp} for a given $T$ and solves
\begin{align}
	\hat{E}_s &= \argmax_{E \in \mathcal{P}(n, p)} \Pr( E = L \tilde{S} T \mid s) \nonumber\\
		&= \argmin_{E \in \mathcal{P}(n, p)} \wt (E = L \tilde{S} T \mid s). \label{trellisprob}
\end{align}
Note that this is not Hamming weight nor minimum-weight perfect matching; it is the optimal minimum weight decoder but not a maximal likelihood decoder. While this incorporates all of the information of the logical operators, it is technically not the degenerate decoder \eqref{degenprob}. One however can construct a ``degenerate" trellis to solve \eqref{degenprob}. Pelchat and Poulin give the procedure for this in their investigation of trellis decoding for quantum convolutional codes \cite{pelchat2013degenerate}.

\begin{restatable}[]{proposition}{proptwenty}
	The Viterbi algorithm correctly computes the minimum weight correction \eqref{trellisprob} for uncorrelated, i.i.d.~noise models.
\end{restatable}

As mentioned in Section \ref{sec:construction}, the full decoding procedure consists of finding a pure error, shifting the trellis, and then applying the Viterbi algorithm. The shifting procedure requires at most $n$ symplectic inner products followed by $|V| + |E|$ shifts. Fortunately, this operation is embarrassingly parallel, and the runtime may be significantly reduced depending on the implementation and available resources. For a potentially large computational speedup, one may skip shifting the vertices, as they play no role in the Viterbi algorithm. It is clear from Algorithm \ref{alg:Viterbi} that every edge incurs a single addition followed by incoming-degree comparisons for every vertex. For large codes, the vertices of $V_i$ may be processed independently in parallel. The Viterbi algorithm may also be run simultaneously in a ``forward pass" from $V_0$ to some $V_i$ and a ``backwards pass" from $V_n$ to $V_{i + 1}$. The choice of the optimal splitting depends highly on the left-right balance of the trellis. Further optimizations such as those based on the coset structure of the code (vertical symmetry) are generally code specific.
\begin{restatable}[Theorem 2.10 \cite{mceliece1996bcjr}]{theorem}{thmtwentyone}
	\label{prop:Viterbicost}
	The Viterbi algorithm requires $\Theta (| E |)$ arithmetic operations.
\end{restatable}

Theorem \ref{thm:quantstate} shows the syndrome trellis minimizes $|V_i|$ and $|E_i|$, but in light of the previous proof it would be useful to show that it also minimizes the quantity $|E| - |V| + 1$. That is, the trellis minimizes the number of algorithmic operations. Classically, Muder first showed that the analogous trellis minimizes $|V|$ \cite{muder1988minimal}, McEliece showed it minimizes $|E|$ \cite{mceliece1994viterbi, mceliece1996bcjr}, and Vardy and Kschischang showed it minimizes $|E| - |V| + 1$ \cite{vardy1996proof}. The proof of this follows for the quantum case with minor modifications (Proposition \ref{prop:ArunInduct}) and we include it here in Theorem \ref{thm:EV1min} for completeness. We first establish a geometric interpretation of this quantity.

The quantity $|E|$ may be further resolved using the edge configurations in Table \ref{tab:edgeconfigs}. While these are drawn for the special case $p = 2$, here we will use them to represent the corresponding diagrams for higher $p$. One may check that substituting $p = 2$ produces the correct coefficients in the formula below. We have,
\begin{align}\label{edgeresolutionA}
	|E| &= 1 \cdot \# \left( \Nzz \right) + p \cdot \# \left( \Noz \right) + p^2 \cdot \# \left( \Ntz \right) + p \cdot \# \left( \Nzo \right) + p \cdot \# \left( \Nooa \right) \nonumber\\
		&+ p^2 \cdot \# \left( \Noob \right) + p^2 \cdot \# \left( \Ntoa \right) + p^3 \cdot \# \left( \Ntob \right) + p^2 \cdot \# \left( \Nzt \right) + p^2 \cdot \# \left( \Nota \right)\\
		&+ p^3 \cdot \# \left( \Notb \right) + p^2 \cdot \# \left( \Ntta \right) + p^3 \cdot \# \left( \Nttb \right) + p^4 \cdot \# \left( \Nttc \right). \nonumber
\end{align}
Similarly, ignoring $V_0$ and counting the right hand sides,
\begin{align}\label{vertexresolutionA}
	|V| - 1 &= 1 \cdot \# \left( \Nzz \right) + 1 \cdot \# \left( \Noz \right) + 1 \cdot \# \left( \Ntz \right) + p \cdot \# \left( \Nzo \right) + 1 \cdot \# \left( \Nooa \right) \nonumber\\
		&+ p \cdot \# \left( \Noob \right) + p \cdot \# \left( \Ntoa \right) + p^2 \cdot \# \left( \Ntob \right) + p^2 \cdot \# \left( \Nzt \right) + 1 \cdot \# \left( \Nota \right)\\
		&+ p \cdot \# \left( \Notb \right) + 1 \cdot \# \left( \Ntta \right) + p \cdot \# \left( \Nttb \right) + p^2 \cdot \# \left( \Nttc \right). \nonumber
\end{align}

\begin{definition}
	A vertex $v$ with $\deg_{\text{out}}(v) > 1$ is called an \emph{expansion} and a \emph{merger} if  $\deg_{\text{in}}(v) > 1$.
\end{definition}

The number of mergers is graphically given by
\begin{align}\label{mergers}
	\mathcal{M} &= (p - 1) \cdot \left( \Noz \right) + (p^2 - 1) \cdot \left( \Ntz \right) + (p - 1) \cdot \left( \Nooa \right) + p (p - 1) \cdot \left( \Noob \right) \nonumber\\
	&+ p (p - 1) \cdot \left( \Ntoa \right) + p^2 (p - 1) \cdot \left( \Ntob \right) + (p^2 - 1) \cdot \left( \Nota \right) + p (p^2 - 1) \cdot \left( \Notb \right)\\
	&+ (p^2 - 1) \cdot \left( \Ntta \right) +  p (p^2 - 1) \cdot \left( \Nttb \right) +  p^2 (p^2 - 1) \cdot \left( \Nttc \right). \nonumber
\end{align}
Subtracting \eqref{vertexresolutionA} from \eqref{edgeresolutionA} shows that this is exactly $|E| - |V| + 1$ \cite{kiely1996trellis}, as expected from the proof of Proposition \ref{prop:Viterbicost}. To match the classical proof in \cite{vardy1996proof} we switch to from mergers to expansions $\mathcal{E}$. Repeating the arguments of the proof with $\deg_{\text{in}}(v)$ in place of $\deg_{\text{out}}(v)$ shows $\mathcal{E} = |E| - |V| + 1$ as well. This is intuitively clear from the fact that the trellis both starts and ends with a single vertex. Writing a similar expression to \eqref{mergers} for expansions and setting $\mathcal{E} = \mathcal{M}$, the diagrams which share a left-right symmetry cancel leaving only an equality of total asymmetric edge configurations. 

Returning to our goal, it is clear from Section \ref{sec:theory} that the number of paths from $V_0$ to any vertex in $V_i$ is $p^{\dim S^\perp_{\mathfrak{p}_i}}$ and the total number of paths from $V_0$ to all of $V_i$ is  $p^{\dim S^\perp_{|_{\mathfrak{p}_i}}}$. Furthermore, for any definition of a trellis, the number of paths is at most $p^{\dim S^\perp_{\mathfrak{p}_i}}$ and the total number of paths is at least $p^{\dim S^\perp_{|_{\mathfrak{p}_i}}}$.
\begin{restatable}[Lemma 5 \cite{vardy1996proof}]{lemma}{lemmatwentythree}
	Let $\mathscr{P}_i(v)$ denote the number of paths from $V_0$ to $v \in V_i$ and $\di \mathscr{P}_i = \sum_{v \in V_i} \mathscr{P}_i(v)$ be the number of paths from $V_0$ to all vertices in $V_i$. Then
	\begin{equation*}
		\mathscr{P}_i = 1 + \sum_{j = 0}^{i - 1} \sum_{v \in V_j} \mathscr{P}_j(v) \left( \deg_\mathrm{out}(v) - 1 \right).
	\end{equation*}
\end{restatable}

The next result is used but not proved in \cite{vardy1996proof}. We present the proof for the quantum case; the proof of their original expression follows from this by restricting to $\dim S^\perp_{\mathfrak{p}_i} - \dim S^\perp_{\mathfrak{p}_{i + 1}} \leq 1$. Recall that $\{\dim S^\perp_{\mathfrak{p}_i}\}$ is an increasing sequence whose value changes when a new right index is encountered in the TOF. Let $\{R_i\}$ be the locations of the right indices. The idea behind the following proof is to divide the index set $\{0, \hdots, n\}$ into intervals of constant past dimension and then make arguments about the locations $\{R_i\}$. To be explicit, in this notation $\dim S^\perp_{\mathfrak{p}_{R_i}} < \dim S^\perp_{\mathfrak{p}_{R_{i + 1}}}$ and $\dim S^\perp_{\mathfrak{p}_i} = \dim S^\perp_{\mathfrak{p}_{R_a}}$ if and only if $R_a \leq i < R_{a + 1}$. The proof is purely algebraic.
\begin{restatable}[]{proposition}{proptwentyfour}
	\label{prop:ArunInduct}
	Let unprimed quantities be with respect to the syndrome trellis and primed quantities be with respect to any other trellis for the same code. Denote by
	\begin{equation*}
		\mathcal{E}_j =  \sum_{v \in V_j} \left( \deg_\mathrm{out}(v) - 1 \right)
	\end{equation*}
	the number of expansions at depth $j$ and likewise for the primes, and define
	\begin{equation*}
		\Delta_i \coloneqq \sum_{j = 0}^{i - 1} p^{\dim S^\perp_{\mathfrak{p}_j}} \left( \mathcal{E}^\prime_j - \mathcal{E}_j \right).
	\end{equation*}
	Then for $\di R_{\kappa} < i \leq R_{\kappa +1}$,
	\begin{equation*}
		\sum_{j = 0}^{i - 1} p^{\dim S^\perp_{\mathfrak{p}_{i - 1}}} \left( \mathcal{E}^\prime_j - \mathcal{E}_j \right) = \Delta_i + \sum_{a = 1}^\kappa p_{\kappa, a} (p_a - 1) \Delta_{R_a}
	\end{equation*}
	where $\di p_a = p^{\dim S^\perp_{\mathfrak{p}_{R_a}} - \dim S^\perp_{\mathfrak{p}_{R_{a-1}}}}$ and $\di p_{\kappa,a} = p^{\dim S^\perp_{\mathfrak{p}_{R_{\kappa}}} - \dim S^\perp_{\mathfrak{p}_{R_a}}}$.
\end{restatable}

\begin{restatable}[Theorem 6 \cite{vardy1996proof}]{theorem}{thmtwentyfive}
	\label{thm:EV1min}
	Let $|E|$ and $|V|$ be with respect to the syndrome trellis and $|E^\prime|$ and $|V^\prime|$ be with respect to any other definition of a trellis for the same code. Then $|E^\prime| - |V^\prime| + 1 \geq |E| - |V| + 1$.
\end{restatable}

Since the syndrome trellis is a minimal representation of the code that is invariant with respect to the generators, following Theorems \ref{prop:Viterbicost} and \ref{thm:EV1min}, it is often argued in the classical literature that the quantity $|E|$ should be regarded as a fundamental description of how hard it is to decode a given code, rivaling in importance with $n$, $k$, and $d$. Adopting this philosophy shows, for example, that the color codes are fundamentally more difficult to decode than the rotated surface code of the same distance without invoking projections or hypergraph matching. We provide quantitative data on this in Example \ref{ex:edgecounts} in the next section after introducing the CSS splitting of trellises.

Proposition \ref{prop:iso} shows that the trellis essentially functions as a compact lookup table for $S^\perp$. The key to the efficiency of the Viterbi algorithm is its ability to make decisions about all of the elements of this set without computing every path individually. For a given $v \in V_i$, edge sharing in the trellis enables the algorithm to simultaneously check every element in $S^\perp_{|_{\mathfrak{p}_i}}$ ending at $v$. When it discards all outgoing edges for $v$, a total of $|S^\perp_{\mathfrak{f}_i}|$ elements are eliminated from $S^\perp$ for further consideration (see Figure \ref{fig:trelliscosets}).

Assuming $|E|$ scales at least cubically in $n$, the cost of the Viterbi algorithm also dominates that of finding the pure error $T$ given the syndrome. Finding a pure error with the given syndrome is not difficult. The trellis decoder may use any valid $T$ since the minimum weight solution is an element of the set $LST$ enumerated by the paths of the trellis and will hence be found by the Viterbi algorithm. Given a potentially high-weight pure error, decoding may therefore be interpreted as a refinement process to the minimum-weight solution. Here, we use pseudoinverses for the syndromes to find $T$.

\section{CSS Codes}\label{sec:CSS}
Calderbank-Shor-Steane (CSS) codes have the property that the generators of $S$ ($S^\perp$) split into those with only powers of $X$ or only powers of $Z$. It is common in quantum error correction to decode each set of generators independently then combine the results into a single correction. This has the advantage of reducing decoding complexity and enabling parallelization at the expense of ignoring potential $X$-$Z$ correlations. The same technique can, of course, be used with trellises. As a first example, consider the trellis diagrams in Appendix \ref{sec:knowncodes}. Figure \ref{fig:Surf17full} shows the trellis diagram of the distance three rotated surface code and Figures \ref{fig:Surf17Xmin} and \ref{fig:Surf17Zmin} show the effect on the trellis of considering the $X$- and $Z$-stabilizers separately. Figure \ref{fig:cc666d3full} shows the same for the distance three color code. The reduction in trellis and decoding complexity is immediately apparent even for these small examples.

An immediate consequence of decoding the $X$- and $Z$-stabilizers separately is that moving from left-to-right in the TOF, the past (future) can only increase (decrease) by a maximum of one stabilizer generator: $\dim S^\perp_{\mathfrak{f}_i} - \dim S^\perp_{\mathfrak{f}_{i + 1}} \leq 1$ and $\dim S^\perp_{\mathfrak{p}_i} - \dim S^\perp_{\mathfrak{p}_{i - 1}} \leq 1$.  Thus, in contrast to Corollary \ref{cor:inoutdegs}, none of the more complicated edge configurations in Table \ref{tab:edgeconfigs} with a two are allowed and the edge configurations with the highest contribution of edges do not occur, i.e.,
\begin{equation}\label{vertexresolutionB}
	|V_{X/Z}| - 1 = 1 \cdot \# \left( \Nzz \right) + 1 \cdot \# \left( \Noz \right) + p \cdot \# \left( \Nzo \right) + 1 \cdot \# \left( \Nooa \right) + p \cdot \# \left( \Noob \right),
\end{equation}
and
\begin{equation}\label{edgeresolutionB}
	|E_{X/Z}| = 1 \cdot \# \left( \Nzz \right) + p \cdot \# \left( \Noz \right) + p \cdot \# \left( \Nzo \right) + p \cdot \# \left( \Nooa \right) + p^2 \cdot \# \left( \Noob \right).
\end{equation}
Thus, stabilizer CSS codes which have unit dimension change in the TOF of the $X$- and $Z$- stabilizers at the same index $i$ will have higher dimension changes and hence more edges and will therefore be more difficult to decode than those which do not. It follows that self-dual codes are more difficult to decode than non-self-dual codes of the same parameters. Note that applying $\mathcal{E} = \mathcal{M}$ to CSS trellises gives \cite{kiely1996trellis}
\begin{equation*}
	\# \left( \Nzo \right) = \# \left( \Noz \right).
\end{equation*}

Let $S$ be a stabilizer code given by the CSS construction with $C_1 = [n, k_1, d_1]_p$, $C_2 = [n, k_2, d_2]_p$, and $C_2^T \subseteq C_1$. Denote by $G_i$ and $H_i$ the generator and parity-check matrices for $C_i$, respectively. The stabilizers of the quantum code are given, in symplectic form $(X \mid Z)$, by
\begin{equation*}
	\begin{pmatrix}
		H_1 & 0\\
		0 & H_2
	\end{pmatrix}.
\end{equation*}
The set $S^\perp$ is generated by the corresponding normalizer matrix \cite{li2008standard, wilde2009logical, grassl2011variations}
\begin{equation*}
	\begin{pmatrix}
		0 & G_1\\
		G_2 & 0
	\end{pmatrix}.
\end{equation*}
The trellis for the $X$- ($Z$-) stabilizers only in a CSS code is therefore determined by the classical trellis whose paths are in one-to-one correspondence with the codewords of $G_1$ ($G_2$). Edge labels for CSS trellises should be restricted to $I$ and $Z$ or $I$ and $X$ with weights $-\log \Pr(I) -\log \Pr(X)$ and $-\log \Pr(Y) - \log \Pr(Z)$ for $X$- and similarly for the $Z$-stabilizers.

It is well-known in classical trellis theory that the dimension of $V_i$ is equal to the dimension of the corresponding $V_i$ for the dual code. Since one has $p^{n - k}$ possible syndromes and the other $p^k$, we immediately get the following.
\begin{corollary}[Wolf Bound For CSS Codes \cite{wolf1978efficient}]
	Let $C_i$ be as above. Then for the $X$-stabilizer trellis, $|V_{X, i}| \leq p^{\min\{k_1, n - k_1\}}$, and for the $Z$-stabilizer trellis, $|V_{Z, i}| \leq p^{\min\{k_2, n - k_2\}}$.
\end{corollary}

\begin{example}
	The distance three color code is equivalent to the well-known $[[7, 1, 3]]$ Steane code. This is a CSS code constructed with the $[7, 4, 3]$ binary Hamming code and its dual. Hence, $|V_{X/Z, i}| \leq 2^{\min\{3, 4\}} = 8$, which agrees with Figure \ref{fig:cc666d3full}.
\end{example}

It is perhaps not surprising that the full trellis for a CSS code turns out to be a product of the two trellises of its component codes. For instance, the active generators of the full code are the active generators of the $X$-code and the active generators of the $Z$-code. Thus,
\begin{equation*}
	|V_i| = p^{\dim S_{X, \mathfrak{a}} + \dim S_{Z, \mathfrak{a}}} = p^{\dim S_{X, \mathfrak{a}}} p^{\dim S_{Z, \mathfrak{a}}} = |V_{X, i} | | V_{Z, i} |.
\end{equation*}
\begin{definition}[Trellis Product \cite{kschischang1995trellis}]\label{def:tp}
	The \emph{trellis product} of two trellises with vertices $V$, $V^\prime$ and edges $E$, $E^\prime$, respectively, is denoted by $\times_i$ and has vertices
	\begin{equation*}
		V \times_i V^\prime = \bigcup_{i = 0}^n V_i \times V_i^\prime
	\end{equation*}
	and edges
	\begin{align*}
		E \times_i E^\prime &= \bigcup_{i = 0}^n E_i \times E_i^\prime,\\
		&= \{ ((v_i, v^\prime_i), P P^\prime,  (v_{i + 1}, v^\prime_{i + 1})) \mid (v_i, P, v_{i + 1}) \in E_i, (v^\prime_i, P^\prime, v^\prime_{i + 1}) \in E^\prime_i \}.
	\end{align*}
\end{definition}

\noindent {\bf Remark:} This is sometimes called the \emph{Shannon product} \cite{sidorenko1996minimal} for historical reasons after Claude Shannon who first described the product of two channels operating at the same time and is distinct from the numerous other standard graph products in the mathematical literature, including those also denoted by $\times$. It was pointed out in \cite{kschischang1995trellis} that the trellis product may or may not give a minimal trellis. For example, the trellis product of Figure \ref{fig:trellisex} with itself gives an improper trellis with $\deg_\text{out}(v_0) = 16$ over a symbol alphabet of size four. In all cases considered in this work, the trellis products will be minimal.\\

\begin{example}
	Using shorthand $(\Delta \dim S^\perp_{\mathfrak{p}_i}, \Delta \dim S^\perp_{\mathfrak{f}_i}) = (m, n)$ for edge configurations in Table \ref{tab:edgeconfigs}, it may be checked that $(m, n) \times_i (m^\prime, n^\prime) = (m + m^\prime, n + n^\prime)$. Such configurations occur in the full trellis as graph products of lower dimensional configurations in the CSS trellises. The vertices of Figure \ref{fig:Surf17full} are ordered by the trellis product of Figure \ref{fig:Surf17Xmin} by Figure \ref{fig:Surf17Zmin} and the vertices of Figure \ref{fig:cc666d3full} (b) by the trellis product of (a) with itself.
\end{example}

\begin{lemma}\label{lem:CSStp}
	Let $S$ be a CSS code with $X$-stabilizers, $S_X$, and $Z$-stabilizers, $S_Z$. Let $V_X$, $E_X$, $\deg_{\text{in}, X}$, and $\deg_{\text{out}, X}$ denote the appropriate quantities for the trellis of $S_X$, and likewise for $S_Z$. Then the trellis for $S$ is given by the trellis product of the trellises for $S_X$ and $S_Z$. Furthermore,
	\begin{enumerate}
		\item $|V_i| = |V_{X, i} | | V_{Z, i} |$,
		\item $|E_i| = |E_{X, i} | | E_{Z, i} |$,
		\item $\deg_{\text{in}, i} = |\deg_{\text{in}, X, i} | | \deg_{\text{in}, Z, i} |$,
		\item $\deg_{\text{out}, i} = |\deg_{\text{out}, X, i} | | \deg_{\text{out}, Z, i} |$.
	\end{enumerate}
\end{lemma}

The proof of the enumerated parts of this lemma follow from Definition \ref{def:tp} and the larger claim is a restatement of the well-known fact that $X$- and $Z$-errors may be decoded independently for CSS codes. A rigorous proof is fairly trivial and we leave it to the reader. Theses results may be demonstrated with the diagrams in Appendix \ref{sec:knowncodes}. Note that the Space Theorems \ref{thm:quantstate} do not hold for CSS splittings. Following the proof of the theorem, the number of vertices is isomorphic to $|S^\perp / \ker \sigma_i |$ but now the kernel with respect to $S_X$ includes $S^\perp_X$ as well as the past and future elements of $S^\perp_Z$:
\begin{equation}
	| V_{X, i} | = p^{\dim S^\perp - \dim S^\perp_{Z, \mathfrak{p}_i} - \dim S^\perp_{Z, \mathfrak{f}_i} - \dim S^\perp_X} = p^{\dim S^\perp_Z - \dim S^\perp_{Z, \mathfrak{p}_i} - \dim S^\perp_{Z, \mathfrak{f}_i}}.
\end{equation}
Similar equations hold for $|V_{Z, i} |$ and the edges. The same idea holds for counting vertices and edges for an arbitrary subset of generators of $S$. This split trellis idea may be used in more generality. The proof of the following with the clear associativity of the trellis product implies the main part of the lemma above.
\begin{restatable}[]{theorem}{thmIVfour}
	\label{thm:split}
	For $S = \< s_1, \hdots, s_{n - k} \>$, the minimal trellis for $S$ is given by the trellis product of the minimal trellises for each $s_i$.
\end{restatable}

This splitting idea may be used to decode in a similar manner to the CSS codes where each grouping of syndrome bits are corrected independently and then combined. The success of this approach is tied to finding an appropriate grouping of generators with respect to which the logical operators split ``nicely" such that the combined corrections do not introduce a logical error not returned by the individual trellises.

If we adopt the idea that $|E|$ represents a fundamental description of how difficult it is to decode, then the two quantities of interest to compare are $|E_X| + |E_Z|$ and $|E|$, where the latter is with respect to the full stabilizer code. In particular, for a self-dual code $2|E_{X/Z}| \leq |E_{X/Z}|^2$. It follows that splitting the decoding of a self-dual CSS code is square-root easier than decoding the full code. A rather trivial bound on the more general case is immediate.
\begin{restatable}[]{proposition}{propIVfive}
	Let $|E|$ be the total number of edges in a trellis diagram for a CSS code whose $X$- and $Z$-codes have trellises with total number of edges $|E_X|$ and $|E_Z|$, respectively. Then, 
	\begin{equation}\label{ineq}
		|E_X| + |E_Z| \leq |E| \leq |E_X| |E_Z| - p^2 n (n - 1)
	\end{equation}
\end{restatable}
The right-hand side of \eqref{ineq} is tight for the case that there are a minimal number of edges at each depth such that $|E_i| = 4$ for all $i$. This, of course, is rare, usually only occurring around the left and right ends of the trellis, and the more the trellis deviates from this value the worst this upper bound becomes.

It would be more useful to obtain bounds on the number of edges in any of the diagrams with only the values $n$, $k$, and $d$. We leave this for future work and simply note that several classical trellis bounds could potentially be exploited for CSS codes. An easy one in particular worth mentioning is the case of classical maximum distance separable (MDS) codes where it is well-known that the sequence $\{|V_i|\}$ has the following pattern \cite{wolf1978efficient}
\begin{equation*}
	\{1, p, p^2, \hdots, p^{\min\{k, n- k\}}, p^{\min\{k, n - k\}}, \hdots, p^2, p, 1\}.
\end{equation*}
Combining this with the fact that $E_i \geq V_i$, removing the contribution from $V_0$ and summing the geometric series we get
\begin{equation}\label{geopsum}
	|E_{X/Z}| \geq \frac{\di 2p^{\min\{k, n - k\}} - p - 1}{p - 1},
\end{equation}
which is tight for (unrealistic) case that the outgoing degree at every $i$ is one. Replacing $p$ with $p^2$ in \eqref{geopsum} provides the corresponding bound on $|E|$ for a self-dual CSS code. Combining the two loose bounds \eqref{ineq} and \eqref{geopsum} does not appear useful.

\begin{example}\label{ex:edgecounts}
	Returning to the proposal that $|E|$ represents a fundamental parameter for the code, vertex and edge counts for the 4.8.8 and 6.6.6 color codes, rotated surface codes, and their CSS splittings are given in the following tables. The non-CSS XZZX surface codes have the same values as the full rotated surface codes. A standard qubit numbering order was applied to the surface codes but the qubit numbering for the color codes was assigned greedily to minimize the trellis. The difference in values between distances may become more consistent given more consistent numbering schemes. Note that the total vertex and edge counts for the CSS trellises are the sum of the $X$- and $Z$-counts.
	\begin{table}[h!]
		\centering
		\begin{tabular}{|l|c|c|c|c|c|c|c|c|c|c|}
			\hline
			& 3 & 5 & 7 & 9 & 11 & 13 & 15 & 17 & 19 & 21\\
			\hline
			4.8.8 & 122 & 4042 & 83402 & 2126282 & 8673370 & 108195242 & 2074018922 & 36433758250 & 618938698730 & 10475888643050\\
			4.8.8 $X/Z$ & 26 & 230 & 1382 & 8198 & 20058 & 66710 & 327278 & 1498758 & 6610590 & 28827294\\
			6.6.6 & 122  & 2522 & 49802 & 496010 & 4719242 & 54470282 & 604814042 & 5357974490 & 40005091802 & 326581575962\\
			6.6.6 $X/Z$ & 26 & 170 & 974 & 3966 & 14414 & 48878 & 174170 & 617322 & 1728842 & 5102498\\
			RSurf & 74 & 1098 & 10058 & 73034 & 464202 & 2708810 & 14898506 & 78468426 & 399856970 & 1985303882\\ 
			RSurf $X$ & 22 & 118 & 470 & 1590 & 4854 & 13814 & 37366 & 97270 & 245750 & 606198\\
			RSurf $Z$ & 30 & 198 & 854 & 2998 & 9334 & 26870 & 73206 & 191478 & 485366 & 1200118\\
			\hline
		\end{tabular}
		\caption{Vertex counts by distance for common codes.}
		\label{tab:vcounts}
	\end{table}
	\begin{table}[h!]
		\centering
		\begin{tabular}{|l|c|c|c|c|c|c|c|c|c|c|}
			\hline
			& 3 & 5 & 7 & 9 & 11 & 13 & 15 & 17 & 19 & 21\\
			\hline
			4.8.8 & 232 & 7080 & 143272 & 3559336 & 14506536 & 175628968 & 3358695592 & 60870993576 & 1031533604008 & 17421128805544\\
			4.8.8 $X/Z$ & 36 & 316 & 1884 & 11100 & 27084 & 89628 & 439100 & 2020188 & 8901500 & 38785916\\
			6.6.6 & 232 & 4648 & 89512 & 832936 & 7708072 & 89669032 & 1007967784 & 8733820456 & 64652391976 & 532838443048\\
			6.6.6 $X/Z$ & 36 & 236 & 1340 & 5372 & 19388 & 65852 & 234956 & 828556 & 2316044 & 6847020\\
			RSurf & 152 & 2152 & 19688 & 143336 & 913384 & 5341160 & 29425640 & 155189224 & 791674856 & 3934257128\\ 
			RSurf $X$ & 30 & 172 & 700 & 2388 & 7316 & 20852 & 56436 & 146932 & 371188 & 915444\\
			RSurf $Z$ & 44 & 284 & 1228 & 4332 & 13548 & 39148 & 106988 & 280556 & 712684 & 1765356\\
			\hline
		\end{tabular}
		\caption{Edge counts by distance for common codes.}
		\label{tab:edgecounts}
	\end{table}
	\begin{figure}[h]
		\centering
		\begin{subfigure}{0.475\textwidth}
		\begin{tikzpicture}
        			\foreach \i in {0, 1, ..., 4}{
        				\foreach \j in {0, 1, ..., 4}{
        					\pgfmathtruncatemacro{\label}{5*\i + \j+1};
        					\node[inner sep=0pt, outer sep=0pt] (a\i\j) at (\i, \j){};
        				}
        			}
			
        			\foreach \i in {0, 1, ..., 4}{
        				\foreach \j in {0, 1, ..., 3}{
        					\pgfmathtruncatemacro{\y}{\j + 1};
        					\draw (a\i\j) -- (a\i\y);
        					\draw (a\j\i) -- (a\y\i);
        				}
        			}
			
        			\foreach \i in {0, 2, ..., 3}{
 		       		\pgfmathtruncatemacro{\y}{\i + 1};
        				\pgfmathtruncatemacro{\x}{\i + 2};
				\draw (a\i0) to [out=-90, in=-90] (a\y0);
        				\draw (a0\y) to [out=174, in=186] (a0\x);
        				\draw (a\y4) to [out=90, in=90] (a\x4);
        				\draw (a4\i) to [out=6, in=-6] (a4\y);
        			}
			
		        	\foreach \i in {0, 2, ..., 3}{
        				\foreach \j in {0, 2, ..., 3}{
	        				\pgfmathtruncatemacro{\x}{\i + 1};
        					\pgfmathtruncatemacro{\y}{\j + 1};
        					\pgfmathtruncatemacro{\s}{\j + 2};
        					\fill[TangoSkyBlue2] (\i,\j) -- (\i,\y) -- (\x,\y) -- (\x,\j) -- cycle; 
        					\fill[TangoAluminium2] (\i,\y) -- (\i,\s) -- (\x,\s) -- (\x,\y) -- cycle;
      	  				\pgfmathtruncatemacro{\r}{\i + 1};
        					\pgfmathtruncatemacro{\s}{\j + 1};
        					\pgfmathtruncatemacro{\x}{\r + 1};
        					\pgfmathtruncatemacro{\y}{\s + 1};
        					\fill[TangoSkyBlue2] (\r,\s) -- (\r,\y) -- (\x,\y) -- (\x,\s) -- cycle; 
        					\fill[TangoAluminium2] (\r,\j) -- (\r,\s) -- (\x,\s) -- (\x,\j) -- cycle; 
        				}
       		 	}
			
		        	\foreach \i in {0, 2}{
        				\pgfmathtruncatemacro{\y}{\i + 1};
        				\pgfmathtruncatemacro{\x}{\i + 2};
        				\draw[fill,TangoSkyBlue2] (0, \y) to [out=180, in=180] (0, \x);
        				\draw[fill,TangoSkyBlue2] (4, \i) to [out=0, in=-0] (4, \y);
        				\draw[fill,TangoAluminium2] (\y, 4) to [out=90, in=90] (\x, 4);
        				\draw[fill,TangoAluminium2] (\i, 0) to [out=-90, in=-90] (\y, 0);
        			}
			
		        	\foreach \i in {0, 1, ..., 4}{
        				\foreach \j in {0, 1, ..., 4}{
	        				\pgfmathtruncatemacro{\label}{5*(4 - \j) + \i + 1};
        					\node[fill, white, circle, inner sep=0pt, minimum size=13pt] (a\i\j) at (\i, \j){\label};
        					\node[draw, circle, inner sep=0pt, minimum size=13pt] (a\i\j) at (\i, \j){\label};
        				}
        			}
        		\end{tikzpicture}
		\caption{}
		\label{fig:rsurfd5}
		\end{subfigure}
    			\begin{subfigure}{0.475\textwidth}
			       		\centering
		\includegraphics[scale=0.23]{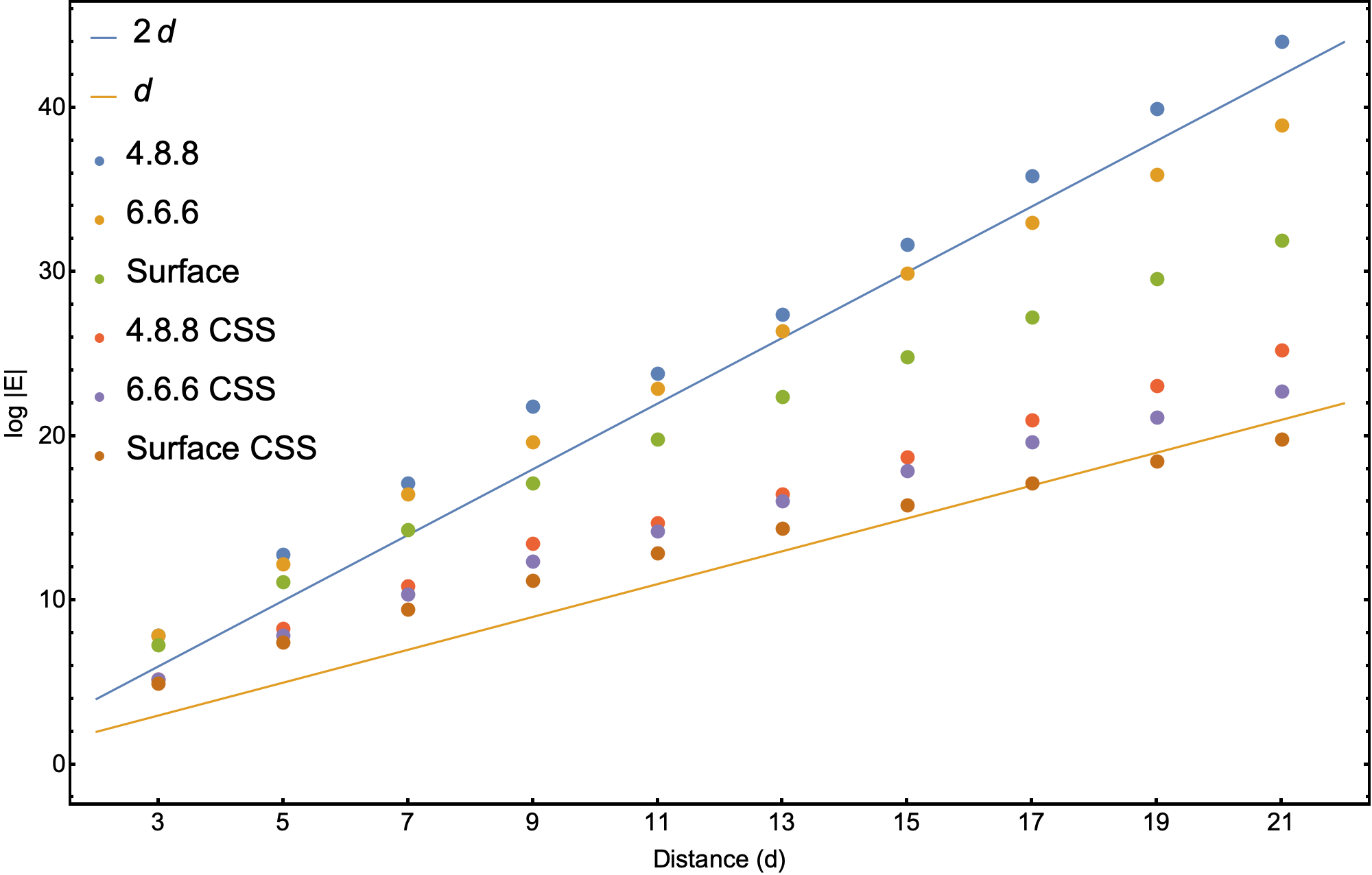}
		\caption{}
		\label{fig:Escaling}
		\end{subfigure}
		\caption{(a) The distance five rotated surface code used in this work: $X$-stabilizers are given by grey (light) faces and $Z$-stabilizers by blue (dark) faces. (b) The scaling of the total edge counts for the trellises listed in Table \ref{tab:edgecounts}.}
	\end{figure}
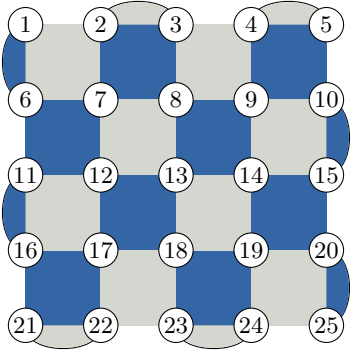
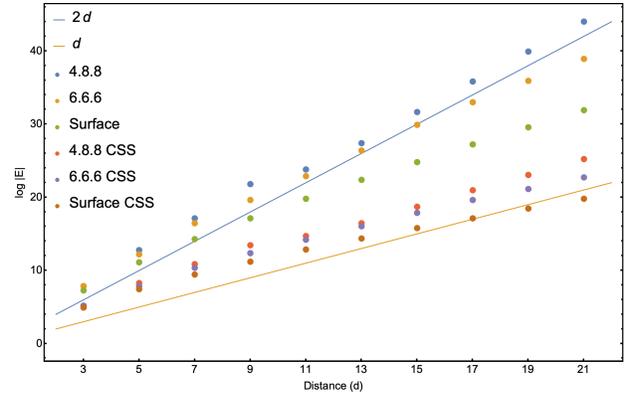
	
	The difference in sizes between the $X$- and $Z$-trellises for the rotated surface code family is an artifact of the particular numbering system used in this work. Consider the $X_2 X_3$ stabilizer in Figure \ref{fig:rsurfd5}. This is active only at depth two since at depth three its syndrome is already zero. On the other hand, the $Z_1 Z_6$ stabilizer is active over five depths and hence contributes vertices over this entire range. Since the $X$- and $Z$-stabilizer measurements are independent, arranging the data with respect to two different numbering schemes will allow both the $X$- and the $Z$-trellises to be isomorphic. This demonstrates the strong affect permutations have on the trellis.

	While an analytical formula for the edge scaling is currently missing, we can gain some numerical insights from Table \ref{tab:edgecounts}. In Figure \ref{fig:Escaling} we see that $|E|$ is exponential in the minimum distance $d$ (compare to the lines $y = 2d$ (blue) and $y = d$ (orange)). The slight bend in the data points, most predominately seen for the CSS rotated surface codes, suggest a more complicated behavior. It is possible that the color codes also have the same trend but were not simulated to high enough distances to make this as visually apparent. It is also possible that the greedy algorithm used to assign qubits to the color code geometries was suboptimal enough to erase this effect.
\end{example}

\begin{example}
	To further emphasize the affect permutations have on the trellis, draw the distance five and seven 4.8.8 color codes, choose an arbitrary boundary, and number the qubits from left to right, level by level. The minimum trellis for this configuration for the distance five code has $|V| = 5,242$ and $|E| = 9,000$ and has $|V| = 177,018$ and $|E| = 293,928$ for distance seven. Comparing to Tables \ref{tab:vcounts} and \ref{tab:edgecounts}, using this numbering scheme would increase the slope of the blue dots in Figure \ref{fig:Escaling}, potentially limiting the ability to do large-scale simulations at a lower distance.
\end{example}

\section{Simulations And Discussion}\label{sec:numerics}
Numerical simulations were performed in the Julia programming language with pre-simulation computations in the MAGMA quantum coding theory library \cite{magmahandbook}. Trellises were efficiently constructed using the theoretical guarantees provided by Corollary \ref{cor:bruteforce}, Theorem \ref{thm:quantstate}, Lemma \ref{lem:VEvs}, and Theorem \ref{thm:bipartthm}. By Theorem \ref{thm:quantstate}, the vertices at depth $i$ are given by the set of all $(n - k)$-tuples whose bits corresponding to generators active at $i$ range through all possible values. To construct the edges at section $i$, determine all of the vertices in $V_{i - 1}$ connecting to $\overline{0} \in V_i$ via Corollary \ref{cor:bruteforce}. Then choose one of these $v \in V_{i - 1}$ and determine all of the vertices in $V_i$ to which it connects, again via Corollary \ref{cor:bruteforce}. These are all of the vertices in the edge configuration by Theorem \ref{thm:bipartthm}, which may be completed with Corollary \ref{cor:bruteforce}. This is the subgroup of $E_i$ by Lemma \ref{lem:VEvs}. The rest of the section consists of all translations (cosets) of this subgroup. The $V_i$ are independent and may constructed in parallel, after which the $E_i$ may also be parallelized. For the codes considered in this work, only the distance 19 and 21 color code trellises were large enough to justify saving so as to not compute them on-the-fly later.

The vertex labels are required for the construction of and theoretical justification for the trellis, but an examination of the Viterbi algorithm shows that they serve no role in decoding. As such, we need not bother shifting the vertex labels for each measured syndrome, and generally we can safely discard them after the construction phase is completed. Even storing the labels as integers is similar to creating a lookup table and can take a non-trivial amount of memory starting at moderately sized codes. The storage requirements for the zero-syndrome trellis valid for any error model became slightly unwieldy to use for large-scale simulations for the distance 21 color codes, but this was implementation dependent. Choosing a specific error model allowed the trellis to be translated to a data structure a fraction of the storage size of the original. For example, the distance 21, 4.8.8 CSS trellises come out to around 300 MB in contrast to the roughly 2 GB file size of the original data structure which includes all the information and generality. We leave a discussion of such improvements to upcoming work. However, since the trellis searches all elements of $S^\perp$, the proper comparison to make here is to the size of a lookup table for the same code. Assuming the $2^{n + k}$ elements of $S^\perp$ are stored as length-$n$ strings of character size 1 byte, a lookup table for the $[[111, 1, 11]]$ code above would be on the order of $\sim 10^{17}$ exabytes.

\begin{figure}
	\begin{center}
		\includegraphics[scale=0.3]{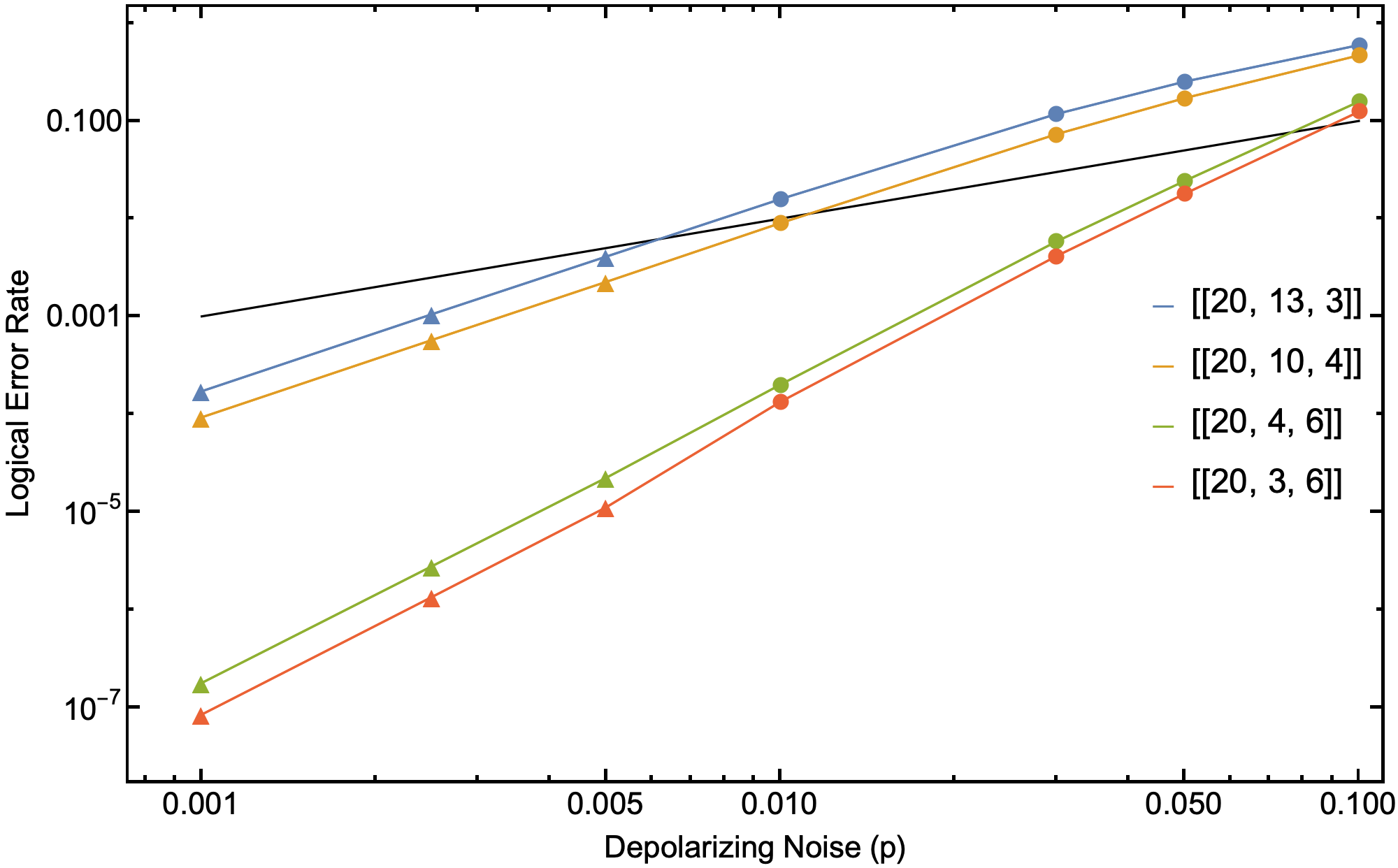}
	\end{center}
	\caption{Simulated logical error rates for the four codes in the row $n = 20$ satisfying our requirements at \cite{codetables}. Triangular data points are importance sampled to a tolerance of $10^{-9}$ and circular points are direct sampled; both methods were used and were found to agree at $p = 0.01$. The black line is $y = x$.}
	\label{fig:codetable}
\end{figure}
To exemplify the versatility of trellis decoding, the row $n = 20$ was selected at random from the quantum error-correcting codes table at \cite{codetables}. Extended codes (with weight one stabilizers) and codes with $k = 0$ or $d \leq 2$ were discarded, and code-capacity (memory model) numerical simulations were performed on the remaining four codes: $[[20, 3, 6]]$, $[[20, 4, 6]]$, $[[20, 10, 4]]$, and $[[20, 13, 3]]$. These codes and their properties were previously unknown to the authors of this work and were never investigated. In particular, it is unknown if another decoding method is known for these codes and if there exist previous pseudothreshold results in the literature. Going from copying the stabilizer matrices from the website to running the simulation took about one minute per code given the automated software tools developed for this work. Simulations consisted of a single round of error correction for a minimum of 30,000 uniformly-sampled non-trivial errors for each data point under a depolarizing noise model acting independently on each qubit,
\begin{equation*}
	\mathcal{E}(\rho) = (1 - p) \rho + \frac{p}{3} X \rho X + \frac{p}{3} Y \rho Y + \frac{p}{3} Z \rho Z.
\end{equation*}
Error correction was considered to have failed if the final state contains a logical error with respect to any of the $2k$ logical operator generators. Results are presented in Figure \ref{fig:codetable}, and the stabilizers for these codes are given in Appendix \ref{sec:codetablesstabs} for convenience.

In order to quantitatively assess the success of trellis decoding, the same simulations were carried out for three common qubit code families, the rotated surface and 4.8.8, 6.6.6 color codes for $Z$-only noise,
\begin{equation*}
	\mathcal{E}(\rho) = (1 - p) \rho + p Z \rho Z,
\end{equation*}
a single-axis dephasing channel, (using the CSS $X$-only trellis) up to at least distance 17. These codes and this noise model were chosen solely for the purpose of comparison with existing decoding literature. Due to the $X-Z$ symmetry of the codes, the same results are obtained for $X$-only noise and depolarizing thresholds are at $3/2$ times the $X$- or $Z$-noise results decoded independently, which we numerically verified for low distances. There are, of course, many other highly-optimized decoders for topological codes, and these choices may not sufficiently demonstrate the value of trellis decoding which shines for moderately-sized quantum error-correcting codes for which there are no efficient decoding strategies. Results for distances three and five were computed exactly and served as a baseline check for our methods. Distances 7 - 21 were sampled to 500,000, 300,000, 100,000, 100,000, 50,000, 50,000, 30,000, and 30,000 non-trivial errors, respectively. Data was taken at intervals of size 0.0025 around the suspected threshold to reduce the error caused by sampling and fitting.

Previous studies have reported a surface code $Z$-only threshold of 10.3\% under minimum-weight perfect matching (MWPM) \cite{wang2003confinement} and 10.9\% using the tensor network decoder of \cite{bravyi2014efficient}. The statistical mechanical threshold is estimated to be 10.9\% \cite{dennis2002topological} without correlations and 12.6\% with $X$- and $Z$-correlations taken into account \cite{bombin2012strong}. Both the trellis and MWPM are minimum-weight decoders, so they should agree for uncorrelated error models. For correlated error models trellis decoding should beat MWPM since the minimum-weight correction on the full trellis will not always match the minimum-weight $X$-correction combined with the minimum-weight $Z$-correction returned by MWPM. For example, consider the error $Y_8 Y_{13} Y_{18}$ on the distance five rotated surface code of Figure \ref{fig:rsurfd5}. The full trellis will return the correction $Y_8 Y_{13} Y_{18}$, while the $X$- and $Z$-trellises will return $Z_3 Z_{23}$ and $X_8 X_{13} X_{18}$, respectively; the latter of which is a logical error. The Viterbi algorithm takes the same number of steps every time and could be slower than using a small matching graph for a large code, so the expected number of errors should be taken into account when comparing the runtime of these two algorithms for uncorrelated error models. The results of our simulations are shown in Figure \ref{fig:rsurfresults}.
\begin{figure}[h!]
	\centering
     	\begin{subfigure}[b]{0.475\textwidth}
        		\centering
		\includegraphics[scale=0.23]{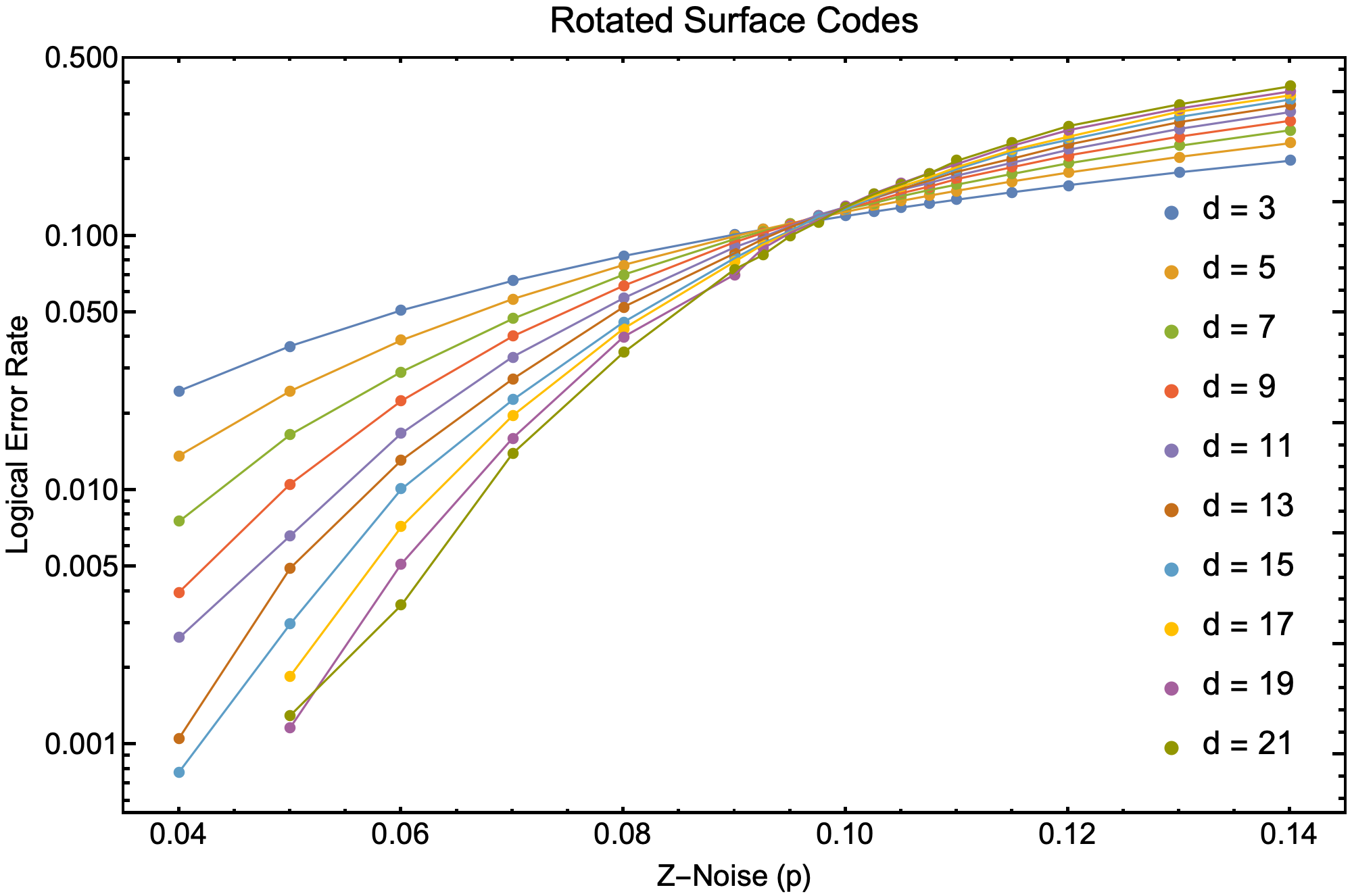}
		\caption{}
		\label{fig:rsurfresults}
	\end{subfigure}
	\hfill
	\begin{subfigure}[b]{0.475\textwidth}
        		\centering
		\includegraphics[scale=0.23]{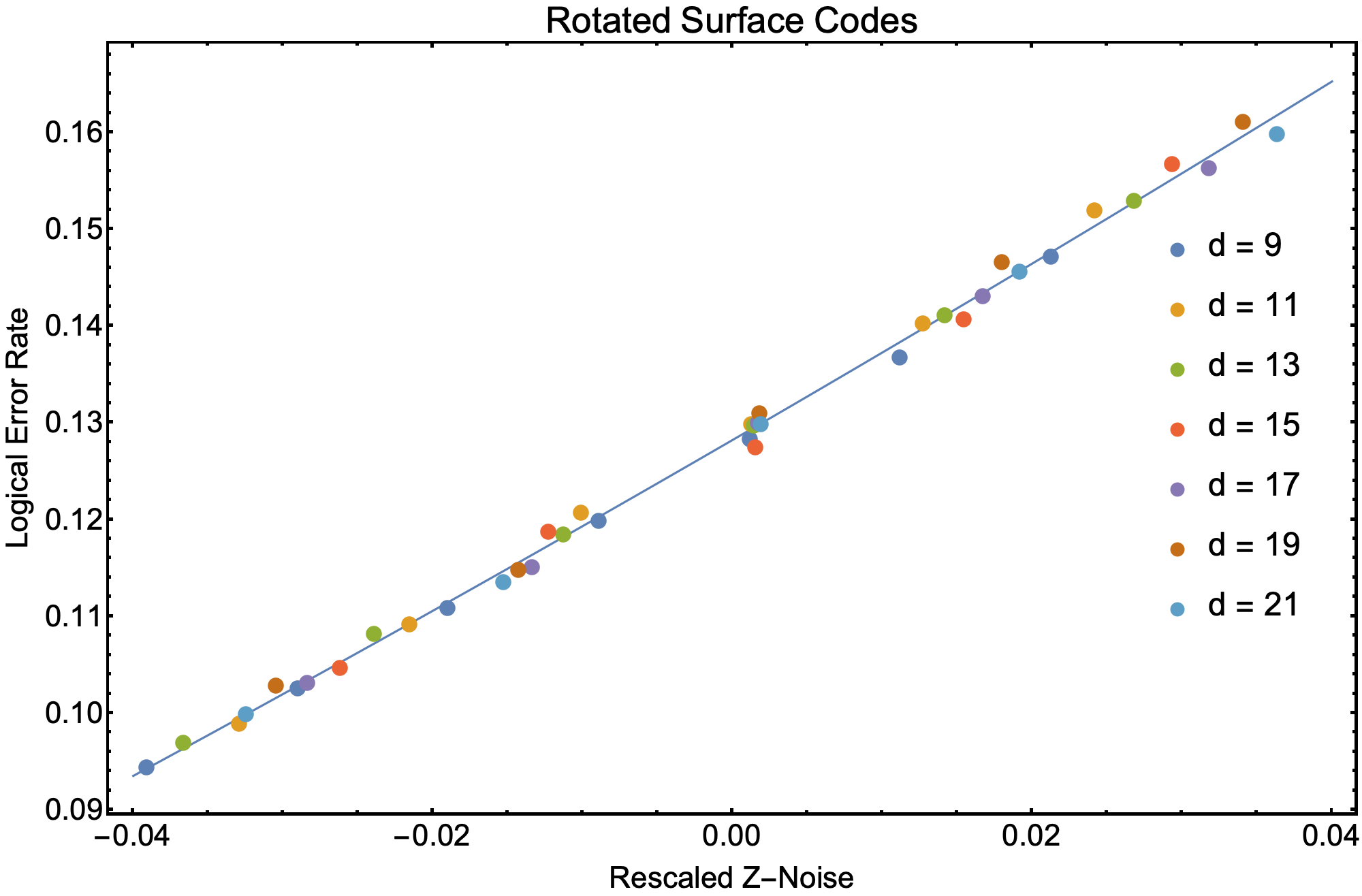}
		\caption{}
		\label{fig:rsurfthresh}
	\end{subfigure}
	\vskip\baselineskip
	\begin{subfigure}[b]{0.475\textwidth}
        		\centering
		\includegraphics[scale=0.23]{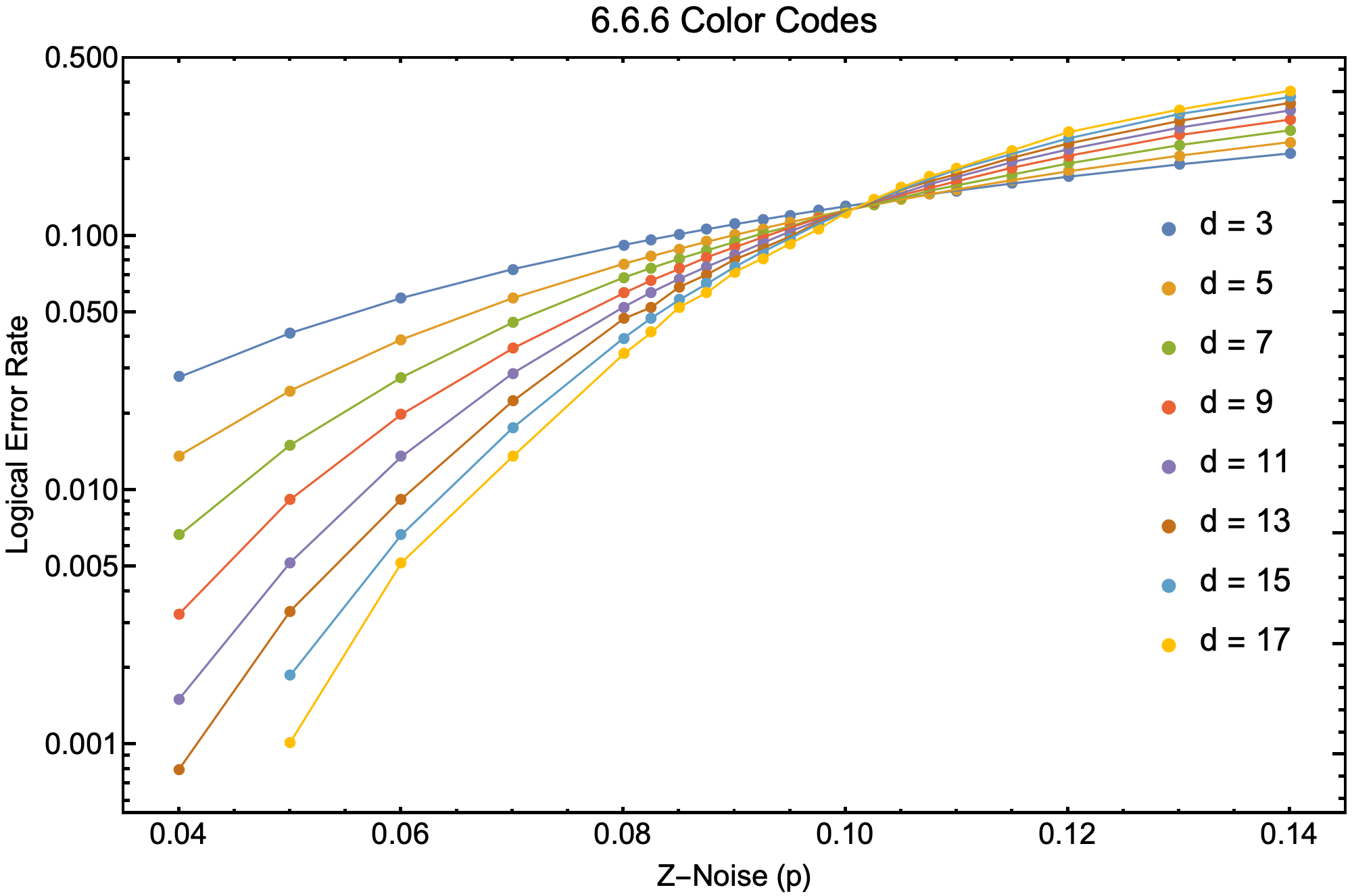}
		\caption{}
		\label{fig:666results}
	\end{subfigure}
	\hfill
	\begin{subfigure}[b]{0.475\textwidth}
        		\centering
		\includegraphics[scale=0.23]{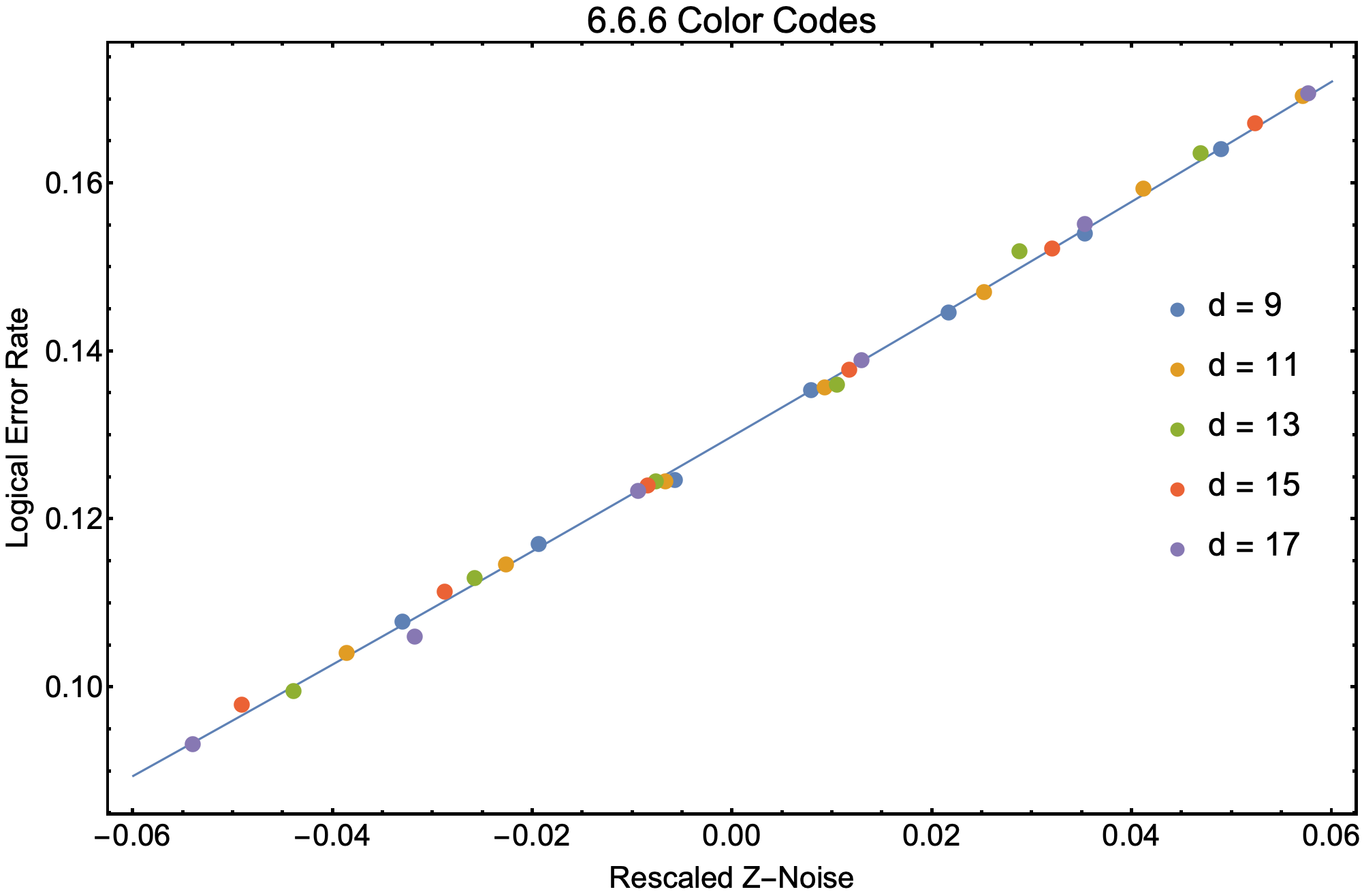}
		\caption{}
		\label{fig:666thresh}
	\end{subfigure}
	\vskip\baselineskip
	\begin{subfigure}[b]{0.475\textwidth}
		\centering
		\includegraphics[scale=0.23]{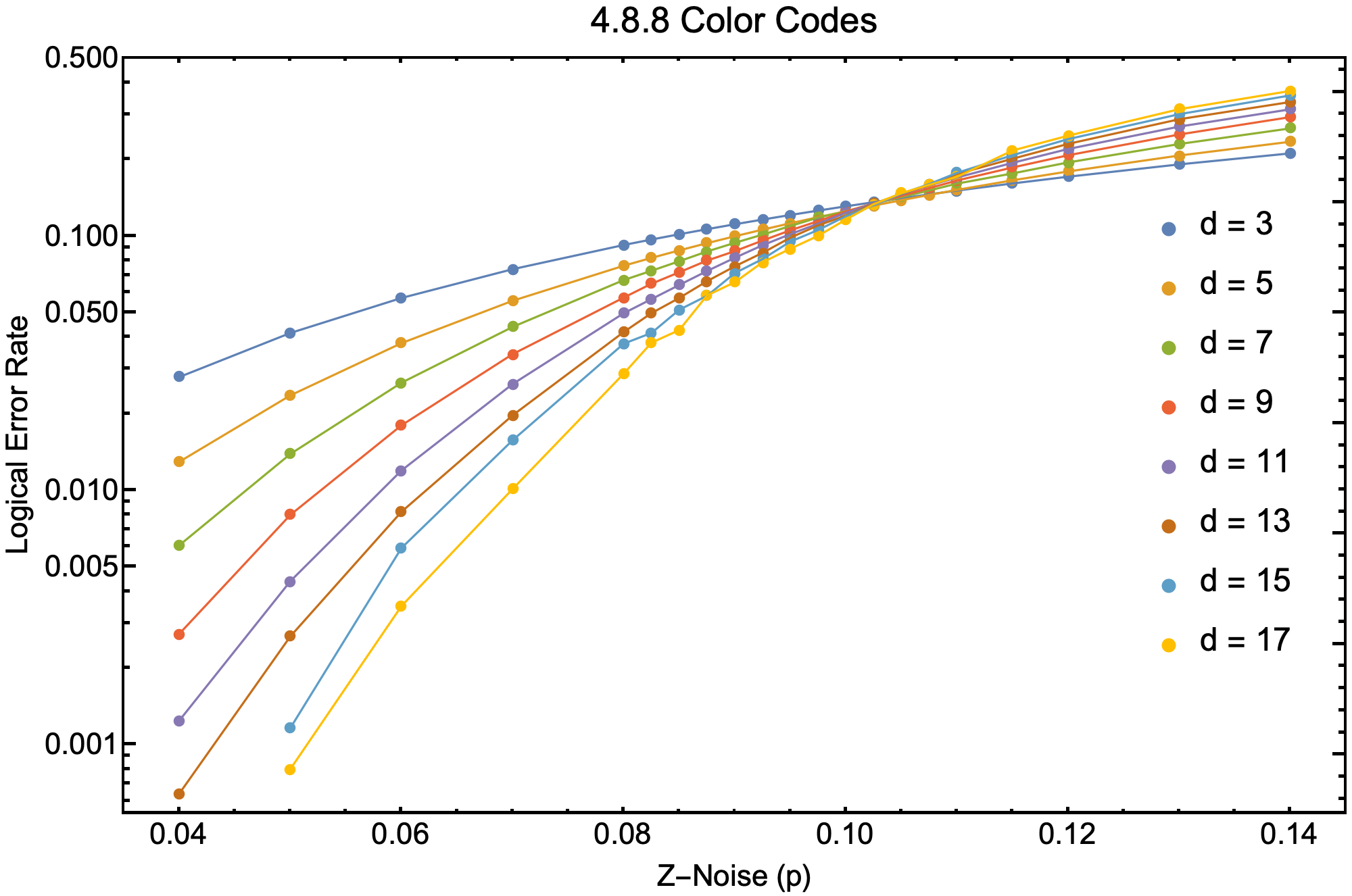}
		\caption{}
		\label{fig:488results}
	\end{subfigure}
	\hfill
	\begin{subfigure}[b]{0.475\textwidth}
		\centering
		\includegraphics[scale=0.23]{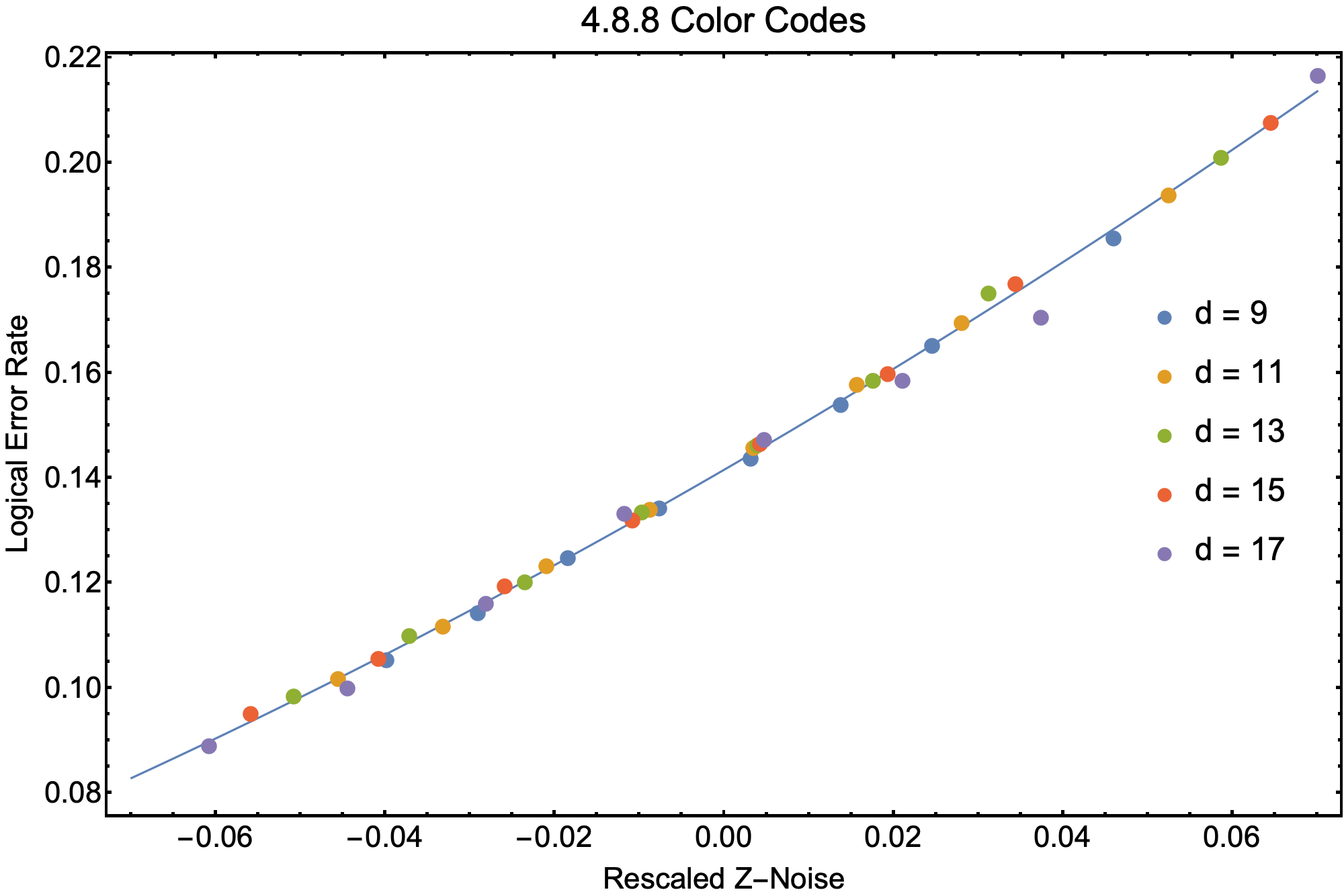}
		\caption{}
		\label{fig:488thresh}
	\end{subfigure}
	\caption{(a), (c), (e) - Threshold results for code capacity (memory model) simulations under $Z$-only noise. (b), (d), (f) - Finite-size threshold analysis for the corresponding codes near threshold. The logical error rate is shown as a function of the rescaled $Z$-noise probability $x = (p - p_{\mathrm{th}}) d^{1/\nu}$. The solid line is the line of best fit to $A + Bx + Cx^2$. Note that the higher distance codes show some signs of under sampling for lower $p$.}
	\label{fig:numresults}
\end{figure}

The Restriction Decoder of \cite{kubica2019efficient} achieves a $Z$-only threshold of 10.2\% for the 4.8.8 color codes with periodic boundary conditions. This was modified to include boundaries in \cite{chamberland2020triangular} which reports a full depolarizing noise threshold of 12.6\% for the 6.6.6 lattice, which results in a $Z$-threshold of roughly 8.4\%. Reference \cite{bombin2012universal} finds a $Z$-threshold of roughly 8.7\% for the 4.8.8 lattice. The statistical mechanical thresholds for both lattices are estimated to be 12.6\% and 12.52\% for the 6.6.6 and 4.8.8 families, respectively \cite{katzgraber2009error, bombin2012strong}. The results of our simulations for the 6.6.6 and 4.8.8 families are shown in Figures \ref{fig:666results} and \ref{fig:488results}, respectively.

Following references such as \cite{stephens2014fault} we perform a finite-size threshold analysis on our data by fitting to the ansatz
\begin{equation*}
	A + B(p -p_{\mathrm{th}}) d^{1/\nu} + C(p -p_{\mathrm{th}})^2 d^{2/\nu}
\end{equation*}
for $A, B, C, p_{\mathrm{th}}$, and $\nu$ starting with $d \geq 9$. For the rotated surface codes we find $p_{\mathrm{th}} \approx 10\%$ ($\nu = 1.58$). We believe the distances sampled in this work are too low to stabilize the third decimal place, but as previously mentioned, we should have a threshold of 10.3\% as $d$ increases. For the color codes, it's possible that all of the distances reported in this work are within the finite-size-effect regime. Nevertheless, we see strong evidence for threshold-like behavior at 10.1\% ($\nu = 1.29$) and 10.4\% ($\nu = 1.51$) for the 6.6.6 and 4.8.8 families, respectively. Further simulations and analysis are required to better estimate the thresholds for these codes using this decoder. Figures \ref{fig:rsurfthresh}, \ref{fig:666thresh}, and \ref{fig:488thresh} show the logical error rate versus the rescaled {\color{red} $Z$-}noise probability $x = (p - p_{\mathrm{th}}) d^{1/\nu}$ for data near each threshold, along with the line of best fit.

Judging by the time it took to complete each simulation, we believe that the CSS, distance 21 trellis for the 4.8.8 color code ($[[421, 1, 21]]$) roughly represents the limit for large-scale simulations. Considering $|E|$ as a measure of the difficulty of decoding, trellis profiles for codes were compared to Tables \ref{tab:vcounts} and \ref{tab:edgecounts} before using this method. We offer this as a rough benchmark but acknowledge that some of this may due in part to not properly leveraging the limited computational resources available for this work. For example, we noticed that a large percentage of the total simulation time for the distance 19 and 21 color codes was consumed by distributing the trellis to all the processors. Switching to the previously alluded to alternative data structure improved this but the data showed slight signs of under-sampling. To fix this we used a process called \textit{sectionalization}, well-known in the classical literature, to reduce the trellis size even further. For example, using this method the distance 21 4.8.8 color code gives a trellis no more than 45\% of the edge size reported in this paper. We have chosen to use this paper as a baseline for trellis theory, including its limits. As such, reporting all improvements to the theory or its practical implementation has been relegated to a future work. The special structure of some codes may be utilized to reduce $|E|$ without the need for further theory, as demonstrated by the next example.

\begin{example}\label{ex:concat}
	\begin{figure}[t]
		\begin{center}
			\includegraphics[scale=0.3]{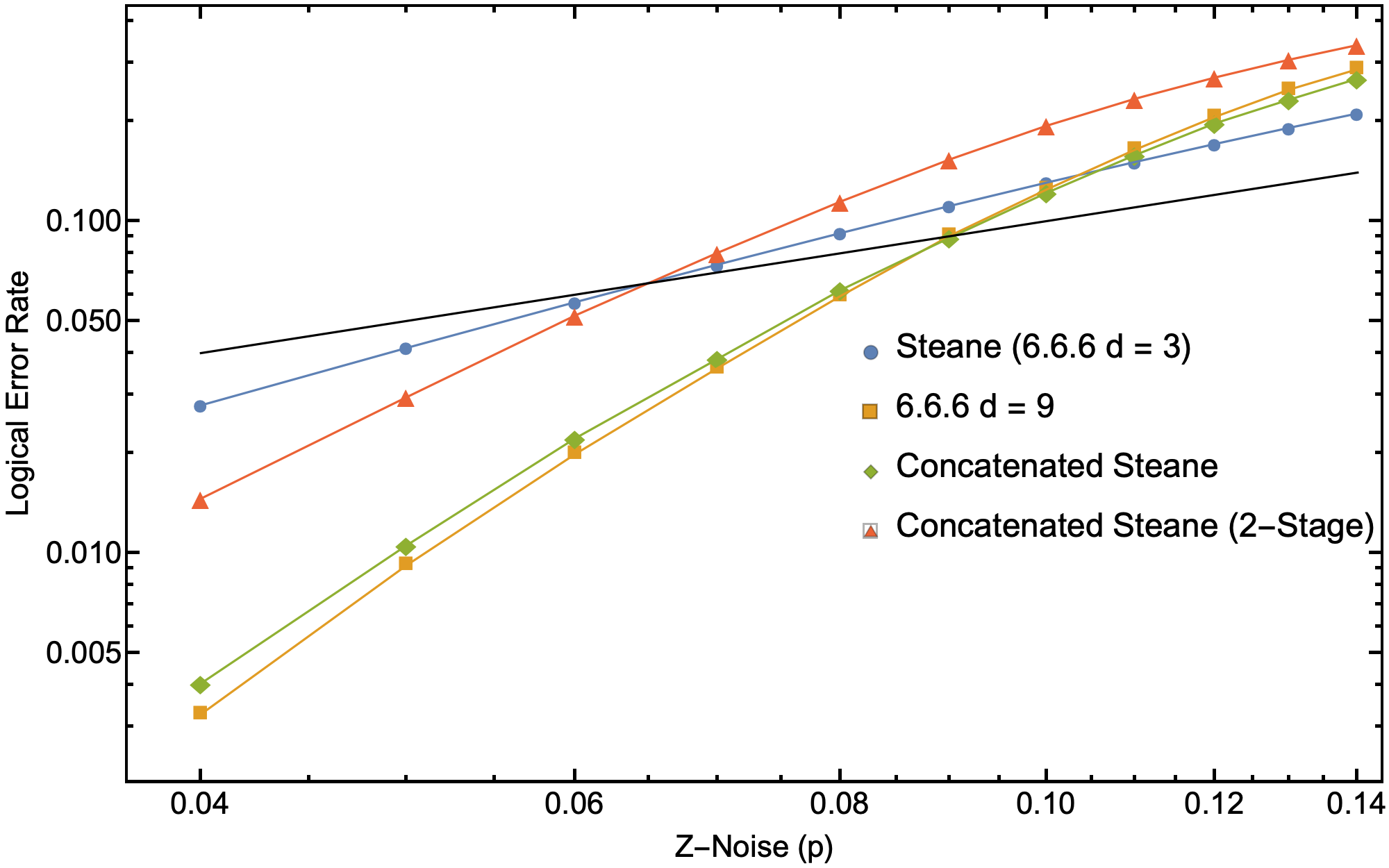}
		\end{center}
		\caption{Simulated logical error rates for Example \ref{ex:concat}. The black line is $y = x$.}
		\label{fig:concat}
	\end{figure}
	Let $H$ denote the numerical matrix representing the $X$- or $Z$-stabilizers for the $[[7,1,3]]$ Steane code,
	\begin{equation*}
		H = \begin{pmatrix}
				1 & 1 & 0 & 1 & 1 & 0 & 0\\
				0 & 1 & 1 & 0 & 1 & 1 & 0\\
				0 & 0 & 0 & 1 & 1 & 1 & 1
			\end{pmatrix}.
	\end{equation*}
	The level-2 concatenated Steane code is the $[[49, 1, 9]]$ code whose $X$- or $Z$-stabilizers are given by the Pauli operators corresponding to $I \otimes H$ and $H \otimes \vec{1}$, where $I$ is the $7 \times 7$ identity matrix and $\vec{1}$ is the length-7 all-ones vector. The trellis for this code has $|V| = 626$ and $|E| = 844$. Using Theorem \ref{thm:split} and the subsequent discussion, we consider a 2-stage, suboptimal decoder. The seven copies of $H$ are decoded independently at a cost of 36 edges each. The resulting corrections are made then the syndrome updated and the bits corresponding to $H \otimes \vec{1}$ are again decoded using the trellis for $H$. The correction returned from this trellis is tensored with $\vec{1}$ to map back to the original problem. Since the first seven decodes can be done simultaneously, the time complexity of this decoder is determined by only $36 + 36 = 72$ edges, a 91.5\% savings. Numerical simulations similar to the ones above were performed for both this decoder and the normal, single stage trellis. Results are presented in Figure \ref{fig:concat} where we see that this is indeed suboptimal, trading speed for accuracy.

The calculated pseudothreshold for dephasing noise between the Steane code and concatenated Steane code using 2-stage block decoding is 6.44\%, in agreement with the block decoding dephasing channel threshold of 6.46\% estimated from the exact depolarizing channel threshold of 9.69\% \cite{RahnPRA2002}. The calculated pseudothreshold for dephasing noise between the Steane and the concatenated Steane code using the full trellis is 10.55\%. This is known to be below the true threshold for continued concatenation based on message passing of 12.5\% estimated from a depolarizing channel threshold of 18.8\% \cite{PoulinPRA2006}. The $[[7, 1, 3]]$ Steane code can also be considered a distance 3 6.6.6 color code. It is informative to compare the performance of the $[[61, 1, 9]]$ 6.6.6 color code to the level-2 concatenated Steane code.  We see that the color code has a slightly better performance at the cost of 12 more qubits.\\
\end{example}

The numerical threshold estimates obtained in this work are competitive with, and sometimes even better than, current state-of-the-art decoders. We attribute this success to two main things. First, with the exception of organizing the data into logical cosets, the decoder is provided with close to the maximum amount of information possible of both the code and the error channel. Traditional decoders often only rely on a fraction of this information. Second, since every Pauli string represented by a path in the trellis has the same syndrome as was measured, any correction returned by the decoder will, at minimum, force the resulting state to have zero syndrome. If there are few or no logical operators of the weight of the true error, the decoder will therefore correct the majority of the correctable errors of that weight regardless of the minimum distance of the code. 

\section{Conclusion}\label{sec:conclusion}
We hope to reinitiate interest in trellis decoding for stabilizer codes with this work. Starting from the basic construction in \cite{ollivier2006trellises}, we have provided a theoretical foundation required for any future analysis. In doing so, we have detailed the use of, improved the runtime of, and demonstrated the effectiveness of trellis decoders. The qualitative metric we propose shed light on the difference between color codes and surface codes and CSS and non-CSS stabilizer codes. In particular, the trellis is a native color code decoder without invoking projections, restrictions, homology, charges, excitations, string nets, or boundaries.

The theoretical work here attempted to closely parallel classical trellis theory, yet also introduced new results unique to the quantum setting. Comparing this work with the remaining classical literature, a considerable number of open problems remain. In particular, many, if not most, fundamental results of the classical theory rely on the relationship between a code and its dual. The concept of duality is a bit tricker for stabilizer codes and we plan to address this in subsequent work.

An important set of classical results missing quantum analogues are bounds for various quantities given only $n$, $k$, and $d$ \cite{muder1988minimal, zyablov1993bounds, kasami1993complexity, forney1994dimension, lafourcade1995lower, lafourcade1995asymptotically, ytrehus1995trellis, kiely1996trellis}. This is easier in the classical setting as the minimum distance is related to the number of paths between a vertex in $V_i$ and a vertex in $V_{i + j}$ for some $j$. Here, this would be the minimum distance of $S^\perp$ and not $S^\perp \backslash S$. A thorough analysis will also shed light on the critical open question of the scalability of trellis decoding, which is required to properly compare it with other decoding techniques. Having a complete answer to this question will allow for the analysis of the asymptotic trellis decoding behavior of various code families. The question then becomes at what point this scaling becomes too cumbersome for modern quantum controller hardware. In such cases, trellis pruning \cite{bertrand2004simplified} or coset techniques may possibly be employed. Classically, trellis decoding becomes computationally more intensive as the encoding rate increases, since $\max_i |V_i |$ grows with $k$.

While the code capacity results are promising, further numerical studies are needed to determine the behavior of this method under the more realistic and interesting phenomenological, circuit-level, and biased noise models. Incorporation into other paradigms such as flag-syndrome extraction and single-shot error correction are also interesting. Trellises are a hidden Markov model for the decoding process and it's intriguing to consider whether a Markov chain with memory could be used to handle correlated error models. Qudit codes were not investigated during this work are important to investigate as decoders for such codes are fewer in number and are not as well understood as their qubit counterparts. The treatment of CSS codes in this work may lead to further quantitative results for these codes. Along these lines, it is common in classical coding theory to exploit the coset structure of codes, which make a visual appearance in the trellis \cite{fujiwara1998trellis, morelos1999constructions}. The study of such trellis substructures may provide new insight into the structure of stabilizer codes.

\section*{Acknowledgements}
E.S. would like to thank Evans Harrell and Benjamin Ide for helpful discussions. This work was supported by the Office of the Director of National Intelligence - Intelligence Advanced Research Projects Activity through an Army Research Office contract (W911NF-16-1-0082), the Army Research Office (W911NF-21-1-0005) and the  Army Research Office Multidisciplinary University Research Initiative (W911NF-18-1-0218). E.S. was also funded in part by the NSF QISE-NET fellowship through NSF award DMR-17474266.

\bibliographystyle{unsrt}
\bibliography{Quantum_Trellis_Decoding.bib}

\newpage
\appendix
\section{The Viterbi Algorithm}\label{sec:viterbiapp}
\begin{figure}[H]
	\centering
     	\begin{subfigure}{\textwidth}
        		\centering
		\includegraphics[scale=1]{trellis.pdf}
		\caption{Section 1}
		\label{fig:Vit1}
	\end{subfigure}
	\begin{subfigure}{\textwidth}
        		\centering
		\includegraphics[scale=1]{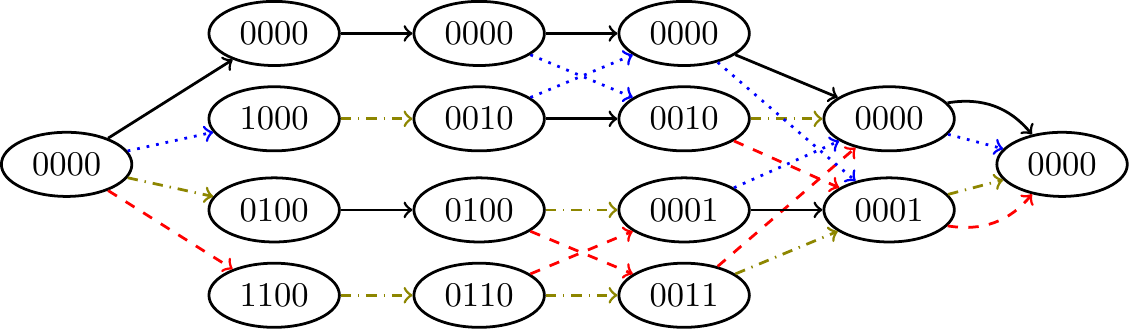}
		\caption{Section 2}
		\label{fig:Vit2}
	\end{subfigure}
	\begin{subfigure}{\textwidth}
        		\centering
		\includegraphics[scale=1]{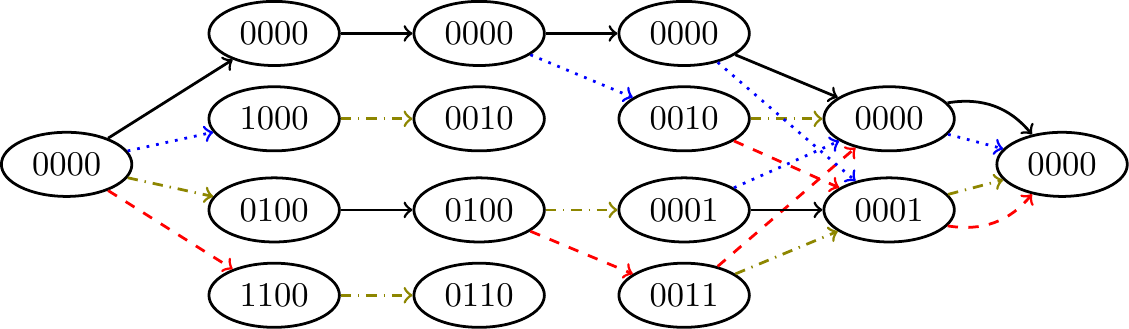}
		\caption{Section 3}
		\label{fig:Vit3}
	\end{subfigure}
	\begin{subfigure}{\textwidth}
        		\centering
		\includegraphics[scale=1]{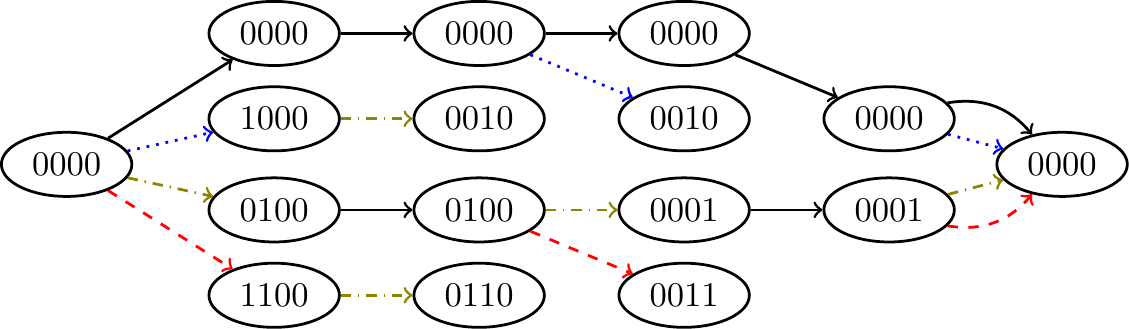}
		\caption{Section 4}
		\label{fig:Vit4}
	\end{subfigure}
	\begin{subfigure}{\textwidth}
        		\centering
		\includegraphics[scale=1]{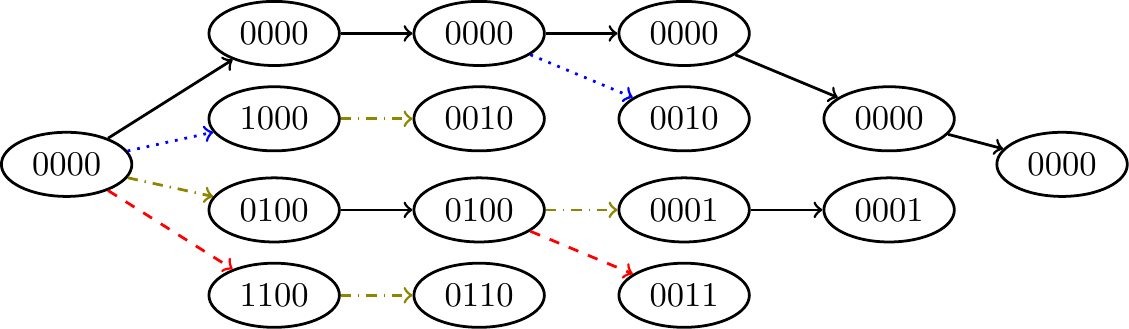}
		\caption{Section 5}
		\label{fig:Vit5}
	\end{subfigure}
	\caption{The Viterbi algorithm applied to the example trellis of Figure \ref{fig:trellisex} for the error model $\Pr(I) = 0$, $\Pr(X) = \Pr(Z) = 1$, and $\Pr(Y) = 2$. Ties were broken in a manner to keep the result looking clean. The path from $V_0$ to $V_n$ provides the final correction, in this case $IIIII$ - which is true since there was no error.}
	\label{fig:VitEx}
\end{figure}

\newpage
\begin{algorithm}[h!]
	\DontPrintSemicolon
	\SetAlgoLined
	\SetKwComment{Comment}{$\triangleright$\ }{}
	\SetKwProg{MyStruct}{Struct}{ contains}{end}
	
	\KwInput{A vertex set, $V$, along with an edge set, $E$, corresponding to a valid trellis.}
	\KwOutput{Pauli string}
	
	\BlankLine
	
	\MyStruct{Vertex}{
  		int prev \Comment*[r]{Final step acts as a linked list}
		float value\;
		Edge edge \Comment*[r]{Initialized to $null$}
	}
	
	\BlankLine
	
	\MyStruct{Edge}{
  		Vertex source\;
		float weight\;
		char label\:
	}
	
	\BlankLine

	$V_0 \ni \overline{0}\text{.value} \gets 0$\;
	 \For{$i \gets 1$ \textbf{to} $n$}{
	 	\ForEach{$v \in V_i$}{
			$v\text{.value} \gets \di \min_{\substack{e \in E_i\\ t(e) = v}} \{e\text{.source.value} + e\text{.weight}\}$ \Comment*[r]{For $e$ which achieves minimum}
			$v\text{.prev} \gets e\text{.source}$\;
			$v\text{.edge} \gets e$\;
		}
	}
	
	\BlankLine
	
	$v \gets \overline{0} \in V_n$ \Comment*[r]{Trace optimal path backward for correction}
	correction $\gets$ ``" \Comment*[r]{Empty string}
	\For{$i \gets n$ \textbf{to} $2$ \textbf{by} $-1$}{
		$\text{correction} \gets v\text{.edge.label} + \text{correction}$ \Comment*[r]{Concatenate to left of string}
		$v \gets v\text{.prev}$\;
	}
	
	\BlankLine
	
	\KwRet{correction}\;
	
	\caption{Viterbi}
	\label{alg:Viterbi}
\end{algorithm}

\newpage
\section{Distance Three Rotated Surface And Color Code Diagrams}\label{sec:knowncodes}
\begin{figure}[h]
     	\begin{subfigure}{\textwidth}
		\includegraphics[scale=1, center]{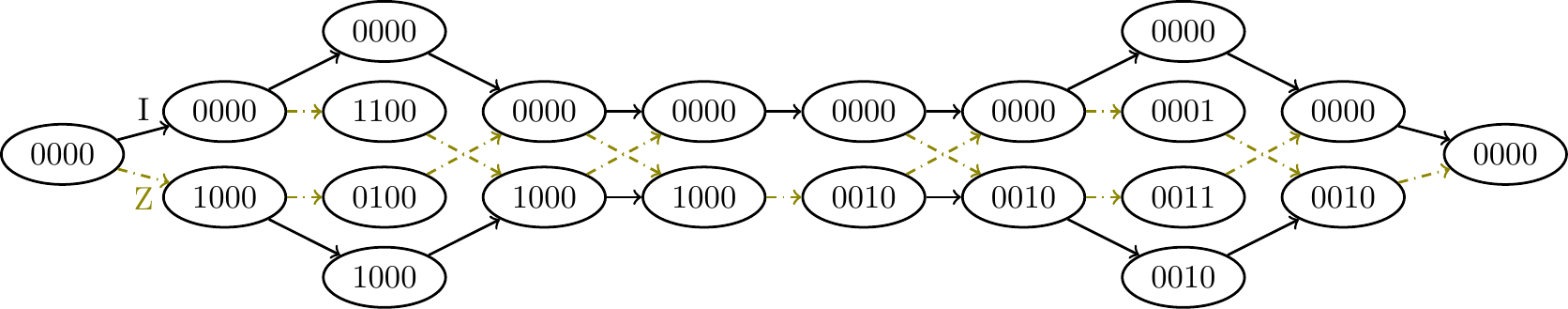}
		\caption{}
		\label{fig:Surf17Xmin}
	\end{subfigure}
	\begin{subfigure}{\textwidth}
		\includegraphics[scale=1, center]{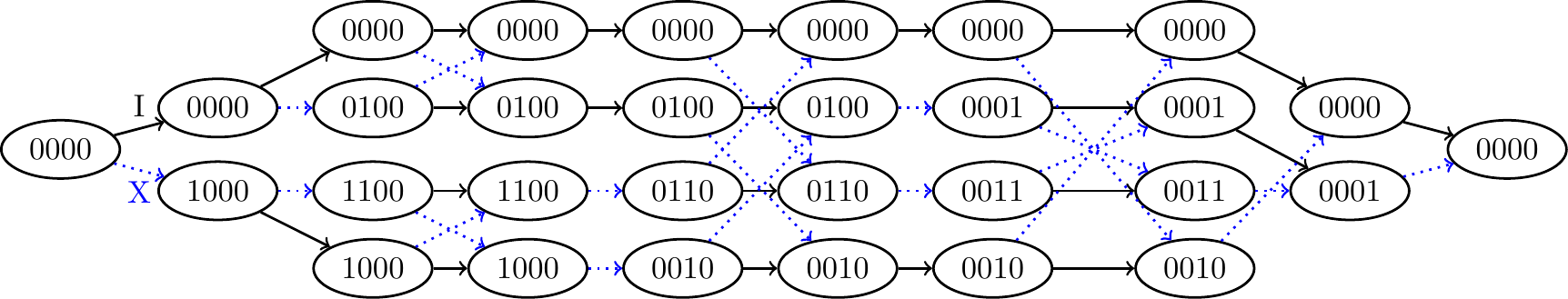}
		\caption{}
		\label{fig:Surf17Zmin}
	\end{subfigure}
	\begin{subfigure}{\textwidth}
		\hspace*{-2.5cm} \includegraphics[scale=0.65, left]{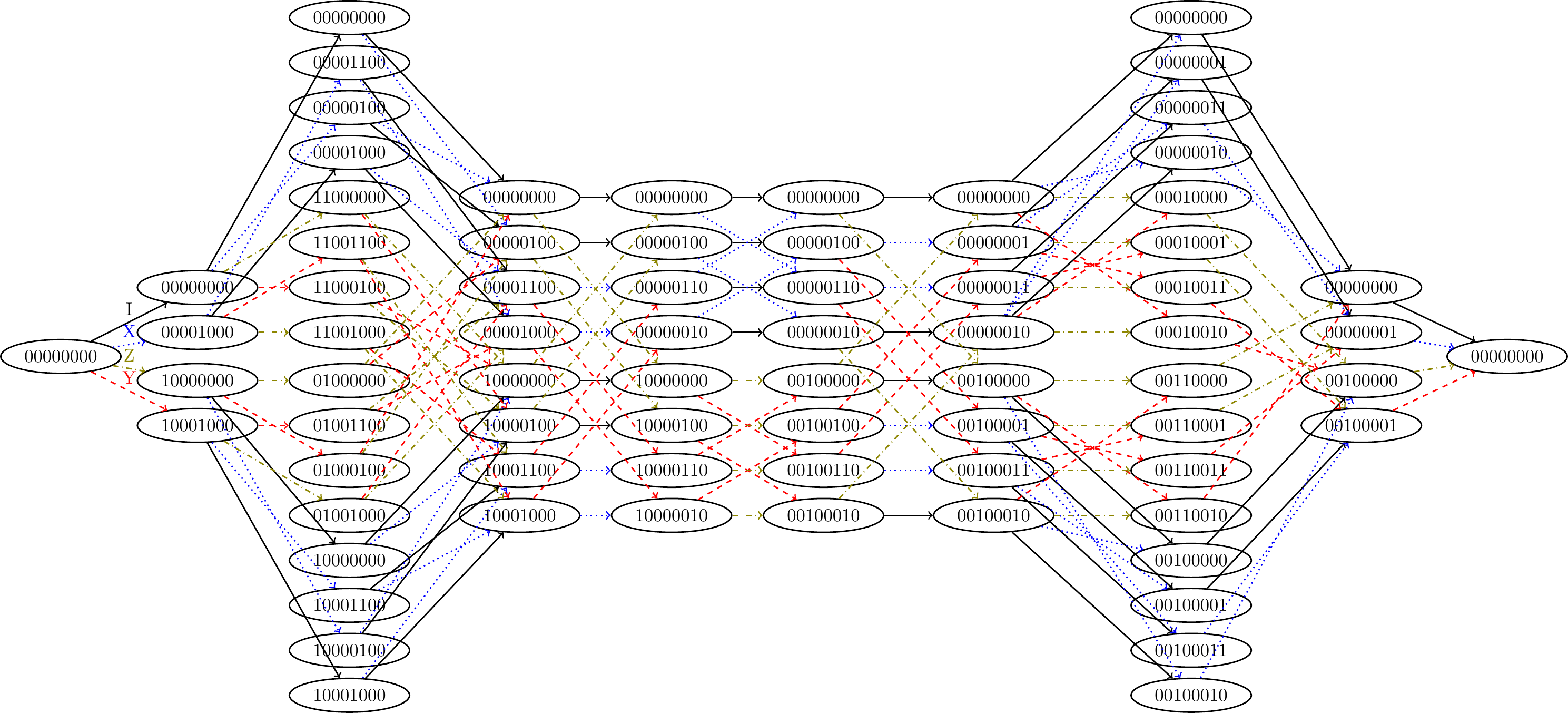}
		\caption{}
    		\label{fig:Surf17full}
	\end{subfigure}
	\caption{Trellis diagrams for the distance three rotated surface codes: (a) $X$-stabilizers only, (b) $Z$-stabilizers only, (c) the full code with vertices organized by the trellis product of (a) with (b). It is common to collapse edge sections when $\deg_{\text{out}} = 1$, combining labels with the section before it. This is called \textit{sectionalization} and will not be discussed further in this work. It is known that sectionalization can reduce $\max |V_i|$ but not $\max |E_i|$ \cite{forney1994dimension}.}
	\label{fig:CSSfigs}
\end{figure}

\begin{figure}[h!]
	\begin{center}
    		\includegraphics[scale=0.37]{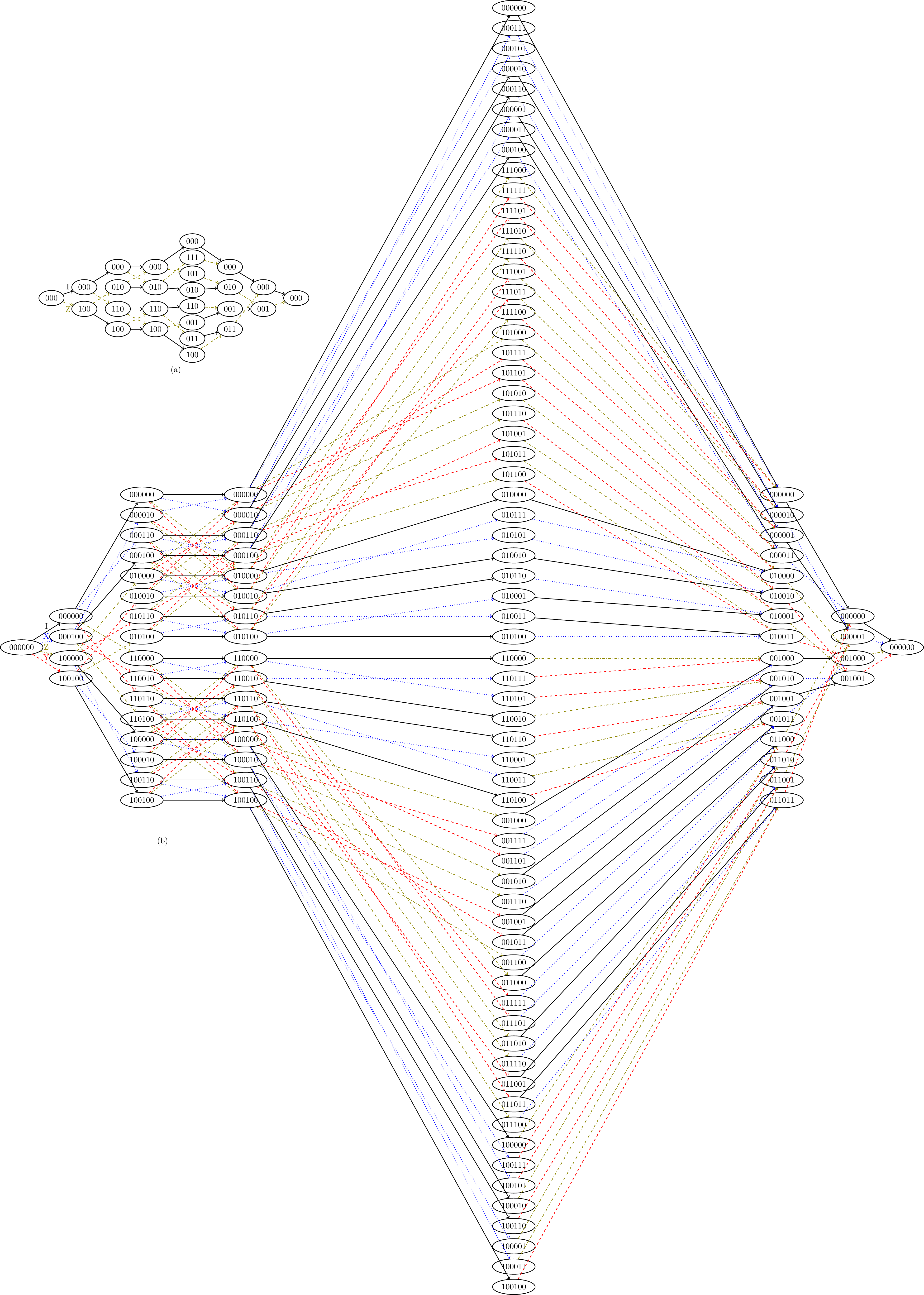}
    	\end{center}
    	\caption{Trellis diagrams for the distance three color code: (a) $X$- (or $Z$-) stabilizers only, (b) the full code with vertices organized by the trellis product of (a) with itself.}
    	\label{fig:cc666d3full}
\end{figure}

\newpage
\section{The Codes Of Figure \ref{fig:codetable}}\label{sec:codetablesstabs}
\begin{tabular}[t]{@{}c@{}}
	$[[20, 3, 6]]$\\
	\hline
	\stretchit{XIIIIIIIIYIXIYZYXZXI}\\
	\stretchit{ZIIIIIIIZXIZIYXYIXZY}\\
	\stretchit{IXIIIIIIZYZYXIXYYYIZ}\\
	\stretchit{IZIIIIIIIXYXZYXZZZIZ}\\
	\stretchit{IIXIIIIIZXIXYXZIYXXY}\\
	\stretchit{IIZIIIIIIZIZXXIYZXZI}\\
	\stretchit{IIIXIIIIIIXXXIZIIIXX}\\
	\stretchit{IIIZIIIIZIZZZZXZZZZX}\\
	\stretchit{IIIIXIIIIXYXXXXYXXXX}\\
	\stretchit{IIIIZIIIIZXZZZZXZZZZ}\\
	\stretchit{IIIIIXIIIYXZXZYZZYIX}\\
	\stretchit{IIIIIZIIZXZYZXYXXYIX}\\
	\stretchit{IIIIIIXIZYYZZXIZZIXY}\\
	\stretchit{IIIIIIZIIXXYYXYIIYZI}\\
	\stretchit{IIIIIIIXZXXIZZYXZIZZ}\\
	\stretchit{IIIIIIIZIZZIYIZXIYYZ}\\
	\stretchit{IIIIIIIIXIIIIXXXXXIZ}
\end{tabular}
\hspace{1cm}
\begin{tabular}[t]{@{}c@{}}
	$[[20, 4, 6]]$\\
	\hline
	\stretchit{XIIIIIIIXYIXIZYZIYXZ}\\
	\stretchit{ZIIIIIIIZXIZIYXYIXZY}\\
	\stretchit{IXIIIIIIYYZYXXIZZZII}\\
	\stretchit{IZIIIIIIXXYXZZIYYYII}\\
	\stretchit{IIXIIIIIYXIXYIYXZIXX}\\
	\stretchit{IIZIIIIIXZIZXIXZYIZZ}\\
	\stretchit{IIIXIIIIXIXXXXYXXXXY}\\
	\stretchit{IIIZIIIIZIZZZZXZZZZX}\\
	\stretchit{IIIIXIIIIXYXXXXYXXXX}\\
	\stretchit{IIIIZIIIIZXZZZZXZZZZ}\\
	\stretchit{IIIIIXIIXYXZXYZYYZIY}\\
	\stretchit{IIIIIZIIZXZYZXYXXYIX}\\
	\stretchit{IIIIIIXIYYYZZIXYYXXX}\\
	\stretchit{IIIIIIZIXXXYYIZXXZZZ}\\
	\stretchit{IIIIIIIXYXXIZYZIYXZI}\\
	\stretchit{IIIIIIIZXZZIYXYIXZYI}
\end{tabular}

\vspace{1.5cm}
\begin{tabular}[t]{@{}c@{}}
	$[[20, 10, 4]]$\\
	\hline
	\stretchit{XIIXZIIXYZIYIXZIZXIZ}\\
	\stretchit{ZIIZYIIZXYIXIZYIYZIY}\\
	\stretchit{IXIZZIIIYYXZXZIYIIZZ}\\
	\stretchit{IZIYYIIIXXZYZYIXIIYY}\\
	\stretchit{IIXZXIIXIXXXXYZIXZXY}\\
	\stretchit{IIZYZIIZIZZZZXYIZYZX}\\
	\stretchit{IIIIIXIXZZZZXIXYZIZY}\\
	\stretchit{IIIIIZIZYYYYZIZXYIYX}\\
	\stretchit{IIIIIIXZZXZXIXZZIZYY}\\
	\stretchit{IIIIIIZYYZYZIZYYIYXX}
\end{tabular}
\hspace{1cm}
\begin{tabular}[t]{@{}c@{}}
	$[[20, 13, 3]]$\\
	\hline
	\stretchit{XIIZIXXIYZYIZYXZXIZY}\\
	\stretchit{ZIXIIXZYZIXZIXIYZYXY}\\
	\stretchit{IXXYIZZIXYXYXYXXIZII}\\
	\stretchit{IZXXIIIIYYYXIIXIYYZX}\\
	\stretchit{IIZYIXYZZXIIZYYIIXYY}\\
	\stretchit{IIIIXXXXXXXXXXXXXXXX}\\
	\stretchit{IIIIZZZZZZZZZZZZZZZZ}
\end{tabular}

\newpage
\section{Proofs Of Technical Lemmas}\label{sec:proofs}
\makeatletter
\renewcommand{\thesection}{\Roman{section}}
\makeatother
\setcounter{section}{3}
\propone*
\begin{proof}
	It is clear that by construction every Pauli string in $S^\perp$ corresponds to a length-$n$ path in the trellis. It remains to show that every length-$n$ path in the trellis corresponds to a Pauli string in $S^\perp$ and that all such paths are unique. Let $(\overline{0}, P_1,\sigma_1(P_1 I \hdots I)) \in E_1$ be an edge with label $P_1$, then pick an arbitrary edge $(\sigma_1(P_1 I \hdots I), P_2, \sigma_2(P_1 P_2 I \hdots I)) \in E_2$. Continuing this process, we end with an edge $(\sigma_{n - 1}(P_1 \hdots P_{n - 1} I), P_n, \sigma_n(P_1 \hdots P_n)) \in E_n$. This last terminus is the zero syndrome by construction, hence, concatenating the edge labels in the path, the Pauli string $P_1 P_2 \hdots P_n$ has zero syndrome and is therefore an element of $S^\perp$ by definition. Since each vertex can only have a unique outgoing edge with a given label, this string uniquely identifies a path in the trellis and no other path can have the same label.
\end{proof}

\corotwo*
\begin{proof}
	Let $v_i \in V_i$, $v_{i + 1} \in V_{i + 1}$, and suppose concatenating a Pauli element $P_i$ to a path with terminus $v_i$ changes the syndrome of $v_i$ to the syndrome of $v_{i + 1}$. Let $\overline{V_0 v_i}$ and $\overline{v_{i + 1} V_n}$ be a path from $V_0$ to $v_i$ and $v_{i + 1}$ to $V_n$, respectively. Then the Pauli string represented by the path $\overline{V_0 v_i} P_i \overline{v_{i + 1} V_n}$ has a length-$n$ and zero syndrome, and is hence an element of $S^\perp$. The result then follows from the previous proposition.
\end{proof}

\setcounter{theoremINNER}{4}
\propfive*
\begin{proof}
	Every path from $V_0$ to $V_n$ must go through some vertex $v \in V_i$. The set of all Pauli strings in $S^\perp$ which map to the vertex $v$ under $\sigma_i$ is a subset of the cosets described in Figure \ref{fig:trelliscosets} and hence has cardinality bounded above by $|S^\perp_{\mathfrak{p}_i}| |S^\perp_{\mathfrak{f}_i}|$. The trellis can be written as a union of paths passing through $v$ over all $v \in V_i$, so the number of such cosets is $|V_i|$. Using Proposition \ref{prop:iso}, we then have
	\begin{equation*}
		|S^\perp| \leq |V_i| |S^\perp_{\mathfrak{p}_i}| |S^\perp_{\mathfrak{f}_i}|.
	\end{equation*}
	Likewise, every path in the trellis contains one edge in section $E_i$. Since edges are of the form $(s(e), P, t(e))$, a similar argument gives
	\begin{equation*}
		|S^\perp| \leq |S^\perp_{\mathfrak{p}_{i - 1}}| |E_i| |S^\perp_{\mathfrak{f}_i}|.
	\end{equation*}
\end{proof}

\setcounter{theoremINNER}{6}
\thmseven*
\begin{proof}
	Let $i$ be fixed.
	\begin{enumerate}
		\item The vertices are the image of $S^\perp$ under $\sigma_i$. By the above comments, the past and the future have zero syndrome therefore, $V_i = \sigma_i(S^\perp) \cong S^\perp / S^\perp_{\mathfrak{p}_i} \times S^\perp_{\mathfrak{f}_i}$ by the first isomorphism theorem for groups.
		\item The kernel of the map from $S^\perp$ to $E_i$ are the Pauli strings which map to $(\overline{0}, I, \overline{0})$, i.e.,
			\begin{align*}
				\ker \sigma_{i - 1}(S^\perp) \cap \{ P \in S^\perp \mid P_i = I \} \cap \ker \sigma_i(S^\perp) &= S^\perp_{\mathfrak{p}_{i - 1}} \times S^\perp_{\mathfrak{f}_{i - 1}} \cap \{ P \in S^\perp \mid P_i = I \} \cap S^\perp_{\mathfrak{p}_i} \times S^\perp_{\mathfrak{f}_i}\\
				&= S^\perp_{\mathfrak{p}_{i - 1}} \times S^\perp_{\mathfrak{f}_i}.
			\end{align*}
	\end{enumerate}
\end{proof}

\lemmaeight*
\begin{proof}
	The proof of Theorem \ref{thm:quantstate} shows that $E_i$ is a group. Clearly the identity edge exists and closure follows since if $(s_1, P_1, \overline{0}), (s_2, P_2, \overline{0}) \in E_i$, then $(s_1, P_1, \overline{0}) (s_2, P_2, \overline{0}) = (s_1 \cdot s_2, P_1 P_2, \overline{0}) \in E_i$, where $\cdot$ is the appropriate group (or vector space) binary operation on the vertices. It remains to show that $E_0$ contains inverses. Let $(a, b, \overline{0}) \in E_0$ and $(c, d, e)$ be its inverse in $E_i$. The group operation on $E_i$ forces $e = \overline{0}$.
\end{proof}

\coronine*
\begin{proof}
	Let $v \in V_i$ and denote by $E_\text{in}(v)$ the set of edges with terminus $v$, $E_\text{in}(v) = \{e \in E_i \mid t(e) = v\}$. By definition, $\deg_\text{in}(v) = |E_\text{in}(v)|$. By Lemma \ref{lem:VEvs}, $E_\text{in}(v)$ is a coset of $E_\text{in}(\overline{0})$ in $E_i$. In particular, $|E_\text{in}(v)| = |E_\text{in}(0)|$. Since there are $|E_i|$ total edges equally divided among $|V_i|$ vertices, $\deg_\text{in}(v) = |E_i| / |V_i|$. The result follows from Theorem \ref{thm:quantstate}. An identical proof gives the second equality with $\deg_\text{out}(v) = |E_{i + 1}|/|V_i|$.
\end{proof}

\lemmaten*
\begin{proof}
	We show the result for the first inequality and the remaining proof is identical. An easy counting argument shows that the number of $I$'s in the $i$th column of $\mathcal{A}_{\mathfrak{f}_i}$ are either $|\mathcal{A}_{\mathfrak{f}_i}|$, $|\mathcal{A}_{\mathfrak{f}_i}|/p$, or $|\mathcal{A}_{\mathfrak{f}_i}|/p^2$. (If there are not all $I$'s, then pair any element with its inverse to create an element with an $I$. Likewise, pair any element without an $I$ with an element with an $I$. Do the same for $X$ and $Z$ combinations.) These are the elements which are also in $\mathcal{A}_{\mathfrak{f}_{i + 1}}$, so $|\mathcal{A}_{\mathfrak{f}_{i + 1}}| \geq |\mathcal{A}_{\mathfrak{f}_i}|/p^2$. The dimension change is therefore no more than two.
\end{proof}

\setcounter{theoremINNER}{11}
\thmtwelve*
\begin{proof}
	From Corollary \ref{cor:bruteforce} and Lemma \ref{lem:VEvs}, it suffices to show this for the edge $(\overline{0}, I, \overline{0}) \in E_i$, as all edge configurations are a shift of this one. Choose an arbitrary $v_i \in \mathcal{I}_i$ and $v_{i + 1} \in \mathcal{I}_{i + 1}$. By definition there exists Pauli labels $P_s$ and $P_t$ such that $(v_i, P_s, \overline{0})$, $(\overline{0}, P_t, v_{i + 1}) \in E_i$. Hence, $(v_i, P_s P_t, v_{i + 1}) \in E_i$. This works for any pair of vertices in $\mathcal{I}_i$ and $\mathcal{I}_{i + 1}$, proving the first statement. Now suppose, without loss of generality, there exists a $v^\prime_{i + 1} \in V_{i + 1} \backslash \mathcal{I}_{i + 1}$ connected to a $v_i \in \mathcal{I}_i$ via $(v_i, P^\prime, v^\prime_{i + 1})$ but is not part of the bipartite graph. As a subgroup, the inverse syndrome to $v_i$, $v_i^{-1}$, exists with some edge label $P_i$. Then $(v^{-1}_i, P_i, \overline{0}) (v_i, P^\prime, v^\prime_{i + 1}) = (\overline{0}, P_i P^\prime, v^\prime_{i + 1}) \in E_i$, a contradiction to the fact that $v^\prime_{i + 1} \not \in \mathcal{I}_{i + 1}$.
	
	To prove the last statement, pick a source/terminus pair $v_i \in V_i$ and $v_{i + 1} \in V_{i + 1}$ with parallel edges uniquely labeled by $\{P_1, \hdots, P_k\}$. We will show that every other edge must also have $k$ parallel edges; hence it suffices, without loss of generality, to assume that both syndromes are $\overline{0}$. Now choose any other edge $(v^\prime_s, P^\prime, v^\prime_t) \in E_i$ with $v^\prime_s$ and $v^\prime_t$ not both $\overline{0}$. By Lemma \ref{lem:VEvs}, $(\overline{0}, P^{-1}_1, \overline{0})$ exists and $(v^\prime_s, P^\prime, v^\prime_t) (\overline{0}, P^{-1}_1, \overline{0}) (\overline{0}, P_j, \overline{0})$ $ = (v^\prime_s, P^\prime P^{-1}_1 P_j, v^\prime_t)$ for some $1 < j \leq k$. Since $P^\prime P^{-1}_1 P_{j_1} = P^\prime P^{-1}_1 P_{j_2}$ implies $P_{j_1} = P_{j_2}$, a contradiction, there are $k - 1$ additional parallel edges between $v^\prime_s$ and $v^\prime_t$ of this form. Since the same idea works when $P^{-1}_1$ is replaced with any other of the $k$ labels, the number of parallel edges from $v^\prime_s$ to $v^\prime_t$ is greater than or equal to $k$. Running the argument backwards shows that for each edge between $v^\prime_s$ to $v^\prime_t$ there is a corresponding edge from $\overline{0}$ to $\overline{0}$. Since the number of these is $k$, we have that the number of parallel edges from $v^\prime_s$ to $v^\prime_t$ is exactly $k$.
\end{proof}

\corothirteen*
\begin{proof}
	Starting with $\dim S^\perp_{\mathfrak{f}_0} = n + k$ and repeatedly applying Lemma \ref{lem:deltadim} gives $\dim S^\perp_{\mathfrak{f}_i} \geq n + k - 2i$ and, likewise, $\dim S^\perp_{\mathfrak{p}_i} \geq - n + k + 2i$. Combining these two equations gives the desired result.
\end{proof}

\lemmafourteen*
\begin{proof}
	Since $\dim S^\perp_\mathfrak{p}$ and $\dim S^\perp_\mathfrak{f}$ are invariant, we know from Corollary \ref{cor:inoutdegs} that any trellises satisfying Theorem \ref{thm:quantstate} have the same vertex degrees and are hence isomorphic. It remains to specify the mapping. Let there be two minimal trellises for the code with sets $V, E$ and $V^\prime, E^\prime$, respectively. Fix a $v \in V_i$ and consider its past coset (see Figure \ref{fig:trelliscosets}). Pick a Pauli string in $S^\perp$ in this coset and determine its terminus $v^\prime \in V^\prime_i$. The map $f : V \to V^\prime$, $f(v) = v^\prime$ is clearly an isomorphism. It immediately follows that the edge isomorphism is given by $(s(e), P_i, t(e)) \mapsto (f(s(e)), P_i, f(t(e)))$.
\end{proof}

\setcounter{theoremINNER}{17}
\propeighteen*
\begin{proof}
	Suppose a set of Pauli strings is in TOF but does not have the left-right property. Then there exists two elements of the set, $\tilde{P}$ and $\tilde{P}^\prime$ such that $L(\tilde{P}) = L(\tilde{P}^\prime)$ or $R(\tilde{P}) = R(\tilde{P}^\prime)$. Without loss of generality, assume $L(\tilde{P}) = L(\tilde{P}^\prime)$ and $R(\tilde{P}) > R(\tilde{P}^\prime)$. One can always replace any power of $X$ or $Z$ with the first power by repeatedly applying it to itself since $p$ is prime and any power generates the cyclic group $\< p \>$. Assume this is done for $\tilde{P}$ and replace the power of $X$ or $Z$ in $\tilde{P}^\prime$ by $X^{p-1}$ or $Z^{p-1}$, $\tilde{P} \mapsto P$ and $\tilde{P}^\prime \mapsto P^\prime$. Then $PP^\prime$ has lower span length than $\tilde{P}$, a contradiction to the fact that the set is in TOF.
	
	Now suppose the set has the left-right property but is not in TOF. The only way to reduce the span length of the set is to increase a left index or decrease a right index. But this is impossible since there is only a single power of $X$ or $Z$ at these indices by the left-right property. Hence the set already has the lowest total span length.
\end{proof}

\setcounter{theoremINNER}{19}
\proptwenty*
\begin{proof}
	Let $\mu(v)$ be the value of the optimal (minimum weight) path at vertex $v$, initialize $\mu(\overline{0}) = 0$ at $V_0$, and denote by $\wt(e)$ the weight of edge $e$ (see Algorithm 1). We proceed by induction on the depth $i$. By construction, all paths from $V_0$ to $v \in V_1$ consist of a single edge. The optimal path $V_0$ to $v$ is given by the minimum weighted edge from $V_0$ to $v$. By Algorithm 1, $\mu(v) = \di \min_{\substack{e \in E_1\\ t(e) = v}} \{\mu(\overline{0}) + \wt(e)\} = \wt(e)$ for the edge $e \in E_1$ of minimum weight, which is correct. Now assume the Viterbi algorithm correctly computes the optimal path for all depths 1 to $i$ and let $v \in V_{i + 1}$. Then
	\begin{equation*}
		\mu(v) = \di \min_{\substack{e \in E_i\\ t(e) = v}} \{\mu(s(e)) + \wt(e)\}
			= \di \min_{\substack{e \in E_i\\ t(e) = v\\ \text{path}: V_0 \to s(e)}} \{\wt(\text{path}) + \wt(e)\}
			= \di \min_{\substack{e \in E_i\\ t(e) = v\\ \text{path}: V_0 \to t(e)}} \{\wt(\text{path})\}.
	\end{equation*}
	Every path from $V_0$ to $v$ is of this form.
\end{proof}

\thmtwentyone*
\begin{proof}
	It suffices to characterize the operations of line 14 in Algorithm \ref{alg:Viterbi}. For each vertex $v \in V_i$ for $1 \leq i \leq n$, $\deg_\text{in}(v)$ additions are preformed:
	\begin{equation*}
		\text{total additions} = \sum_{i = 1}^n \sum_{v \in V_i} \deg_\text{in}(v).
	\end{equation*}
	For simplicity we take a computational model where $x - 1$ minimums are computed for a list of size $x$. From this we then have
	\begin{equation*}
		\text{total minimums} = \sum_{i = 1}^n \sum_{v \in V_i} \deg_\text{in}(v) - \sum_{i = 1}^n \sum_{v \in V_i} 1.
	\end{equation*}
	The first term in this sum is equal to the total number of edges, $|E|$, and the second term is total number of vertices minus $V_0$, $|V| - 1$:
	\begin{align*}
		\text{total additions} &= |E|\\
		\text{total minimums} &= |E| - |V| + 1.
	\end{align*}
	Therefore,
	\begin{equation*}
		\text{total number of arithmetic operations} = \text{total additions} + \text{total minimums} = 2|E| - |V| + 1.
	\end{equation*}
	Since the trellis is connected, $|E| - |V| + 1 \geq 0$, so $2|E| - |V| + 1 \geq |E|$. The total number of arithmetic operations is hence upper bounded by $2|E|$ and lower bounded by $|E|$.
\end{proof}

\setcounter{theoremINNER}{22}
\lemmatwentythree*
\begin{proof}
	We proceed by induction. For $i = 1$, $\mathscr{P}_1 = \deg_\mathrm{out}(\overline{0})$, as desired. Now suppose the formula is true for some $i$. Then
	\begin{align*}
		\mathscr{P}_i &= \sum_{v \in V_{i - 1}} \mathscr{P}_{i - 1}(v) \deg_\mathrm{out}(v)\\
			&= \sum_{v \in V_{i - 1}} \mathscr{P}_{i - 1}(v) \left( \deg_\mathrm{out}(v) - 1 \right) + \sum_{v \in V_{i - 1}} \mathscr{P}_{i - 1}(v)\\
			&= \sum_{v \in V_{i - 1}} \mathscr{P}_{i - 1}(v) \left( \deg_\mathrm{out}(v) - 1 \right) + 1 + \sum_{j = 0}^{i - 2} \sum_{v \in V_j} \mathscr{P}_j(v) \left( \deg_\mathrm{out}(v) - 1 \right)\\
			&= 1 + \sum_{j = 0}^{i - 1} \sum_{v \in V_j} \mathscr{P}_j(v) \left( \deg_\mathrm{out}(v) - 1 \right).
	\end{align*}
\end{proof}

\proptwentyfour*
\begin{proof}
	We proceed by induction on $\kappa$. For $\kappa = 1$,
	\begin{gather*}
		R_1 < i \leq R_2,\\
		\dim S^\perp_{\mathfrak{p}_{R_0}} = 0,\\
		\dim S^\perp_{\mathfrak{p}_{i - 1}} = \dim S^\perp_{\mathfrak{p}_{R_1}},\\
		p_{1,1} = p^{\dim S^\perp_{\mathfrak{p}_{R_1}} - \dim S^\perp_{\mathfrak{p}_{R_1}}} = 1,\\
		p_1 = p^{\dim S^\perp_{\mathfrak{p}_{R_1}} - \dim S^\perp_{\mathfrak{p}_{R_0}}} = p^{\dim S^\perp_{\mathfrak{p}_{R_1}}}.
	\end{gather*}
	Then,
	\begin{align*}
		\sum_{j = 0}^{i - 1} p^{\dim S^\perp_{\mathfrak{p}_{i - 1}}} \left( \mathcal{E}^\prime_j - \mathcal{E}_j \right) =& \sum_{j = 0}^{R_1 - 1} p^{\dim S^\perp_{\mathfrak{p}_{i - 1}}} \left( \mathcal{E}^\prime_j - \mathcal{E}_j \right) + \sum_{j = R_1}^{i - 1} p^{\dim S^\perp_{\mathfrak{p}_{i - 1}}} \left( \mathcal{E}^\prime_j - \mathcal{E}_j \right) \\
		=& \, p^{\dim S^\perp_{\mathfrak{p}_{R_1}}} \sum_{j = 0}^{R_1-1} \left( \mathcal{E}^\prime_j - \mathcal{E}_j \right) + \sum_{j = R_1}^{i - 1} p^{\dim S^\perp_{\mathfrak{p}_j}} \left( \mathcal{E}^\prime_j - \mathcal{E}_j \right)\\
		=& \, p^{\dim S^\perp_{\mathfrak{p}_{R_1}} - \dim S^\perp_{\mathfrak{p}_{R_0}}} \sum_{j = 0}^{R_1 - 1} p^{\dim S^\perp_{\mathfrak{p}_{R_0}}}  \left( \mathcal{E}^\prime_j - \mathcal{E}_j \right) - \sum_{j = 0}^{R_1 - 1} p^{\dim S^\perp_{\mathfrak{p}_j}} \left( \mathcal{E}^\prime_j - \mathcal{E}_j \right)\\
		&+ \sum_{j = 0}^{R_1 - 1} p^{\dim S^\perp_{\mathfrak{p}_j}} \left( \mathcal{E}^\prime_j - \mathcal{E}_j \right) + \sum_{j = R_1}^{i - 1} p^{\dim S^\perp_{\mathfrak{p}_j}} \left( \mathcal{E}^\prime_j - \mathcal{E}_j \right)\\
		=& \, p_1 \sum_{j = 0}^{R_1 - 1} p^{\dim S^\perp_{\mathfrak{p}_j}} \left( \mathcal{E}^\prime_j - \mathcal{E}_j \right) - \sum_{j = 0}^{R_1 - 1} p^{\dim S^\perp_{\mathfrak{p}_{R_j}}} \left( \mathcal{E}^\prime_j - \mathcal{E}_j \right) + \sum_{j = 0}^{i - 1} p^{\dim S^\perp_{\mathfrak{p}_j}} (\mathcal{E}^\prime_j -\mathcal{E}_j)\\
		=& \, (p_1 - 1) \sum_{j = 0}^{R_1 - 1} p^{\dim S^\perp_{\mathfrak{p}_j}} \left( \mathcal{E}^\prime_j - \mathcal{E}_j \right) + \Delta_i\\
		=& \, p_{1, 1} (p_1 - 1) \Delta_{R_1} + \Delta_i. 
	\end{align*}

	Assume the result is true for $\kappa - 1$. We wish to show that
	\begin{equation*}
		\sum_{j = 0}^{i - 1} p^{\dim S^\perp_{\mathfrak{p}_{i - 1}}} \left( \mathcal{E}^\prime_j - \mathcal{E}_j \right) = \Delta_i + \sum_{a = 1}^\kappa p_{\kappa, a} (p_a - 1) \Delta_{R_a},
	\end{equation*}
	where
	\begin{gather*}
		R_\kappa < i \leq R_{\kappa + 1},\\
		p_{\kappa, a} = p^{\dim S^\perp_{\mathfrak{p}_{R_\kappa}} - \dim S^\perp_{\mathfrak{p}_{R_a}}},\\
		p_a = p^{\dim S^\perp_{\mathfrak{p}_{R_a}} - \dim S^\perp_{\mathfrak{p}_{R_{a - 1}}}}.
	\end{gather*}
	We have,
	\begin{align*}
		\sum_{j = 0}^{i - 1} p^{\dim S^\perp_{\mathfrak{p}_{i - 1}}} \left( \mathcal{E}^\prime_j - \mathcal{E}_j \right) =& \sum_{j = 0}^{R_\kappa - 1} p^{\dim S^\perp_{\mathfrak{p}_{i - 1}}} \left( \mathcal{E}^\prime_j - \mathcal{E}_j \right) + \sum_{j = R_\kappa}^{i - 1} p^{\dim S^\perp_{\mathfrak{p}_{i - 1}}} \left( \mathcal{E}^\prime_j - \mathcal{E}_j \right)\\
		=& \sum_{j = 0}^{R_\kappa - 1} p^{\dim S^\perp_{\mathfrak{p}_{R_\kappa}}} \left( \mathcal{E}^\prime_j - \mathcal{E}_j \right) + \sum_{j = R_\kappa}^{i - 1} p^{\dim S^\perp_{\mathfrak{p}_{R_\kappa}}} \left( \mathcal{E}^\prime_j - \mathcal{E}_j \right)\\
		=& \, p^{\dim S^\perp_{\mathfrak{p}_{R_\kappa}} - \dim S^\perp_{\mathfrak{p}_{R_{\kappa - 1}}}} \sum_{j = 0}^{R_\kappa - 1} p^{\dim S^\perp_{\mathfrak{p}_{R_{\kappa - 1}}}} \left( \mathcal{E}^\prime_j - \mathcal{E}_j \right) + \sum_{j = R_\kappa}^{i - 1} p^{\dim S^\perp_{\mathfrak{p}_j}} \left( \mathcal{E}^\prime_j - \mathcal{E}_j \right)\\
		=& \, p^{\dim S^\perp_{\mathfrak{p}_{R_\kappa}} - \dim S^\perp_{\mathfrak{p}_{R_{\kappa - 1}}}} \left( \Delta_{R_\kappa} + \sum_{a = 1}^{\kappa - 1} p_{\kappa - 1, a} (p_a - 1) \Delta_{R_a} \right) + \sum_{j = R_\kappa}^{i - 1} p^{\dim S^\perp_{\mathfrak{p}_j}} \left( \mathcal{E}^\prime_j - \mathcal{E}_j \right)\\
		=& \, p^{\dim S^\perp_{\mathfrak{p}_{R_\kappa}} - \dim S^\perp_{\mathfrak{p}_{R_{\kappa - 1}}}} \Delta_{R_\kappa} + p^{\dim S^\perp_{\mathfrak{p}_{R_\kappa}} - \dim S^\perp_{\mathfrak{p}_{R_{\kappa - 1}}}} \sum_{a = 1}^{\kappa - 1} p^{\dim S^\perp_{\mathfrak{p}_{R_{\kappa - 1}}} - \dim S^\perp_{\mathfrak{p}_{R_a}}} (p_a - 1) \Delta_{R_a}\\
		&+ \sum_{j = R_\kappa}^{i - 1} p^{\dim S^\perp_{\mathfrak{p}_j}} \left( \mathcal{E}^\prime_j - \mathcal{E}_j \right)\\
		=& \, p_\kappa \Delta_{R_k} - \Delta_{R_k} + \Delta_{R_k} + \sum_{a = 1}^{\kappa - 1} p^{\dim S^\perp_{\mathfrak{p}_{R_\kappa}} - \dim S^\perp_{\mathfrak{p}_{R_a}}} (p_a - 1) \Delta_{R_a} + \sum_{j = R_\kappa}^{i - 1} p^{\dim S^\perp_{\mathfrak{p}_j}} \left( \mathcal{E}^\prime_j - \mathcal{E}_j \right)\\
		=& \, (p_\kappa - 1) \Delta_{R_\kappa} + \sum_{j = 0}^{R_\kappa - 1} p^{\dim S^\perp_{\mathfrak{p}_j}} \left( \mathcal{E}^\prime_j - \mathcal{E}_j \right) + \sum_{a = 1}^{\kappa - 1} p_{\kappa, a} (p_a - 1) \Delta_{R_a} + \sum_{j = R_\kappa}^{i - 1} p^{\dim S^\perp_{\mathfrak{p}_j}} \left( \mathcal{E}^\prime_j - \mathcal{E}_j \right)\\
		=& \, p_{\kappa, \kappa} (p_\kappa - 1) \Delta_{R_\kappa} + \sum_{a = 1}^{\kappa - 1} p_{\kappa, a} (p_a - 1) \Delta_{R_a} + \sum_{j = 0}^{i - 1} p^{\dim S^\perp_{\mathfrak{p}_j}} \left( \mathcal{E}^\prime_j - \mathcal{E}_j \right)\\
		=& \, \sum_{a = 1}^\kappa p_{\kappa, a} (p_a - 1) \Delta_{R_a} + \Delta_i.
	\end{align*}
\end{proof}

\thmtwentyfive*
\begin{proof}
	Let unprimed quantities be with respect to the syndrome trellis and primed quantities for the other trellis. For the syndrome trellis,
	\begin{equation*}
		\mathscr{P}_i = 1 + \sum_{j = 0}^{i - 1} \sum_{v \in V_j} \mathscr{P}_j(v) \left( \deg_\mathrm{out}(v) - 1 \right) = 1 + \sum_{j = 0}^{i - 1} p^{\dim S^\perp_{\mathfrak{p}_j}} \sum_{v \in V_j} \left( \deg_\mathrm{out}(v) - 1 \right),
	\end{equation*}
	and for any other trellis,
	\begin{equation*}
		\mathscr{P}^\prime_i = 1 + \sum_{j = 0}^{i - 1} \sum_{v^\prime \in V^\prime_j} \mathscr{P}^\prime_j(v^\prime) \left( \deg_\mathrm{out}(v^\prime) - 1 \right) \leq 1 + \sum_{j = 0}^{i - 1} p^{\dim S^\perp_{\mathfrak{p}_j}} \sum_{v^\prime \in V^\prime_j} \left( \deg_\mathrm{out}(v^\prime) - 1 \right).
	\end{equation*}
	Since $\mathscr{P}^\prime_i \geq \mathscr{P}_i$, subtracting the two expressions gives $\Delta_i \geq 0$. Then
	\begin{align*}
		\left(|E^\prime| - |V^\prime| + 1 \right) - \left(|E| - |V| + 1 \right) &= \sum_{j = 0}^{n - 1} \left(\mathcal{E}^\prime_j - \mathcal{E}_j\right)\\
			&= \frac{1}{p^{\dim S^\perp_{\mathfrak{p}_{n - 1}}}} \sum_{j = 0}^{n - 1} p^{\dim S^\perp_{\mathfrak{p}_{n - 1}}} \left(\mathcal{E}^\prime_j - \mathcal{E}_j\right)\\
			&= \frac{1}{p^{\dim S^\perp_{\mathfrak{p}_{n - 1}}}} \left( \Delta_n + \sum_{a = 1}^{n - 1} p_{n - 1, a} (p_a - 1) \Delta_{R_a} \right)\\
			&\geq 0,
	\end{align*}
	where the last line follows because every term in the sum is positive.
\end{proof}

\setcounter{section}{4}
\setcounter{theoremINNER}{3}
\thmIVfour*
\begin{proof}
	We show minimality, which is the non-trivial part of the proof.\\
	
	We proceed by induction on the number of generators in the product. Let $\tilde{S}^\perp_i$ be the subset of generators of $S^\perp$ (in TOF) which have trivial partial syndromes with respect to $s_i$. Let $[L_1, R_1]$ and $[L_2, R_2]$ be the spans of $s_1$ and $s_2$, respectively, and construct trellises for each following Section \ref{sec:construction} with vertex profiles
	\begin{equation*}
		|V^1_i | = \begin{cases}
				1 & \text{ if $i \in [0, L_1)$}\\
				p^{\dim S^\perp - \dim \tilde{S}^\perp_1} & \text{ if $i \in [L_1, R_1)$}\\
				1 & \text{ if $i \in [R_1, n]$}
			\end{cases}
		\qquad
		|V^2_i | = \begin{cases}
				1 & \text{ if $i \in [0, L_2)$}\\
				p^{\dim S^\perp - \dim \tilde{S}^\perp_2} & \text{ if $i \in [L_2, R_2)$}\\
				1 & \text{ if $i \in [R_2, n]$}
			\end{cases}.
	\end{equation*}
	The vertex profile of the trellis product of the trellises for $s_1$ and $s_2$ are given by Lemma \ref{lem:CSStp} (i) to be
	\begin{align*}
		|V^{1, 2}_i | &= \begin{cases}
				1 & \text{ if $i \in [0, L_1)$}\\
				| V^1_i | & \text{ if $i \in [L_1, L_2)$}\\
				| V^1_i | | V^2_i | = p^{\left(\dim S^\perp - \dim \tilde{S}^\perp_1\right) + \left(\dim S^\perp - \dim \tilde{S}^\perp_2\right)} = p^{\dim S^\perp - \dim \tilde{S}^\perp_{1 \cap 2}} & \text{ if $i \in [L_2, R_1)$}\\
				| V^2_i | & \text{ if $i \in [R_1, R_2 )$}\\
				1 & \text{ if $i \in [R_2, n]$}
			\end{cases}\\
			&= p^{\dim S^\perp - \dim S^\perp_{\mathfrak{p}_i} - \dim S^\perp_{\mathfrak{f}_i} - \dim \tilde{S}^\perp_i},
	\end{align*}
	where we have used that for sets $A, B \subset C$, $(C \backslash A) \cup (C \backslash B) = C \backslash (A \cap B)$, $\tilde{S}^\perp_{1 \cap 2} \coloneqq \tilde{S}^\perp_1 \cap \tilde{S}^\perp_2$, the past and future are taken with respect to the generators $s_1$ and $s_2$, and the intersection $\tilde{S}^\perp_i$ depends on $i$. Having recovered the minimal form of Theorem \ref{thm:quantstate}, assume the theorem holds for the trellis product of $s_1, \hdots, s_{n - k - 1}$. It remains to show that the overlap $| V^{1, \hdots, n - k - 1}_i | | V^{n - k}_i |$ produces Equation \eqref{Vi}. In this case $\tilde{S}^\perp_{1 \cap \hdots \cap n - k} = \emptyset$ so
	\begin{equation*}
		| V^{1, \hdots, n - k - 1}_i | | V^{n - k}_i | = | V^{1, \hdots, n - k - 1}_i | \, p^{\dim S^\perp - \dim \tilde{S}^\perp_{n - k}} = p^{\dim S^\perp - \dim S^\perp_{\mathfrak{p}_i} - \dim S^\perp_{\mathfrak{f}_i}}.
	\end{equation*}
	The proof for the edges is similar.
\end{proof}

\propIVfive*
\begin{proof}
	Let $\xi_i = \log_p E_{X, i}$ and $\eta_i = \log_p E_{Z, i}$ such that $|E_X| = \sum_{i = 1}^n p^{\xi_i}$ and $|E_Z| = \sum_{i = 1}^n p^{\eta_i}$, and recall that $\xi_i, \eta_i \geq 1$. By Lemma \ref{lem:CSStp}, $|E| = \sum_{i = 1}^n p^{\xi_i + \eta_i}$. The left inequality is clear. For the right, we have
	\begin{align*}
		\left( |E_X| + |E_Z|\right)^2 &= |E_X|^2 + |E_Z|^2 + 2 |E_X| |E_Z|\\
			&= |E_X|^2 + |E_Z|^2 + 2 \left[ \left(\sum_{i = 1}^n p^{\xi_i + \eta_i}\right) + \sum_{i = 1}^n \sum_{\substack{j = 1\\ j \neq i}}^n p^{\xi_i + \eta_i}\right]\\
			&= |E_X|^2 + |E_Z|^2 + 2 \left[|E| + p^2  \sum_{i = 1}^n \sum_{\substack{j = 1\\ j \neq i}}^n p^{\xi_i + \eta_i - 2} \right]\\
			&\geq |E_X|^2 + |E_Z|^2 + 2 |E| + 2p^2 n(n - 1),
	\end{align*}
	where in the last line we have assumed that $\xi_i = \eta_i = 1$ for all $i$.
\end{proof}

\end{document}